\documentclass[acmsmall,usenames,dvipsnames]{acmart}
\usepackage{ifthen}

\newboolean{talk}
\setboolean{talk}{false}
\newboolean{paper}
\setboolean{paper}{false}

\newboolean{IEEEstyle}
\setboolean{IEEEstyle}{false}
\newboolean{lipicsstyle}
\setboolean{lipicsstyle}{false}
\newboolean{eptcsstyle}
\setboolean{eptcsstyle}{false}
\newboolean{sigplanstyle}
\setboolean{sigplanstyle}{false}
\newboolean{PACMPL}
\setboolean{PACMPL}{false}
\newboolean{acmartstyle}
\setboolean{acmartstyle}{false}

\newboolean{needstheorems}
\setboolean{needstheorems}{false}
\newboolean{withimages}
\setboolean{withimages}{false}
\newboolean{withproofs}
\setboolean{withproofs}{true}

\newboolean{french}
\setboolean{french}{false}

\setboolean{PACMPL}{true}
\ifthenelse{\boolean{talk}}{}{\usepackage[utf8]{inputenc}}
\ifthenelse{\boolean{PACMPL}}{}{
	\ifthenelse{\boolean{french}}{\usepackage[french]{babel}}{
	  \ifthenelse{\boolean{lipicsstyle}}{}{\usepackage[english]{babel}}}
	  }
\ifthenelse{\boolean{PACMPL}}{}{
	\usepackage{amssymb} }
\usepackage{amsmath}
\usepackage{graphicx}
\usepackage{bussproofs}
\usepackage{fancybox}
\usepackage{cmll}
\usepackage{stmaryrd}
\usepackage{multirow}
\ifthenelse{\boolean{IEEEstyle}}{
\usepackage[bookmarks={false}]{hyperref}}{
\usepackage{hyperref}}
\usepackage{proof}
\usepackage{hhline}
\usepackage{xspace}
\usepackage{booktabs}
\usepackage{subcaption}
\usepackage{wrapfig}
\usepackage{array}
\usepackage{arydshln}
\usepackage{commath}
\ifthenelse{\boolean{IEEEstyle}}{
\setboolean{needstheorems}{true}}{}

\ifthenelse{\boolean{eptcsstyle}}{
\setboolean{needstheorems}{true}}{}

\ifthenelse{\boolean{sigplanstyle}}{
\setboolean{needstheorems}{true}}{}

\ifthenelse{\boolean{needstheorems}}{

\usepackage{amsthm}

    \newtheorem{theorem}{Theorem}[section]
    \newtheorem{lemma}[theorem]{Lemma}
    \newtheorem{corollary}[theorem]{Corollary}
    \newtheorem{proposition}[theorem]{Proposition}
    \newtheorem{definition}[theorem]{Definition}
    \newtheorem{remark}[theorem]{Remark}
    \newtheorem{example}[theorem]{Example}
}{}

\newcommand{\myproof}[1]{
\ifthenelse{\boolean{withproofs}}{#1}{}}

\newcommand{\withproofs}[1]{
\ifthenelse{\boolean{withproofs}}{#1}{}}

\newcommand{\withoutproofs}[1]{
\ifthenelse{\boolean{withproofs}}{}{#1}}

\newcommand{\tm}{t}
\newcommand{\tmtwo}{u}
\newcommand{\tmthree}{r}
\newcommand{\tmfour}{w}
\newcommand{\tmfive}{s}

\newcommand{\var}{x}
\newcommand{\vartwo}{y}
\newcommand{\varthree}{z}

\newcommand{\Rew}[1]{\rightarrow_{#1}}

\renewcommand{\to}{\Rew{}}

\newcommand{\towh}{\Rew{wh}}

\newcommand{\symfont}[1]{\mathsf{#1}}

\newcommand{\varsym}{{\symfont{var}}}

\newcommand{\ctxholep}[1]{[#1]}
\newcommand{\ctxhole}{\ctxholep{\cdot}}

\newcommand{\ctx}{C}
\newcommand{\ctxtwo}{D}
\newcommand{\ctxthree}{E}

\newcommand{\ctxp}[1]{\ctx\ctxholep{#1}}

\newcommand{\hctx}{H}
\newcommand{\hctxtwo}{K}
\newcommand{\hctxthree}{G}
\newcommand{\hctxp}[1]{\hctx\ctxholep{#1}}
\newcommand{\hctxtwop}[1]{\hctxtwo\ctxholep{#1}}
\newcommand{\hctxthreep}[1]{\hctxthree\ctxholep{#1}}

\newcommand{\nbvctxtwo}[1]{\nbvctxtwo{#1}}

\newcommand{\defeq}{:=}
\newcommand{\eqdef}{=:}
\newcommand{\grameq}{::=}

\newcommand{\isub}[2]{\{#1/#2\}}
\newcommand{\esub}[2]{[#1/#2]}

\newcommand{\llbrace}{\{ \kern -0.27em \vert}
\newcommand{\rrbrace}{\vert \kern -0.27em \}}

\newcommand{\grammarpipe}{\mathrel{\big |}}

\renewcommand{\l}{\lambda}
\newcommand{\ie}{\textit{i.e.}\xspace}
\newcommand{\eg}{\textit{e.g.}\xspace}
\newcommand{\ih}{\textit{i.h.}\xspace}
\newcommand{\fv}[1]{\symfont{fv}(#1)}

\newcommand{\red}[1]{{\color{red} {#1}}}
\newcommand{\blue}[1]{{\color{blue} {#1}}}

\newcommand{\ignore}[1]{}

\newcommand{\myinput}[1]{\ifthenelse{\boolean{withimages}}{\input{#1}}{}}

\newcommand{\levy}{{L{\'e}vy}\xspace}

\newcommand{\nat}{\mathbb{N}}

\newcommand{\size}[1]{|#1|}
\newcommand{\sizeparam}[2]{|#1|_{#2}}

\newcommand{\clos}{c}

\newcommand{\env}{E}
\newcommand{\envtwo}{E'}
\newcommand{\envthree}{E''}

\newcommand{\stack}{S}

\ifthenelse{\boolean{PACMPL}}{\renewcommand{\state}{s}}{\newcommand{\state}{s}}
\newcommand{\statetwo}{{s'}}
\newcommand{\statethree}{s''}
\newcommand{\statefour}{q}

\newcommand{\tokam}{\Rew{\KAM}}

\newcounter{numberone}

\newenvironment{varenumerate}
{
\begin{list}{\arabic{numberone}.}
{
  \usecounter{numberone}
  \setlength{\itemsep}{0pt}
  \setlength{\topsep}{0pt}
  \setlength{\parsep}{0pt}
  \setlength{\partopsep}{0pt}
  \setlength{\leftmargin}{15pt}
  \setlength{\rightmargin}{0pt}
  \setlength{\itemindent}{0pt}
  \setlength{\labelsep}{5pt}
  \setlength{\labelwidth}{15pt}
}}
{
\end{list} 
}

\newcounter{numbertwo}

\newcommand{\trpos}{logged position\xspace}
\newcommand{\trposs}{logged positions\xspace}

\newcommand{\TrPoss}{Logged Positions\xspace}

\newcommand{\Log}{Log\xspace}

\newcommand{\JAMtoIAM}[1]{I(#1)}
\newcommand{\JAMtoIAMstate}[1]{I#1}
\newcommand{\indstate}[1]{#1^\circ}

\renewcommand{\ctxholep}[1]{\langle #1\rangle}
\newcommand{\ctxtwop}[1]{\ctxtwo\ctxholep{#1}}
\newcommand{\ctxthreep}[1]{\ctxthree\ctxholep{#1}}

\newcommand{\dom}[1]{dom(#1)}

\newcommand{\reflemma}[1]{Lemma~\ref{l:#1}}
\newcommand{\reflemmap}[2]{Lemma~\ref{l:#1}.\ref{p:#1-#2}}

\newcommand{\refprop}[1]{Prop.~\ref{prop:#1}}
\newcommand{\refsect}[1]{Sect.~\ref{sect:#1}}

\newcommand{\refeq}[1]{(\ref{eq:#1})}

\newcommand{\refthm}[1]{Theorem~\ref{thm:#1}}
\newcommand{\refthmp}[2]{Theorem~\ref{thm:#1}.\ref{p:#1-#2}}

\newcommand{\reffig}[1]{Fig.~\ref{fig:#1}}
\newcommand{\refapp}[1]{Appendix~\ref{sect:#1}}

\newcommand{\lctx}[1]{\ctx_{#1}}
\newcommand{\lctxp}[2]{\lctx{#1}\ctxholep{#2}}

\newcommand{\octx}[1]{O_{#1}}

\renewcommand{\esub}[2]{[#1{\shortleftarrow}#2]}
\renewcommand{\isub}[2]{\{#1{\shortleftarrow}#2\}}

\newcommand{\resm}{\psym}
\renewcommand{\resm}{\bullet}

\newcommand{\lpos}{p}
\renewcommand{\lpos}{l}
\newcommand{\lpostwo}{{\lpos'}}

\newcommand{\pos}{p}
\newcommand{\postwo}{\pos'}

\newcommand{\sizelpos}[1]{\size{#1}_{\lpos}}
\newcommand{\sizeclpos}[1]{\size{#1}_{\clpos}}

\newcommand{\upp}{\blue{\uparrow}}
\newcommand{\downp}{\red{\downarrow}}
\newcommand{\uppt}{\red{\uparrow}}
\newcommand{\downpt}{\blue{\downarrow}}
\newcommand{\tlog}{L}
\newcommand{\tlogtwo}{\tlog'}
\newcommand{\tlogthree}{\tlog''}
\newcommand{\tlogn}{L_n}
\newcommand{\tape}{T}
\newcommand{\tapetwo}{\tape'}
\newcommand{\tapethree}{\tape''}
\newcommand{\pol}{d}
\newcommand{\poltwo}{\pol'}

\newcommand{\run}{\pi}
\newcommand{\runtwo}{\sigma}
\newcommand{\runthree}{\rho}

\newcommand{\relf}{{\blacktriangleright}}

\newcommand{\nopolstate}[5]{(#1,#2,#4,#3,#5)}
\newcommand{\dstate}[4]{(\red{\underline{#1}},#2,#4,#3)}
\newcommand{\kamstate}[4]{(\red{#1},#2,#4,#3)}
\newcommand{\ustate}[4]{(#1,\blue{\underline{#2}},#4,#3)}

\newcommand{\dstatetab}[4]{\red{\underline{#1}} & #2 & #4 & #3 }
\newcommand{\kamstatetab}[4]{\red{#1} & #2 & #4 & #3 }

\newcommand{\ustatetab}[4]{#1 & \blue{\underline{#2}} & #4 & #3 }

\newcommand{\ndstatetab}[5]{\red{\underline{#1}} & #2 & #4 & #3 & #5}
\newcommand{\nustatetab}[5]{#1 & \blue{\underline{#2}} & #4 & #3 & #5}

\newcommand{\cons}{{\cdot}}

\newcommand{\mach}{\mathrm{M}}

\newcommand{\IAM}{IAM\xspace}
\newcommand{\JAM}{JAM\xspace}
\newcommand{\LIAM}{$\lambda$IAM\xspace}
\newcommand{\LJAM}{$\lambda$JAM\xspace}
\newcommand{\PAM}{PAM\xspace}
\newcommand{\LPAM}{$\lambda$PAM\xspace}
\newcommand{\HAM}{HAM\xspace}
\newcommand{\SIAM}{SIAM\xspace}
\newcommand{\KAM}{KAM\xspace}

\newcommand{\IAMold}{\mathrm{IAM}}
\newcommand{\JAMold}{\mathrm{JAM}}

\newcommand{\tomachhole}[1]{\rightarrow_{#1}}
\newcommand{\tomach}{\tomachhole{}}

\newcommand{\btsym}{\mathsf{bt}}
\newcommand{\tomachdotone}{\tomachhole{\resm 1}}
\newcommand{\tomachdottwo}{\tomachhole{\resm 2}}
\newcommand{\tomachvar}{\tomachhole{\varsym}}
\newcommand{\tomachbttwo}{\tomachhole{\btsym 2}}

\newcommand{\iamdap}{\tomachdotone}
\newcommand{\iamdlamone}{\tomachdottwo}
\newcommand{\iamdvar}{\tomachvar}
\newcommand{\iamdlamtwo}{\tomachbttwo}

\newcommand{\argsym}{\mathsf{arg}}
\newcommand{\tomachdotthree}{\tomachhole{\resm 3}}
\newcommand{\tomachdotfour}{\tomachhole{\resm 4}}
\newcommand{\tomacharg}{\tomachhole{\argsym}}
\newcommand{\tomachbtone}{\tomachhole{\btsym 1}}
\newcommand{\iamuapltwo}{\tomachdotthree}
\newcommand{\iamulam}{\tomachdotfour}
\newcommand{\iamuaplone}{\tomacharg}
\newcommand{\iamuapr}{\tomachbtone}

\newcommand{\tomachbttwodec}{\tomachhole{\btsym 2, \lpos}}
\newcommand{\tomachbtonedec}{\tomachhole{\btsym 1,\lpos}}
\newcommand{\tomachbtonedecp}[1]{\tomachhole{\btsym 1,#1}}
\newcommand{\tomachbttwodecp}[1]{\tomachhole{\btsym 2,#1}}

\newcommand{\tomachup}{\tomachhole{\upp}}
\newcommand{\tomachdown}{\tomachhole{\downp}}

\newcommand{\jumpsym}{\mathsf{jmp}}
\newcommand{\tomachjump}{\tomachhole{\jumpsym}}
\newcommand{\iamujump}{\tomachjump}

\newcommand{\appsym}{\mathsf{app}}
\newcommand{\abssym}{\mathsf{abs}}
\newcommand{\tomachapp}{\tomachhole{\appsym}}
\newcommand{\tomachabs}{\tomachhole{\abssym}}
\newcommand{\kamdap}{\tomachapp}
\newcommand{\kamdlam}{\tomachabs}

\newcommand{\varjsym}{\varsym \mathsf{J}}
\newcommand{\varksym}{\symfont{hop}/\varsym \mathsf{K}}
\newcommand{\tomachdotoneapp}{\tomachhole{\resm 1/\appsym}}
\newcommand{\tomachdottwoabs}{\tomachhole{\resm 2/\abssym}}
\newcommand{\tomachvarj}{\tomachhole{\varjsym}}
\newcommand{\tomachvark}{\tomachhole{\varksym}}
\newcommand{\jamdvar}{\tomachvarj}

\newcommand{\toliam}{\rightarrow_{\lambda\textsc{IAM}}}
\newcommand{\toljam}{\rightarrow_{\lambda\textsc{JAM}}}
\newcommand{\tosiam}{\rightarrow_{\textsc{SIAM}}}
\newcommand{\toham}{\rightarrow_{\textsc{HAM}}}
\newcommand{\tohamj}{\rightarrow_{\textsc{HAM}_{\symfont{J}}}}
\newcommand{\tohamk}{\rightarrow_{\textsc{HAM}_{\symfont{K}}}}
\newcommand{\tolpam}{\rightarrow_{\lambda\textsc{PAM}}}
\newcommand{\stempty}{\epsilon}

\newcommand{\sizee}[1]{\sizelpos{#1}}

\newcommand{\la}[1]{\lambda #1.}

\newcommand{\sem}[2]{\llbracket#1\rrbracket_{#2}}

\newcommand{\exstates}{\mathcal{E}}

\newcommand{\midd}{\; \; \mbox{\Large{$\mid$}}\;\;}

\newcommand{\timem}[1]{|#1|_\mathtt{t}}

\newcommand{\bigo}[1]{\mathcal{O}(#1)}

\newcommand{\sizevar}[1]{|#1|_{\varsym}}
\newcommand{\sizenotvar}[1]{|#1|_{\neg\varsym}}
\newcommand{\depthsym}{\mathsf{depth}}
\newcommand{\spdepthnopar}[1]{\depthsym #1}
\newcommand{\spdepth}[1]{\depthsym(#1)}

\newcommand{\sizeabs}[1]{|#1|_{\mathsf{abs}}}

\newcommand{\terminates}[2]{#1(#2){\Downarrow}}
\newcommand{\diverges}[2]{#1(#2){\Uparrow}}

\newcommand{\ccallbn}{Closed Call-by-Name\xspace}
\newcommand{\ccbn}{Closed CbN\xspace}
\newcommand{\history}{H}
\newcommand{\historytwo}{\history'}

\newcommand{\bisimtape}{\simeq_\tape}
\newcommand{\bisimlog}{\simeq_{\tlog\history}}
\newcommand{\bisimstate}{\simeq}

\newcommand{\pamstatenopol}[6]{(#1,#2,#3,#4,#5,#6)}
\newcommand{\pamstated}[5]{(\red{\underline{#1}},#2,#3,#4,#5)}
\newcommand{\pamstateu}[5]{(#1,\blue{\underline{#2}},#3,#4,#5)}

\newcommand{\pamstatedtab}[5]{\red{\underline{#1}} & #2 & #3 & #4 & #5}
\newcommand{\pamstateutab}[5]{#1 & \blue{\underline{#2}} & #3 & #4 & #5}

\newcommand{\hamstatenopol}[6]{(#1,#2,#3,#4,#5,#6)}
\newcommand{\hamstated}[5]{(\red{\underline{#1}},#2,#3,#4,#5)}
\newcommand{\hamstateu}[5]{(#1,\blue{\underline{#2}},#3,#4,#5)}

\newcommand{\hamstatedtab}[5]{\red{\underline{#1}} & #2 & #3 & #4 & #5}
\newcommand{\hamstateutab}[5]{#1 & \blue{\underline{#2}} & #3 & #4 & #5}

\newcommand{\Id}{\symfont{I}}

\newcommand{\sizeup}[1]{\sizeparam{#1}{\upp}}

\newcommand{\lclos}{\hat{\clos}}
\newcommand{\lclostwo}{\lclos'}

\newcommand{\clpos}{\hat{\lpos}}

\newcommand{\mset}[1]{[#1]}
\newcommand{\emmset}{[]}

\newcommand{\initty}{\star}
\newcommand{\ty}{\tau}
\newcommand{\tytwo}{\rho}

\renewcommand{\ty}{A}
\renewcommand{\tytwo}{\ty'}

\newcommand{\mty}{\mathcal{A}}
\newcommand{\mtytwo}{\mathcal{B}}
\renewcommand{\mty}{S}
\renewcommand{\mtytwo}{\mty'}
\newcommand{\arr}[2]{#1\rightarrow #2}
\newcommand{\tyctx}{\chi}
\newcommand{\tyctxp}[1]{\tyctx\ctxholep{#1}}
\renewcommand{\tyctx}{\mathbb{B}}
\renewcommand{\tyctxp}[1]{\tyctx\ctxholep{#1}}
\newcommand{\tyctxtwo}{\tyctx'}
\newcommand{\tyctxtwop}[1]{\tyctxtwo\ctxholep{#1}}
\newcommand{\tyctxthree}{\tyctx''}
\newcommand{\tyctxthreep}[1]{\tyctxthree\ctxholep{#1}}
\newcommand{\mtyctx}{\mathbb{S}}

\newcommand{\tye}{\Gamma}
\newcommand{\tyetwo}{\Delta}
\newcommand{\tjudg}[3]{#1\vdash #2:#3}
\newcommand{\tjudgi}[3]{#1\vdash_i #2:#3}
\newcommand{\ctjudg}[2]{\vdash #1:#2}
\newcommand{\wtjudgone}[1]{\vdash^{\textcolor{Violet}{#1}}}
\newcommand{\wtjudg}[4]{#1\stackrel{\textcolor{Violet}{#2}}{\vdash}#3:#4}

\newcommand{\wtjudgt}[4]{#1\stackrel{\textcolor{Violet}{#2}}{\vdash}#3 #4}
\renewcommand{\wtjudgt}[4]{#1\vdash^{\textcolor{Violet}{#2}} #3 #4}
\newcommand{\tyvar}{\textsc{T-Var}}
\newcommand{\tylamstar}{\tylam_\star}
\newcommand{\tylam}{\textsc{T-}\l}
\newcommand{\tyapp}{\textsc{T-@}}

\newcommand{\tyd}{\pi}
\newcommand{\tydtwo}{\tyd'}

\newcommand{\pof}{\;\triangleright}

\newcommand{\WeightTimeKAM}[1]{\mathbf{W}_{\KAM}(#1)}
\newcommand{\WeightTimeIAM}[1]{\mathbf{W}_{\text{\LIAM}}(#1)}

\newcommand{\occstar}[1]{\norm{#1}}
\newcommand{\tsys}[1]{\mathsf{T}_{#1}}

\newcommand{\DiPref}[1]{\mathsf{DiPref}(#1)}
\newcommand{\myldots}{...}

\newcommand{\ruleoc}{J}

\newcommand{\extr}[1]{\mathsf{ext}(#1)}

\newcommand{\extsym}{\symfont{ext}}
\newcommand{\bisimtypes}{\simeq_{\extsym}}
\newcommand{\etape}[1]{\tape_{\extsym}(#1)}
\newcommand{\etapeaux}[2]{\tape_{\extsym}^{#2}(#1)}
\newcommand{\etapeauxs}[1]{\tape_{\extsym}^{\state}(#1)}
\newcommand{\elog}[1]{\tlog_{\extsym}(#1)}
\newcommand{\elpos}[1]{\lpos_{\extsym}(#1)}
\newcommand{\estate}[1]{\state_{\extsym}(#1)}

\newcommand{\indp}[2]{(#1,#2)}

\newcommand{\relfrdx}{\relf_\mathsf{rdx}}
	\newcommand{\relfbody}{\relf_\mathsf{body}}
	\newcommand{\relfarg}{\relf_\mathsf{arg}}
	\newcommand{\relfext}{\relf_\mathsf{ext}}

\newcommand{\focus}{f}

\usepackage{tikz}
\makeatletter
\protected\def\tikz@nonactivecolon{\ifmmode\mathrel{\mathop\ordinarycolon}\else:\fi} 
\makeatother

\usetikzlibrary{matrix,arrows}
\usetikzlibrary{decorations.shapes}
\usetikzlibrary{decorations.text}
\usetikzlibrary{decorations.pathmorphing}
\usetikzlibrary{decorations.markings}
\usetikzlibrary{arrows.meta}

\usepackage{versions}
\excludeversion{LONG}
\excludeversion{SHORT}

\acmJournal{PACMPL}
\acmVolume{1}
\acmNumber{CONF} % CONF = POPL or ICFP or OOPSLA
\acmArticle{1}
\acmYear{2018}
\acmMonth{1}
\acmDOI{} % \acmDOI{10.1145/nnnnnnn.nnnnnnn}
\startPage{1}

\setcopyright{none}
\usepackage{booktabs}   %% For formal tables:
\usepackage{subcaption} %% For complex figures with subfigures/subcaptions
\begin{document}

%% Title information
\title{The (In)Efficiency of Interaction}        
%\subtitle{}
%\titlenote{with title note}             %% \titlenote is optional;
                                        %% can be repeated if necessary;
                                        %% contents suppressed with 'anonymous'

%\subtitlenote{with subtitle note}       %% \subtitlenote is optional;
                                        %% can be repeated if necessary;
                                        %% contents suppressed with 'anonymous'

%% Author information
%% Contents and number of authors suppressed with 'anonymous'.
%% Each author should be introduced by \author, followed by
%% \authornote (optional), \orcid (optional), \affiliation, and
%% \email.
%% An author may have multiple affiliations and/or emails; repeat the
%% appropriate command.
%% Many elements are not rendered, but should be provided for metadata
%% extraction tools.

%% Author with single affiliation.
\author{Beniamino Accattoli}
\orcid{nnnn-nnnn-nnnn-nnnn}             %% \orcid is optional
\affiliation{  
  \department{LIX}              %% \department is recommended
  \institution{Inria \& LIX, \'Ecole Polytechnique, UMR 7161}            %% \institution is required  
  \country{France}                    %% \country is recommended
}
\email{beniamino.accattoli@inria.fr}          %% \email is recommended

%% Author with single affiliation.
\author{Ugo Dal Lago}
\orcid{nnnn-nnnn-nnnn-nnnn}             %% \orcid is optional
\affiliation{
  \institution{Universit\`{a} di Bologna \& INRIA}            %% \institution 
  %%is required
  \country{Italy}                    %% \country is recommended
}
\email{ugo.dallago@unibo.it}          %% \email is recommended

%% Author with single affiliation.
\author{Gabriele Vanoni}
\orcid{nnnn-nnnn-nnnn-nnnn}             %% \orcid is optional
\affiliation{
  \institution{Universit\`{a} di Bologna \& INRIA}            %% \institution 
  %%is required
  \country{Italy}                    %% \country is recommended
}
\email{gabriele.vanoni2@unibo.it}          %% \email is recommended

%% Abstract
%% Note: \begin{abstract}...\end{abstract} environment must come
%% before \maketitle command
\begin{abstract}
% !TeX spellcheck = en_US
% !TEX root = main.tex
%%%%%%%%%%%%%%%%%%%%%%%%%%%%%%%%%%%%%%%%%%%%%%%%%%%%%%%%%
Evaluating higher-order functional programs through abstract machines
  inspired by the geometry of the interaction is known to
  induce \emph{space} efficiencies, the price being \emph{time}
  performances often poorer than those obtainable with traditional,
  environment-based, abstract machines. Although families of
  lambda-terms for which the former is exponentially less efficient
  than the latter do exist, it is currently unknown how \emph{general}
  this phenomenon is, and how far the inefficiencies can go, in
  the worst case.  We answer these questions formulating four different
  well-known abstract machines inside a common definitional framework,
  this way being able to give sharp results about the relative time
  efficiencies. We also prove that non-idempotent intersection type
  theories are able to precisely reflect the time performances of the
  interactive abstract machine, this way showing that its
  time-inefficiency ultimately descends
  from the presence of higher-order types.

\end{abstract}

%% 2012 ACM Computing Classification System (CSS) concepts
%% Generate at 'http://dl.acm.org/ccs/ccs.cfm'.
\begin{CCSXML}
<ccs2012>
<concept>}
<concept_id>10011007.10011006.10011008</concept_id>
<concept_desc>Software and its engineering~General programming languages</concept_desc>
<concept_significance>500</concept_significance>
</concept>
<concept>
<concept_id>10003456.10003457.10003521.10003525</concept_id>
<concept_desc>Social and professional topics~History of programming languages</concept_desc>
<concept_significance>300</concept_significance>
</concept>
</ccs2012>
\end{CCSXML}

%% End of generated code

\newcommand{\UDL}[1]{\textcolor{red}{#1}}

  %% \keywords are mandatory in final camera-ready submission

\maketitle

%%%%%%%%%%%%%%%%%%%%
% !TeX spellcheck = en_US
% !TEX root = main.tex
%%%%%%%%%%%%%%%
\section{Introduction}
\label{sect:intro}
Sometimes, simple objects such as natural numbers can generate
theories of marvelous richness, such as number theory. Something
similar happens with the $\lambda$-calculus, the universally accepted
model of purely functional programs. Its definition is simple: three
constructors and just one rewriting rule, $\beta$-reduction, based on
a natural notion of substitution. The theory of $\beta$-reduction,
however, is surprisingly rich, and still the object of research,
despite decades of deep investigations.

In the eighties, Barendregt's book~\cite{barendregt_lambda_1984} presented a 
stable operational and
denotational theory, \levy had already developed his sophisticated
optimality theory~\cite{levy_reductions_1978}, and languages such as Haskell 
were using tricky
sharing mechanisms in their implementations. In 1987, however, the
linear logic~\cite{girard_linear_1987} earthquake came together with a 
completely new viewpoint
on the $\lambda$-calculus, requiring to revisit the whole theory. For
our story, {two of its byproducts} are relevant, namely the geometry of
interaction~\cite{girard_geometry_1989} (shortened to GoI) and game 
semantics~\cite{DBLP:journals/iandc/HylandO00,DBLP:journals/iandc/AbramskyJM00}.

\paragraph{GoI and Game Semantics} At the time, GoI was a radically new interpretation of proofs, arising from 
connections between linear logic and functional analysis, and based on
an abstract notion of interactive execution {for proofs}. Game semantics
was introduced to solve the \emph{full abstraction problem for PCF}~ 
\cite{DBLP:journals/tcs/Milner77},
and along the years affirmed itself as the sharpest and most flexible
form of semantics for higher-order languages. Roughly, the models 
known
at the time were not able to capture fine computational behaviors---that is, 
they were not \emph{intensional} enough.
Strategies from game semantics, instead, allow to faithfully model
these behaviors of $\l$-terms: program composition is modeled as the
interplay between {the corresponding} strategies---a concrete form of 
interaction---having 
the flavor of executions in some sort of abstract machine.  In
fact, there are two styles of game semantics. One, AJM games, is due
to Abramsky, Jagadeesan, and
Malacaria~\cite{DBLP:journals/iandc/AbramskyJM00}, and {it} is directly
inspired by GoI. Another one, HO games, is due to Hyland and
Ong~\cite{DBLP:journals/iandc/HylandO00}, and models interaction in a
different, pointer-based, way.

\paragraph{Game Machines} The computational content of GoI was first explored by~\citet{danos_regnier_1995} and~\citet{mackie_geometry_1995}, who proposed a new form of
implementation schema called \emph{interaction abstract machine}
(shortened to \IAM). The \IAM works in a fundamentally
different way with respect to environment-based abstract machines,
which are the standard and time-efficient way of modeling the
implementations of functional languages. The link between game semantics and abstract machines was first explored by 
\citet{DBLP:conf/lics/DanosHR96}. They showed the \IAM to be the machinery behind AJM
games, and the new \emph{pointer abstract machine} (\PAM) the one
behind HO games. They also established a correspondence between the
two styles of games, providing an indirect relationship between the
\IAM and the \PAM. Finally, from a technical study of the
\IAM, \citet{DR99} introduced an optimized machine, the
\emph{jumping abstract machine} (\JAM), claiming it isomorphic to
the \PAM despite using very different data structures. In the
following, we refer collectively to the \IAM, the \JAM, and
the \PAM, as to \emph{game machines} (\emph{interaction machines} would be
ambiguous, because of the \IAM).

\paragraph{A Blind Spot} Despite the existence of a huge literature about GoI and game semantics, their related 
abstract machines remain---somewhat surprisingly---not well understood. Game 
machines are quite sophisticated and their presentations
are hard to grasp, sometimes even far from being formally defined. For
instance, the \PAM has always been presented informally, as an
algorithm described in natural language or pseudocode.  Additionally, the
relationship between these machines is not clear, especially at the
level of the relative performances. One of the aims of this paper is taking the first steps towards a
proper theory of the efficiency of game machines.

\paragraph{Space and Interactions} 
%What has been studied is the space efficiency of the \IAM.
It is well known 
that environment machines can be space inefficient, because the
environment (or closure) mechanism they rely on uses space
proportional to the number of $\beta$-steps, \ie\ the natural time
cost model of the $\l$-calculus. Using as much space as time is in
fact the worst one can do, from a space efficiency point of view.
The \IAM relies on a different mechanism, that---similarly to
offline Turing machines~\cite{DBLP:conf/esop/LagoS10}---sacrifices
time in order to be space-efficient. {This phenomenon was first pointed out by 
\citet{mackie_geometry_1995}, 
but it is the extensive work by Sch\"opp and
coauthors~\cite{bllspace,DBLP:conf/esop/LagoS10,Schopp14,Schopp15}
that showed that the \IAM allows} for capturing
sub-linear space computations\footnote{evaluating a $\lambda$-term
  without fully inspecting it is indeed possible if the term
  is accessed by way of pointers, in the spirit of so-called offline Turing machines (themselves
  an essential ingredient in the definition of complexity classes
  such as $\mathsf{LOGSPACE}$), and this is precisely the way the \IAM works;
  see~\cite{dal_lago_computation_2016} for a thorough discussion
  about sub-linear space computations in the $\lambda$-calculus.}, something impossible in environment machines.
Along the same lines, one can mention the \emph{Geometry of
  Synthesis}~\cite{ghica_geometry_2007,DBLP:journals/entcs/GhicaS10},
in which the geometry of interaction is seen as a compilation
scheme towards circuits, and computation space is \emph{finite},
and of paramount importance.

%Further
%recent works can be found below among related works. By the way, the
%inefficiency of environment machines has already been observed by
%Krishnaswami, Benton, and Hoffman, who proposed some techniques to
%alleviate it in the context of functional-reactive programming and
%based on linear types~\cite{DBLP:conf/popl/KrishnaswamiBH12}.

\paragraph{Time and Interactions} About time, instead, not much is known for game machines.
Since the early papers on the
\IAM~\cite{danos_regnier_1995,mackie_geometry_1995}, it is known
that it can be exponentially slower than environment machines. As an
example, the family of terms $\tm_n$ defined as $\tm_1 \defeq \Id$ and
$\tm_{n+1} \defeq \tm_n \Id$ (where $\Id$ is the identity combinator)
takes time exponential in $n$ to be evaluated by the \IAM, but
only linear time in any environment machine.  Therefore, game and environment machines are fundamentally
different devices, and  game ones---at least the \IAM---\emph{can} be
time-inefficient. There remain however various open questions. Is
  the inefficiency of the \IAM a \emph{general} phenomenon, that
  is, are \emph{all} $\lambda$-terms concerned? What about the other
  game machines? How bad can the aforementioned phenomenon be,
quantitatively speaking? On \emph{which} $\lambda$-terms does the
phenomenon show up? The main objective of this paper is to provide
answers to these questions.

\paragraph{Jumping is Dizzying} The time inefficiency of the \IAM is addressed by Danos and
Regnier and Mackie via an optimized machine, the
\JAM~\cite{mackie_geometry_1995,DR99}. In which relation the
\JAM is with other machines is unclear.  Danos and Regnier
present the \JAM as an optimization of the \IAM defined on
top of proof nets. Then, they claim (without proving it) that if one
considers the call-by-name translation of the $\lambda$-calculus into
proof nets, the \JAM is isomorphic to the \PAM, while they
prove that using the call-by-value translation one obtains the
\KAM. This is somehow puzzling, since the \KAM is a
call-by-\emph{name} machine.

\paragraph{Time, Environments, and Types} Another natural question comes from 
the study of {the 
relationship between intersection types and} environment
machines. {The non-idempotent variant of intersection types---here 
shortened to
\emph{multi types}---provides a type theoretic
understanding of time for environment machines, as shown by \citet{deCarvalho18}, since the execution time of 
environment machines can be extracted from multi types derivations}. It is natural to
wonder whether similar connections exist between game machines and
multi types, or some other form of type system. That would be
particularly useful as a way of comparing the time behaviour of a
given term when evaluated by distinct machines.

\paragraph{This Paper} The aim of this work is precisely giving the first sharp
results about the time (in)efficiency of the interaction mechanism at
work in game machines. We adopt the simplest possible setting, that
is, weak call-by-name evaluation on closed terms, called here
\emph{\ccbn}, and we provide four main contributions.

\paragraph{Contribution 1: a Formal Common Framework}
Inspired by a very recent reformulation of the \IAM on $\l$-terms
(rather than proof nets, as in the original papers) by
\citet{IamPPDPtoAppear}, called \LIAM, we provide new similar
presentations of the \JAM and the \PAM, called \LJAM and
\LPAM. These formulations are easily manageable and comparable,
enabling neat formal results about them---in particular, ours is the
first formal and manageable definition of the \PAM.

\paragraph{Contribution 2: Comparative Complexity}
We provide bisimulations between the \LIAM, the \LJAM, the
\LPAM, and additionally the \KAM, taken as the reference for
environment machines. This allows for a precise comparison of the time
behavior of the four machines:
\begin{enumerate}
\item \emph{Hierarchy}: we show that the \KAM is never slower
than the \LJAM which is never slower than the \LIAM, establishing
a sort of hierarchy. 
\item \emph{\LJAM is (slowly) reasonable}: a close inspection shows that the 
\LJAM is at most quadratically 
slower than the \KAM. Since the \KAM is a \emph{reasonable}\footnote{\emph{Reasonable} is a technical word meaning 
polynomially related to Turing 
machines. In our context, a machine for \ccbn is reasonable if the number of transition it takes on $\tm$ is polynomial 
in the number of weak head $\beta$-steps to reduce $\tm$ and in the size $\size\tm$ of $\tm$. For more details about 
reasonable cost models for the $\l$-calculus, see the overview in
\cite{DBLP:journals/entcs/Accattoli18}.} implementation scheme, we obtain than 
the \LJAM is reasonable as well.
\item \emph{\LJAM and \LPAM isomorphism}: we confirm Danos and Regnier's 
claim that the
\LJAM and the \LPAM are isomorphic (and have the same time
behavior), working out the elegant and yet far from trivial
isomorphism\footnote{In the note \cite{Danos04headlinear}, the authors claim that the \PAM "is faster than the \KAM 
in many cases" referring to private communication with Herbelin. Our results falsify the claim, as the 
\PAM---behaving 
as the \LJAM---is never faster and at most quadratically slower than the 
\KAM.}.
\end{enumerate}

\paragraph{Contribution 3: \LIAM Time and Multi Types} We show how to extract 
the length of the \LIAM run on a term 
$\tm$---that is, the time cost---from multi type derivations for
$\tm$. This
study complements de Carvalho's, showing that multi types can measure
also the time of the \IAM, not just the \KAM.  A key point is
that, by comparing how multi types measure game and environment
machines, we obtain a clear insight about the time gap between the two
approaches: the time usage of the \LIAM depends on the multi type
derivation \emph{and} the size of the involved multi types, while the time
usage of the \KAM depends only on the former. Therefore, the gap is
bigger on terms whose \KAM evaluation is much shorter than {the involved 
types}\footnote{Note that even the smallest {multi type 
derivation for} the
  inefficient \IAM family $\tm_n$ given above {uses types}
  exponential in $n$ {(inside the derivation)}.}.

\paragraph{Contribution 4: a Uniform Proof Technique}
The proofs of the three most challenging theorems---namely the bisimulations of
the \LIAM and the \LJAM, of the \LJAM and the \KAM, and the correctness 
of cost analyses via the type system---are all proved using the same novel
technique. To prove the correctness of the \LIAM,
\citet{IamPPDPtoAppear} study a new invariant, the
\emph{exhaustible (state) invariant}, expressing a form of coherence
of the data structures used by the \LIAM. Our theorems are all
proved by adapting the exhaustible invariant to each specific case,
providing concise technical developments and conceptual
unity. This point contrasts strikingly with the original papers on
  game machines, whose proof techniques are involved and
  indirect\footnote{The relationship between the \IAM and
      the \PAM in \cite{DBLP:conf/lics/DanosHR96} is not direct as it goes through both AJM games and HO games. 
Similarly, the
      relationship between the \IAM and the \JAM in
      \cite{DR99} is not self-contained, as it is based on the
      non-trivial equivalence between regular and legal paths proved
      in \cite{DBLP:conf/lics/AspertiDLR94}.}, and often informal.
  The exhaustible invariant---while certainly 
  technical---is simpler and provides direct arguments. It
seems to be the key tool to study game machines. As a slogan,
\emph{interacting is exhausting}.

\paragraph{Our Two Cents about Space}
At the end of the paper we also
briefly discuss space. Specifically, we provide examples showing that
the \LIAM can use more space than the \LJAM, despite the former
being considered space-efficient and the latter being as
space-inefficient as possible. This fact does not contradict the space
efficiency of the \LIAM, as it concerns only specific terms.
Our example however shows that space relationships
between game machines---if they can be established at all---are
subtler than the time ones, and of a less uniform nature.

\paragraph{Our Results, at a Distance} The body of the paper is quite technical 
and this is inevitable, because 
abstract machines---for as much as they can be abstract---are low-level tools. It is however easy to provide a 
high-level perspective.
Comparing with the original papers on game machines, our presentations
play the role of a \emph{Rosetta stone}, allowing to connect concepts
and decode many technical subtleties and invariants. Our exhaustible
invariant, additionally, removes the need to resort to game semantics or
legal paths when relating the machines.  Our complexity study suggests
that interaction as modeled by HO games (seen as the \LJAM and the
\LPAM) is a time-reasonable process, while as modeled by the GoI
and AJM games (seen as the \LIAM) is a time-inefficient
process\footnote{Whether the \IAM is reasonable is unclear. The
  aforementioned time inefficiency of the \IAM is not a proof that
  it is unreasonable.}. Our
multi type study, however, suggests that the gap between the two is
big only on terms whose type derivations are much smaller than {the involved} 
types.  Focusing on HO 
games, the
quadratic overhead of the \LJAM with respect to the \KAM shows
that interaction as modeled by HO games is time-reasonable but not
time-efficient. Summing up, \emph{interacting takes time, and is
  exhausting.}

\paragraph{Evaluating without Duplicating} Let us provide a conclusive insight. 
$\beta$-reducing a $\l$-term 
(potentially) duplicates arguments, whose different copies may be used differently, typically being applied to 
different 
further arguments. The machines in this paper never duplicate parts of the code\footnote{Note that not all machines 
avoid duplications: machines with a single global environment duplicate pieces of the code, see 
\cite{DBLP:conf/ppdp/AccattoliB17}. Perhaps surprisingly, performing duplications is not as costly as one may expect. 
Global environment machines are indeed time-efficient and faster than the game machines studied in this paper.}, but 
have nonetheless to distinguish 
different uses of a same piece of code during execution. Each one does it in a clever and sophisticated different 
way---multi types also fit this view, as they remove duplications altogether by taking all the needed copies in advance 
(see 
\refsect{typed-invariant}). This paper can then be seen as a systematic and thorough study of \emph{the art of 
evaluating without duplicating}.

\paragraph{Related Work}
Beyond the works already cited above, game machines are also studied
by
\citet{CurienHFlops,DBLP:journals/corr/abs-0706-2544} who consider different 
machines as directly obtained by game semantics, 
\citet{DBLP:conf/wollic/Mackie17} who derives a proof net based token machine 
for System \textsf{T}, \citet{DBLP:conf/tlca/Pinto01} who develops a parallel 
implementation of the \IAM, and \citet{DBLP:conf/gg/FernandezM02} who extend 
token machines to call-by-value. A
game machine accommodating the additive connectives of linear logic is
in \cite{DBLP:conf/tlca/Laurent01}. Space-efficiency of variants of
the IAM is addressed by
\citet{DBLP:conf/csl/Mazza15,DBLP:conf/icalp/MazzaT15}. Game machines
for languages beyond the $\lambda$-calculus, like $\lambda$-calculi
with algebraic effects, quantum $\lambda$-calculi or concurrent
calculi are in~\cite{hoshino_memoryful_2014,DBLP:conf/csl/LagoFHY14,
  DBLP:conf/lics/LagoFVY15,DBLP:conf/lics/LagoTY17}. A different kind
of machine inspired by the GoI is
in~\cite{danos_local_1993,DBLP:journals/tocl/PediciniQ07}. Interaction
and rewriting are mixed in recent work
by~\citet{DBLP:conf/csl/MuroyaG17,DBLP:journals/lmcs/MuroyaG19}. 
\citet{DBLP:conf/fossacs/Clairambault11,DBLP:journals/corr/Clairambault15}
uses an abstraction of the \PAM to bound evaluation lengths, and
similar studies are also done by
\citet{DBLP:journals/jsyml/Aschieri17}.

The space inefficiency of environment machines has already been observed by
Krishnaswami, Benton, and Hoffman, who proposed some techniques to
alleviate it in the context of functional-reactive programming and
based on linear types~\cite{DBLP:conf/popl/KrishnaswamiBH12}.

The time efficiency of environment machines has been recently studied
in depth. Before 2014, the topic had been mostly neglected---the only
two counterexamples being
\citet{DBLP:conf/fpca/BlellochG95,DBLP:conf/birthday/SandsGM02}. Since
2014---motivated by advances on time cost models for the $\l$-calculus
by \citet{accattoli_leftmost-outermost_2016}---the topic has actively
been studied
\cite{DBLP:conf/icfp/AccattoliBM14,DBLP:journals/scp/AccattoliG19,DBLP:conf/ppdp/AccattoliB17,
DBLP:conf/ppdp/AccattoliCGC19}.

Intersection types are a standard tool to study
$\lambda$-calculi---see standard references such as
\cite{DBLP:journals/aml/CoppoD78,DBLP:journals/ndjfl/CoppoD80,Pottinger80,krivine1993lambda}.
Non-idempotent intersection types, \ie multi types, were first
considered by \citet{DBLP:conf/tacs/Gardner94}, and then by
\citet{DBLP:journals/logcom/Kfoury00,DBLP:conf/icfp/NeergaardM04,Carvalho07,deCarvalho18}---a
survey is~\cite{BKV17}. De Carvalho's use of multi types to give
bounds to evaluation lengths has also been used in
\cite{DBLP:journals/tcs/CarvalhoPF11,Bernadet-Lengrand2013,DBLP:journals/jfp/AccattoliGK20,DBLP:conf/aplas/AccattoliG18,
DBLP:conf/esop/AccattoliGL19,DBLP:conf/lics/KesnerV20}.

\section{Preliminaries: \ccallbn, and Abstract Machines}

	Let $\mathcal{V}$ be a countable set 
	of variables. 
	Terms of the \emph{$\lambda$-calculus} $\Lambda$ are defined as follows.
	\begin{center}$\begin{array}{rrcl}
	\textsc{$\l$-terms} & \tm,\tmtwo,\tmthree & \grameq & x\in\mathcal{V}\midd 
	\lambda x.\tm\midd 
	\tm\tmtwo.
	\end{array}$
\end{center}
	\emph{Free} and \emph{bound variables} are defined as 
	usual: $\la\var\tm$ binds $\var$ in $\tm$. A term is \emph{closed} when 
	there are no free occurrences of variables in it.
	%\emph{Closed} terms are terms without free 	variables. 
	Terms are considered modulo $\alpha$-equivalence, and capture-avoiding (meta-level) substitution of 
	all the free occurrences of $\var$ for $\tmtwo$ in $\tm$ is noted 
	$\tm\isub\var\tmtwo$. Contexts are just $\lambda$-terms containing exactly 
	one occurrence of a special symbol, the hole $\ctxhole$, intuitively standing for a removed subterm. Here we adopt 
	leveled contexts, whose index, \ie\ the level, stands for the number of 
	arguments (\ie\ the 
	number of !-boxes in linear logic terminology) the hole lies in.
	\begin{center}$
	\begin{array}{rclrrrcl}
	\multicolumn{8}{c}{\textsc{Leveled contexts}}
	\\
	\ctx_0		& \grameq &	\ctxhole \midd \la\var\ctx_0 \midd \ctx_0\tm;
	&&&
	\ctx_{n+1}	& \grameq &	\ctx_{n+1}\tm\midd\la\var\ctx_{n+1}\midd\tm\ctx_{n}.
	\end{array}
$\end{center}
	We simply write $\ctx$ for a context whenever the level is not relevant. 
	The operation replacing the hole $\ctxhole$ with a term $\tm$ 
in a context $\ctx$ is noted $\ctxp\tm$ and called \emph{plugging}.
	
	The operational semantics that we adopt here is weak head evaluation $\towh$, defined as follows:
	\begin{center}$
(\la\vartwo \tm) \tmtwo \tmthree_1 \ldots 
\tmthree_h \ \ \towh \ \ \tm \isub\vartwo \tmtwo 
\tmthree_1 \ldots \tmthree_h.
$\end{center}
We further restrict the setting by considering only closed terms, and refer to 
our framework as \emph{\ccallbn} (shortened to \ccbn). Basic well known facts 
are that in \ccbn the normal forms are precisely the abstractions and that 
$\towh$ is 
deterministic.

\paragraph{Abstract Machines Glossary}  In this paper, an \emph{abstract 
machine} $\mach = (\state, \tomach)$ is a transition system $\tomach$ over a 
set of states, noted $\state$. 
The machines considered in this paper move over the code without ever changing it. A \emph{position} in a term 
$\tm$ is represented as a pair $(\tmtwo,\ctx)$ of a sub-term $\tmtwo$ and a context $\ctx$ such that $\ctxp\tmtwo=\tm$. 
The shape of states depends on the specific machine, but they always include a position $(\tmtwo,\ctx)$ plus some data 
structures. 

A state is \emph{initial}, and noted $\state_\tm$, if it is positioned on $(\tm,\ctxhole)$, $\tm$ is closed, and all the 
data structures are empty. We may write $\state_\tm^{\mach}$ to stress the machine, and $\tm$ is always implicitly 
considered closed, without further mention. A state is \emph{final} if no transitions apply.
 
 A \emph{run} $\run: \state \tomach^*\statetwo$ is a possibly empty sequence of transitions, whose length is noted 
$\size\run$. If $a$ and $b$ are transitions labels (that is, $\tomachhole{a}\subseteq \tomach$ and 
$\tomachhole{b}\subseteq \tomach$) then $\tomachhole{a,b} \defeq \tomachhole{a}\cup \tomachhole{b}$, $\sizeparam\run a$ 
is the number of $a$ transitions in $\run$, and $\sizeparam\run{\neg a}$ is the size of transitions in $\run$ that are 
not $\tomachhole{a}$. An \emph{initial run} is a run from an initial state $\state_\tm$, and it is also called \emph{a run from $\tm$}. A state $\state$ is \emph{reachable} if it 
is the target state of an initial run. A \emph{complete run} is an initial run ending on a final state. Given a 
machine 
$\mach$, we write $\terminates\mach\tm$ if  $\mach$ 
reaches a final state starting from $\state_\tm^{\mach}$, and $\diverges\mach\tm$ otherwise. We say that \emph{$\mach$ 
implements \ccbn} when $\terminates\mach\tm$ if and only if $\towh$ terminates on $\tm$, for every closed term $\tm$.

%%%%%%%%%%%%%%%%%%%%
% !TeX spellcheck = en_US
% !TEX root = main.tex
%%%%%%%%%%%%%%%%%%%%%%%%%%%%%%%%%%%%%%%%%%%%%%%%%%%%%%%%%%%%%%%%%%%%%

\section{The Interaction Abstract Machine, Revisited}
\label{sec:IJK}
%%%%%%%%%%%%%%%%%%%%%%%%%%%%%%%%%%%%%%%%%%%%%%%%%%%%%%%%%%%%%%%%%%%%%
%\UDL{This section should be devoted to introducing the IAM, the JAM (as a variation
%  of the IAM), and the KAM, observing that the latter can be seen as a machine implementing
%  some of the steps of the JAM. For each of them, I would explain informally underlying
%  idea, and I would give a proof of soundness and adequacy in this section. I would also
%  cite (referring to the long version, that the JAM is isomorphic to the PAM.}

In this section we provide an overview of the Interaction Abstract Machine (IAM). We adopt the $\l$-calculus 
presentation of the $\IAMold$, rather called \LIAM and recently developed by 
\citet{IamPPDPtoAppear}---we refer to 
their work for an in-depth study of the \LIAM. 
The literature usually studies the ($\l$-)$\IAMold$ with respect to head evaluation of potentially open terms. Here 
we only deal with \ccbn, that is closer to the practice of functional 
programming and also the setting underlying the \KAM. We first define the 
machine and then provide explanations and 
examples. Keep in mind that the \LIAM is an unusual machine, and that finding 
it hard to grasp is 
normal---probably, the next sections about the \LJAM and the \KAM provide 
clarifying insights.
%Here, we just define it.
 \begin{figure*}[t]
	\input{machines/LIAM}
	\vspace{-8pt}
	\caption{Data structures and transitions of the $\lambda$ Interaction 
	Abstract Machine (\LIAM).}
	\label{fig:iam}
\end{figure*}

\paragraph{\LIAM States}
Intuitively, the behaviour of the \LIAM can be seen as 
that of a token that travels around the syntax tree of the program under 
evaluation---the transitions and all the data structures are defined in \reffig{iam}.
The \LIAM travels on a $\l$-term $\tm$ carrying data 
structures---representing the token---storing 
information about the computation and determining the next transition to apply. 
A key point is that navigation is done locally, moving only between adjacent 
positions\footnote{Note that also the transition from the 
  variable occurrence to the binder in $\tomachvar$ and $\tomachbttwo$ are local if $\l$-terms are represented by implementing occurrences as pointers to their binders, as in the proof net 
  representation of $\l$-terms, upon which some concrete implementation schemes 
  are based, see \cite{DBLP:conf/ppdp/AccattoliB17}.}. The \LIAM has also a 
  \emph{direction} of navigation that is either 
$\downp$ or $\upp$ (pronounced \emph{down} and \emph{up}).
The token is given by two stacks, 
called \emph{log} and \emph{tape}, whose main components are \emph{\trposs}. 
Roughly, a log is a trace of the relevant 
positions in the history of a computation, and a logged 
position is a position plus a log, meant to trace the history that led to that position. Logs and logged positions are 
defined by mutual induction\footnote{This is similar to the KAM, where closures and environments are defined by mutual 
induction, but logs and logged positions play a different role. Moreover, there 
also is a constraint about the length.}. 
%The set of \trposs is $\lposet$, and 
We use $\cdot$ also to concatenate logs, 
writing, \eg, $\tlog_n\cdot\tlog$, using $\tlog$ for a log of unspecified 
length. The \emph{tape} $\tape$ is a list of logged positions plus occurrences of the special symbol 
$\resm$, needed to record the crossing of abstractions 
and applications. A \emph{state} of the machine is given by a position and a 
token (that is, a log $\tlog$ and a tape $\tape$), together with a \emph{direction}.
Initial states have the form $\state_{\tm}\defeq \dstate{\tm}{\ctxhole}{\epsilon}{\epsilon}$.
Directions are often omitted and represented via colors and underlining: 
$\downp$ is represented by a
\red{red} and underlined code term, $\upp$ by a \blue{blue} and underlined code context.

\paragraph{Transitions.} Intuitively, the machine
evaluates the term $\tm$ until the head abstraction of the head normal form is 
found (more explanations below).
The transitions of the \LIAM are in 
\reffig{iam}. Their union is noted $\toliam$. The idea is 
that $\downp$-states 
$\dstate\tm\ctx\tape\tlog$ are queries about the head variable of (the head 
normal form of) $\tm$ and 
$\upp$-states $\ustate\tm\ctx\tape\tlog$ are queries about the argument of an 
abstraction. 

Next, we explain how the transitions realize three entangled mechanisms. Let us 
anticipate that the 
\LJAM shall be obtained by short-circuiting the third mechanism, backtracking, 
and
the KAM by the further removal of the second one, 
that shall also require to 
modify the first one.

	\paragraph{Mechanism 1: Search Up to $\beta$-Redexes} 
	\begingroup
	\setlength{\intextsep}{0pt}
	\begin{wrapfigure}{r}{0pt}${\footnotesize
			\begin{array}{l@{\hspace{0.1cm}}l|c|c|c|l}
			&\mathsf{Sub}\mbox{-}\mathsf{term} & \mathsf{Context} & 
			\mathsf{\Log} & 
			\mathsf{Tape} & \mathsf{Dir}
			\\
			\cline{1-6}
			&\ndstatetab{(\la\vartwo{\la\var{\var\vartwo}})\mathsf{II}} 
			{\ctxhole} 
			{\epsilon} {\epsilon} 
			\downp\\
			\iamdap&\ndstatetab{(\la\vartwo{\la\var{\var\vartwo}})\mathsf{I}} 
			{\ctxhole\mathsf{I}} {\resm} {\epsilon} \downp\\
			\iamdap&\ndstatetab{\la\vartwo{\la\var{\var\vartwo}}} 
			{\ctxhole\mathsf{II}} 
			{\resm\cdot\resm} {\epsilon} \downp\\
			\iamdlamone&\ndstatetab{\la\var{\var\vartwo}} 
			{(\la\vartwo\ctxhole)\mathsf{II}} 
			{\resm} 
			{\epsilon} \downp\\
			\iamdlamone&\ndstatetab{\var\vartwo} 
			{(\la\vartwo{\la\var\ctxhole})\mathsf{II}} 
			{\epsilon} {\epsilon} 
			\downp
			\end{array}}
		$
	\end{wrapfigure}
Note that $\iamdap$ skips the 
 argument and adds a $\resm$ on the tape. The idea is that $\resm$ keeps track 
 that an argument has been encountered---its identity is however forgotten. 
 Then $\iamdlamone$ does the dual job: it skips an abstraction when the tape 
 carries a $\resm$, that is, the trace of a previously encountered 
 argument. Note that, when the \LIAM moves through a $\beta$-redex with the 
 two steps one after the other, the token is 
left unchanged. This mechanism thus realizes search \emph{up to $\beta$-redexes}, 
 that is, without ever recording them. Note that 
 $\iamuapltwo$ and $\iamulam$ realize the same during the $\upp$ phase. 
 Let us 
 illustrate this mechanism with an example (on the right): the first steps of 
 evaluation on 
 the term $(\la\vartwo{\la\var{\var\vartwo}})\mathsf{II}$, where $\mathsf{I}$ 
 is the identity combinator. One can notice that the \LIAM traverses two 
 $\beta$-redexes without altering the token, that is empty both at the 
 beginning and at the end.

\paragraph{Mechanism 2: Finding Variables and Arguments} As a first approximation, navigating in direction $\downp$ 
corresponds to looking for the head variable of the term code, while navigating with direction $\upp$ corresponds to 
looking for the sub-term to replace the previously found head variable, what we 
call \emph{the argument}.\endgroup More 
precisely, when the head variable $\var$ of the active subterm is found, transition $\iamdvar$  switches direction from 
$\downp$ to $\upp$, and the machine starts looking for potential substitutions 
for $\var$. The \LIAM then moves to 
the position of the binder $\lambda\var$ of $\var$, and starts exploring the context $\ctx$, looking for the first 
argument up to $\beta$-redexes. The relative position of $\var$ w.r.t. its binder is recorded in a new \trpos that is 
added to the tape. Since the machine moves out of a context of level $n$, namely $\ctxtwo_n$, the \trpos contains the 
first $n$ \trposs of the log. Roughly, this is an encoding of the run that led 
from the level of 
$\la\var\ctxtwo_n\ctxholep\var$ to the occurrence of $\var$ at hand, in case 
the machine would later need to backtrack.
 
 When the argument $\tm$ for the abstraction binding the variable $\var$ in 
 $\lpos$ is found, transition $\iamuaplone$ switches direction from $\upp$ to 
 $\downp$, making the machine looking for 
the head variable of $\tm$. 
 Note that moving to $\tm$, the level 
 increases, and that the \trpos $\lpos$ is moved from the tape to the 
log. 
The idea is that $\lpos$ is now a completed argument query, 
 and it becomes part of the history of
  how the machine got to the current 
 position, to be potentially used for backtracking.
 \begingroup
 \setlength{\intextsep}{0pt}
 \begin{wrapfigure}{r}{0pt}${\footnotesize
 		\begin{array}{l@{\hspace{0.1cm}}l|c|c|c|c}
 		&\mathsf{Sub}\mbox{-}\mathsf{term} & \mathsf{Context} & \mathsf{\Log} & 
 		\mathsf{Tape} & \mathsf{Dir}
 		\\
 		\cline{1-6}
 		&\ndstatetab{\var\vartwo} {(\la\vartwo{\la\var\ctxhole})\mathsf{II}} 
 		{\epsilon} {\epsilon} \downp\\
 		\iamdap&\ndstatetab{\var} 
 		{(\la\vartwo{\la\var{\ctxhole\vartwo}})\mathsf{II}} 
 		{\resm} {\epsilon} \downp\\
 		\iamdvar&\nustatetab{\la\var{\var\vartwo}} 
 		{(\la\vartwo\ctxhole)\mathsf{II}} 
 		{(\var,\la\var{\ctxhole\vartwo},\stempty)\cons\resm} 
 		{\epsilon} \upp\\
 		\iamulam&\nustatetab{\la\vartwo{\la\var{\var\vartwo}}} 
 		{\ctxhole\mathsf{II}} 
 		{\resm\cons(\var,\la\var{\ctxhole\vartwo},\stempty)\cons\resm} 
 		{\epsilon} \upp\\
 		\iamuapltwo&\nustatetab{(\la\vartwo{\la\var{\var\vartwo}})\mathsf{I}} 
 		{\ctxhole\mathsf{I}} 
 		{(\var,\la\var{\ctxhole\vartwo},\stempty)\cons\resm} 
 		{\epsilon} \upp\\
 		\iamuaplone&\ndstatetab{\mathsf{I}} 
 		{(\la\vartwo{\la\var{\var\vartwo}})\mathsf{I}\ctxhole} 
 		{\resm} 
 		{(\var,\la\var{\ctxhole\vartwo},\stempty)} \downp\\
 		\end{array}}
 	$
 \end{wrapfigure}
  We continue the example of 
 the previous point: the machine finds the head variable $\var$ and looks for 
 its argument in $\upp$ mode. When it has been found, the direction turns 
 $\downp$ again.

\paragraph{Mechanism 3: Backtracking} It is started by transition $\iamuapr$. 
 The idea is that the search for an argument of the $\upp$-phase has to 
 temporarily stop, because there are no arguments left at the current level. 
 The search of the argument then has to be done among the arguments of the 
 variable occurrence that triggered the search, encoded in $\lpos$. Then the 
 machine enters into backtracking mode, which is denoted by a $\downp$-phase 
 with a \trpos on the tape, to reach the position in $\lpos$. 
 Backtracking is over when $\iamdlamtwo$ is fired.
 \endgroup
 
 The $\downp$-phase and the \trpos
 on the tape mean that the \LIAM is backtracking. In fact, in this 
 configuration the machine is 
 not looking for the head variable of the current subterm $\la\var\tm$, it is 
 rather going back to the variable position in the tape, to find its 
 argument. This is realized by moving to the position in the tape and 
 changing direction. Moreover, the log $\tlog_n$
 encapsulated in the \trpos is put back on the global log. An invariant 
 shall guarantee that the \trpos on the tape always contains a position 
 relative to the active abstraction. In our running example, a backtracking 
 phase, noted with a \textsf{BT} label starts when the \IAM looks for the 
 argument of $\varthree$. Since $\la\varthree\varthree$ has been virtually 
 substituted for $\var$, its argument its actual{}ly $\vartwo$. Backtracking is 
 needed to recover the variable a term was virtually substituted for. 
  \begin{center}${\footnotesize
  	\begin{array}{c@{\hspace{0.2cm}}l@{\hspace{0.2cm}}l|c|c|c|c}
  	&&\mathsf{Sub}\mbox{-}\mathsf{term} & \mathsf{Context} & \mathsf{\Log} & 
  	\mathsf{Tape} & \mathsf{Dir}
  	\\
  	\cline{1-7}
  	&&\ndstatetab{\la\varthree\varthree} 
  	{(\la\vartwo{\la\var{\var\vartwo}})\mathsf{I}\ctxhole} 
  	{\resm} 
  	{(\var,\la\var{\ctxhole\vartwo},\stempty)} \downp\\
  	&\iamdlamone&\ndstatetab{\varthree} 
  	{(\la\vartwo{\la\var{\var\vartwo}})\mathsf{I}(\la\varthree\ctxhole)} 
  	{\stempty} 
  	{(\var,\la\var{\ctxhole\vartwo},\stempty)} \downp\\
  	&\iamdvar&\nustatetab{\la\varthree\varthree} 
  	{(\la\vartwo{\la\var{\var\vartwo}})\mathsf{I}\ctxhole} 
  	{(\varthree,\la\varthree\ctxhole,\stempty)} 
  	{(\var,\la\var{\ctxhole\vartwo},\stempty)} \upp\\
  	\mathsf{BT}&\iamuapr&\ndstatetab{(\la\vartwo{\la\var{\var\vartwo}})\mathsf{I}}
  	{\ctxhole\mathsf{I}} {(\var,\la\var{\ctxhole\vartwo},\stempty)\cons 
  		(\varthree,\la\varthree\ctxhole,\stempty)} {\epsilon} \downp\\
  	\mathsf{BT}&\iamdap&\ndstatetab{\la\vartwo{\la\var{\var\vartwo}}} 
  	{\ctxhole\mathsf{II}} 
  	{\resm\cons(\var,\la\var{\ctxhole\vartwo},\stempty)\cons 
  		(\varthree,\la\varthree\ctxhole,\stempty)} {\epsilon} \downp\\
  	\mathsf{BT}&\iamdlamone&\ndstatetab{\la\var{\var\vartwo}} 
  	{(\la\vartwo\ctxhole)\mathsf{II}} 
  	{(\var,\la\var{\ctxhole\vartwo},\stempty)\cons 
  		(\varthree,\la\varthree\ctxhole,\stempty)} {\epsilon} \downp\\
  	&\iamdlamtwo&\nustatetab{\var} 
  	{(\la\vartwo\la\var{\ctxhole\vartwo})\mathsf{II}} 
  	{(\varthree,\la\varthree\ctxhole,\stempty)} {\epsilon} \upp\\
  	\end{array}}$\end{center}
For the sake of completeness, we conclude the example, which runs until
 the head abstraction of the weak head normal form of the term under 
evaluation, namely the first occurrence of \textsf{I}, is found.
\begin{center}${\small
	\begin{array}{l@{\hspace{0.2cm}}l|c|c|c|c}
	&\mathsf{Sub}\mbox{-}\mathsf{term} & \mathsf{Context} & \mathsf{\Log} & 
	\mathsf{Tape} & \mathsf{Dir}
	\\
	\cline{1-6}
	&\nustatetab{\var} 
	{(\la\vartwo\la\var{\ctxhole\vartwo})\mathsf{II}} 
	{(\varthree,\la\varthree\ctxhole,\stempty)} {\epsilon} \upp\\
	\iamuaplone&\ndstatetab{\vartwo} 
	{(\la\vartwo\la\var{\var\ctxhole})\mathsf{II}} 
	{\stempty} {(\varthree,\la\varthree\ctxhole,\stempty)} \downp\\
	\iamdvar&\nustatetab{\la\vartwo\la\var{\var\vartwo}} 
	{\ctxhole\mathsf{II}} 
	{(\vartwo,\la\var{\var\ctxhole},(\varthree,\la\varthree\ctxhole,\stempty))} 
	{\stempty} \upp\\
	\iamuaplone&\ndstatetab{\mathsf{I}} 
	{(\la\vartwo\la\var{\var\vartwo})\ctxhole\mathsf{I}} 
	{\stempty} 
	{(\vartwo,\la\var{\var\ctxhole},(\varthree,\la\varthree\ctxhole,\stempty))} 
	\downp\\
	\end{array}}$\end{center}

 \paragraph{Basic invariants.} Given a state 
$(\tm,\ctx,\tlog,\tape,\pol)$, the log and the tape, \ie the token, 
verify two easy invariants connecting them to the position $(\tm,\ctx)$ and the 
direction $\pol$. The log $\tlog$, together with the position 
$(\tm,\ctx)$,  forms a \trpos, \ie the length of $\tlog$ is exactly the level of 
the code context 
$\ctx$\footnote{ Then, the length of $\tlog$ is exactly the number of (linear logic) \emph{boxes} in which the code term is contained.}. This fact guarantees 
that the \LIAM never gets stuck because the log is not long enough for 
transitions $\iamdvar$ and $\iamuapr$ to apply.

About the tape, note that every time the machine switches from a 
$\downp$-state to an $\upp$-state (or vice versa), a \trpos is 
pushed (or popped) from the tape $\tape$. Thus, for reachable states, the number of \trposs in $\tape$ gives the 
direction of the state. These intuitions are formalized by the \emph{tape and direction} invariant below. Given a direction $\pol$ we use
$\pol^n$ for the direction obtained by switching $\pol$ exactly $n$
times (i.e., $\downp^0=\downp$, $\upp^0=\upp$, $\downp^{n+1} =
\upp^{n}$ and $\upp^{n+1}=\downp^{n}$).
%for which we need a preliminary definition.
%
%\begin{definition}[Balanced States]
%  Given a tape $\tape$, we denote with $\sizee\tape$ the number
%  of \trposs in $\tape$. A state $\state = \nopolstate{\tm}{\ctx_n}{\tape}{\tlog}{\pol}$
%  is \emph{balanced} if $(\tm,\ctx_n, \tlog)$ is
%    a \trpos and the direction $\pol$ is
%    $\downp^{\sizee\tape}$.
%\end{definition}
%As expected, all reachable states are balanced.
%\begin{lemma}[Balancing Invariant]\label{lemma:invarianttwo}
%  Let $\state$ be a reachable state. Then $\state$ is balanced.
%\end{lemma}
\begin{lemma}[\LIAM basic invariants]\label{l:invarianttwo}
  Let $\state = \nopolstate{\tm}{\ctx_n}{\tape}{\tlog}{\pol}$ be a reachable state and $\sizee\tape$ the number of 
\trposs in $\tape$. Then 
  \begin{enumerate}
  	\item \emph{Position and log}: $(\tm,\ctx_n, \tlog)$ is a \trpos, and 
	\item \emph{Tape and direction}: $\pol=\downp^{\sizee\tape}$.
  \end{enumerate}
\end{lemma}

\paragraph{Final States and Interpretation}
If the \LIAM starts on the initial state 
$\state_{\tm}$, the execution may either 
never stop or end in a state $\state$ of the shape 
$\state=\dstate{\la\var\tmtwo}{\ctx}{\epsilon}{\tlog}$. 
The fact that no other shapes are possible for $s$ is proved in~\citet{IamPPDPtoAppear}. The \emph{tape and 
direction} invariant guarantees that the machine never stops because the log or 
the tape have not enough logged positions to apply a $\iamdvar$, $\iamuapr$, or a $\iamuaplone$ transition. 
Additionally, on states such as 
$\dstate{ \la\var\ctxtwop\var }{ \ctx }{ \lpos\cons\tape }{ \tlog  }$, the 
logged position $\lpos$ has shape $(\var, \la\var\ctxtwo, \tlog')$, so that 
transition $\iamdlamtwo$ can always apply---this is a consequence of the \emph{exhaustible state invariant} in 
\refsect{invariant}, as shown in \citet{IamPPDPtoAppear}.

The exhaustible state invariant shall be the technical blueprint for the proof of the relationship between the \LIAM 
and the \LJAM, amounting to short-circuiting backtracking phases. Similarly, we shall use it to relate the 
KAM and the \LJAM, and the \LIAM with multi type derivations.

\paragraph{Implementation} Usually, the \LIAM is shown to implement (a 
micro-step variant of) head reduction. The details are quite different from 
those in the usual notion of implementation for environment machines, such as 
the KAM. Essentially, it is shown that the \LIAM induces a 
semantics $\sem{\cdot}{\text{\LIAM}}$ of terms that is a sound and adequate with 
respect to head reduction, rather than showing a bisimulation between the 
machine and head reduction---this is explained at length in 
\cite{IamPPDPtoAppear}. 
For the sake of simplicity, here we restrict to \ccbn. The \LIAM semantics then reduces to just observing 
termination: $\sem{\tm}{\text{\LIAM}}$ is defined if and only if weak head 
reduction terminates on $\tm$. Therefore, we avoid 
discussing semantics and only study termination. 

\begin{theorem}[\cite{IamPPDPtoAppear}]
	The \LIAM implements \ccbn.
\end{theorem}

%The characterization of final states of the \LIAM induces an interpretation of 
%$\lambda$-terms. 
%\begin{definition}[\LIAM Semantics]
%	\label{def:semantics}
%	We define the \LIAM semantics of $\lambda$-terms by way of a 
%	function
%	$\sem{\cdot}{\IAM}:\Lambda\rightarrow\{\Downarrow,\bot\}$, defined as 
%	follows.
%	\[
%	\sem{\tm}{\IAM}=\begin{cases}
%	\Downarrow & \text{ if } 
%	(\red{\tm},\ctxhole,\epsilon,\epsilon)\toiam^*(\red{\la\var\tmtwo},\ctx,\tlog,\epsilon),\\
%	\bot & \text{otherwise.}
%	\end{cases}
%	\]
%\end{definition}
%The semantics above reflects the termination behaviour of terms and vice versa, 
%when weak head reduction is considered as notion of evaluation:
%\[
%(\la\var\tm)\tmtwo\tmthree_1\ldots\tmthree_n\towh 
%\tm\isub\var\tmtwo\tmthree_1\ldots\tmthree_n
%\]
%Note that, for a semantics testing only termination, such as the one considered here, adequacy with respect to a 
%strategy $\to$ implies soundness: if $\tm \to \tmtwo$ and $\sem{\cdot}{}$ is adequate with respect to $\to$, then 
%necessarily $\sem\tm{} = \sem\tmtwo{}$. This is why in this paper we only state adequacy of the interpretations.

\paragraph{Cost of \LIAM Transitions} For all the abstract machines in this paper 
we take random access machines (shortened to RAM) with the 
uniform cost model as the computational model of reference. This is standard in 
the time analyses of abstract machines 
for functional languages. Roughly, it amounts to seeing variables and positions 
as objects  whose manipulation 
takes constant time.

 Every \LIAM transition can then be 
implemented on RAM in constant time but for transition $\tomachvar$, 
whose cost is bounded by the integer $n$ given by $\ctxtwo_n$ (referring to the 
notation of the rules), as the rule needs to split the log after the first $n$ entries. This is in accordance with 
the proof nets interpretation of the \LIAM, because transitions 
$\tomachvar$ correspond to \emph{sequences} of $\IAMold$ transitions on proof nets---see 
\cite{IamPPDPtoAppear}\footnote{Actually, also transition $\tomachbttwo$ corresponds to $n$ $\IAMold$ transitions on 
proof nets. By implementing logs as bi-linked lists, however, $\iamdlamtwo$ can be implemented in constant time. For 
$\tomachvar$ instead, 
there is no easy way out.}. Note that $n$ is bound by the size $\size\tm$ of the (immutable) initial
code $\tm$. The cost of implementing on RAM a \LIAM run $\run$ from $\tm$ then is $\sizenotvar\run + 
\sizevar \run \cdot \size{\tm}$.

\paragraph{Two Useful Properties of the \LIAM} A key property is that the \LIAM is \emph{bi-deterministic}, that 
is, it is deterministic and also deterministically reversible. Another more technical property is that it verifies a 
sort of context-freeness with respect 
to the tape $\tape$. Namely, extending the tape preserves the shape of the 
run and of the final state (up to the extension).

\begin{lemma}[\LIAM tape lift]
\label{l:iam-pumping}
Let $\tape$ be a tape and $ \run: \state = 
\nopolstate{\tm}{\ctx}{\tapetwo}{\tlog}{\pol}\toliam^n
	\nopolstate{\tmtwo}{\ctxtwo}{\tapethree}{\tlogtwo}{\poltwo} = \statetwo$ a run. Then there is a \LIAM run 
	$\run^\tape: \state^\tape = 
	\nopolstate{\tm}{\ctx}{\tapetwo\cons\tape}{\tlog}{\pol}\toliam^n
	\nopolstate{\tmtwo}{\ctxtwo}{\tapethree\cons\tape}{\tlogtwo}{\poltwo} = \statetwo^\tape$.
\end{lemma}

%%%%%%%%%%%%%%%%%%%%
% !TeX spellcheck = en_US
% !TEX root = main.tex
%%%%%%%%%%%%%%%%%%%%%%%%%%%%%%%%%%%%%%%%%%%%%%%%%%%%%%%%%%%%%%%%%%%%%
\section{The Jumping Abstract Machine, Revisited}
\label{sect:jam}
The Jumping Abstract Machine ($\JAMold$) is introduced in \cite{DR99} as an optimization of the $\IAMold$ obtained 
via a sophisticated analysis of $\IAMold$ runs. Here we present the \LJAM, the recasting of the $\JAMold$ in the 
same syntactic framework of the \LIAM. In particular, the \LIAM and the \LJAM rest on the same 
grammars and data structures, they only differ on some transitions.

\paragraph{Jumping Around the Log} The difference between the \LIAM and the \LJAM is in how they create logged 
positions, and consequently on how they backtrack. The \LIAM has a 

\emph{local} approach to logs, and backtracks via potentially long sequences of 
transitions, while the \LJAM follows a \emph{global} approach to logs, and it 
backtracks in just one \emph{jump}.
The transition system is presented in \reffig{jam}. The details of the two 
variations over the \LIAM are:
\begin{itemize}
	\item \emph{Global \trpos}: \trposs created by rule $\iamdvar$ are now 
	global, in that they record the global position of the variable, and not 
	only 
	the position relative to its binder. This way, also the log has to be 
	entirely copied. Differently from the \LIAM, there is some duplication of 
	information.
	\item \emph{Backtracking is short-circuited}: backtracking is a phase of a 
	\LIAM run which is contained between $\iamuapr$ and $\iamdlamtwo$ 
	transitions. It starts when the machine has to rebuild the history of a 
	redex/substitution and ends when the substituted variable occurrence 
	$\lpos$ is found. The optimization at the heart of the \LJAM comes out 
	from the observation that $\lpos$ is exactly the leftmost position on the 
	log. This way, one could directly jump to that position, instead of doing 
	the backtracking. Of course, this is possible only if positions are saved 
	globally in the log, because the $\iamujump$ transitions is not local, but 
	global. 
\end{itemize}
\begin{figure*}[t]
	\input{machines/LJAM}
	\vspace{-8pt}
	\caption{Transitions of the $\lambda$ Jumping Abstract Machine (\LJAM).}
	\label{fig:jam}
\end{figure*}

The absence of the backtracking phase makes the \LJAM easier to 
understand than the \LIAM. In particular, the $\downp$ and $\upp$
phases have now a precise meaning: the former being the quest for the head 
variable of the current subterm, and the latter being the search of the 
argument of the \emph{only} variable occurrence in the tape. This is the second point of the following lemma.

\begin{lemma}[\LJAM basic invariants]\label{l:jam-simple-tape}
  Let $\state = \nopolstate{\tm}{\ctx_n}{\tape}{\tlog}{\pol}$ be a reachable state. Then 
  \begin{enumerate}
  	\item \emph{Position and log}: $(\tm,\ctx_n, \tlog)$ is a \trpos, and 
	\item \emph{Tape and direction}: if $\pol=\downp$, then $\tape$ does not contain any \trpos, 
    otherwise, if $\pol=\upp$, then $\tape$ contains exactly one \trpos.
  \end{enumerate}
\end{lemma}
\begingroup
\setlength{\intextsep}{0pt}
\begin{wrapfigure}{r}{0pt}${\footnotesize
		\begin{array}{l@{\hspace{0.1cm}}l|c|c|c|l}
			&\mathsf{Sub}\mbox{-}\mathsf{term} & \mathsf{Context} & 
			\mathsf{\Log} & 
			\mathsf{Tape} & \mathsf{Dir}
			\\
			\cline{1-6}
			&\ndstatetab{(\la\vartwo{\la\var{\var\vartwo}})\mathsf{II}} 
			{\ctxhole} 
			{\epsilon} {\epsilon} 
			\downp\\
			\iamdap&\ndstatetab{(\la\vartwo{\la\var{\var\vartwo}})\mathsf{I}} 
			{\ctxhole\mathsf{I}} {\resm} {\epsilon} \downp\\
			\iamdap&\ndstatetab{\la\vartwo{\la\var{\var\vartwo}}} 
			{\ctxhole\mathsf{II}} 
			{\resm\cdot\resm} {\epsilon} \downp\\
			\iamdlamone&\ndstatetab{\la\var{\var\vartwo}} 
			{(\la\vartwo\ctxhole)\mathsf{II}} 
			{\resm} 
			{\epsilon} \downp\\
			\iamdlamone&\ndstatetab{\var\vartwo} 
			{(\la\vartwo{\la\var\ctxhole})\mathsf{II}} 
			{\epsilon} {\epsilon}\downp\\
			\iamdap&\ndstatetab{\var} 
			{(\la\vartwo{\la\var\ctxhole\vartwo})\mathsf{II}} 
			{\resm} {\epsilon}\downp
	\end{array}}
	$
\end{wrapfigure}
We present the \LJAM execution trace of the same term considered for the \LIAM. 
In particular, the first transitions are identical to the \LIAM execution since 
no $\iamdvar$ and $\iamuapr$ rules are involved. Instead, we observe 
that full 
context and log are saved at the occurrence of $\iamdvar$ transitions. We put 
$\lpos_\var\defeq(\var,(\la\vartwo{\la\var{\ctxhole\vartwo}}) 
\mathsf{I}(\la\varthree\varthree),\stempty)$.
\begin{center}${\footnotesize
	\begin{array}{l@{\hspace{0.1cm}}l|c|c|c|c}
		&\mathsf{Sub}\mbox{-}\mathsf{term} & \mathsf{Context} & \mathsf{\Log} & 
		\mathsf{Tape} & \mathsf{Dir}
		\\
		\cline{1-6}
		&\ndstatetab{\var} 
		{(\la\vartwo{\la\var{\ctxhole\vartwo}})\mathsf{I}(\la\varthree\varthree)}
		{\resm} {\epsilon} \downp\\
		\iamdvar&\nustatetab{\la\var{\var\vartwo}} 
		{(\la\vartwo\ctxhole)\mathsf{I}(\la\varthree\varthree)} 
		{\lpos_\var\cons\resm}
		{\epsilon} \upp\\
		\iamulam&\nustatetab{\la\vartwo{\la\var{\var\vartwo}}} 
		{\ctxhole\mathsf{I}(\la\varthree\varthree)} 
		{\resm\cons\lpos_\var\cons\resm}
		{\epsilon} \upp\\
		\iamuapltwo&\nustatetab{(\la\vartwo{\la\var{\var\vartwo}})\mathsf{I}} 
		{\ctxhole(\la\varthree\varthree)} 
		{\lpos_\var\cons\resm}
		{\epsilon} \upp\\
		\iamuaplone&\ndstatetab{(\la\varthree\varthree)} 
		{(\la\vartwo{\la\var{\var\vartwo}})\mathsf{I}\ctxhole} 
		{\resm} 
		{\lpos_\var}
		 \downp\\
		\iamdlamone&\ndstatetab{\varthree} 
		{(\la\vartwo{\la\var{\var\vartwo}})\mathsf{I}(\la\varthree\ctxhole)} 
		{\stempty} 
		{\lpos_\var}
		 \downp\\
		\iamdvar&\nustatetab{\la\varthree\varthree} 
		{(\la\vartwo{\la\var{\var\vartwo}})\mathsf{I}\ctxhole} 
		{(\varthree,(\la\vartwo{\la\var{\var\vartwo}})\mathsf{I} 
		(\la\varthree\ctxhole),\lpos_\var)}
		{\lpos_\var}
		 \upp\\
\end{array}}$\end{center}
Finally, as already explained, backtracking is \emph{jumped}: the \LJAM 
restores the 
previously encountered state, saved in the logged position, when 
exiting from 
the right-hand side of an application. We put 
$\lpos_\varthree\defeq(\varthree,(\la\vartwo{\la\var{\var\vartwo}})\mathsf{I} 
(\la\varthree\ctxhole),\lpos_\var)$.
\begin{center}${\footnotesize
	\begin{array}{l@{\hspace{0.2cm}}l|c|c|c|c}
		&\mathsf{Sub}\mbox{-}\mathsf{term} & \mathsf{Context} & \mathsf{\Log} 
		& 
		\mathsf{Tape} & \mathsf{Dir}
		\\
		\cline{1-6}
		&\nustatetab{\la\varthree\varthree} 
		{(\la\vartwo{\la\var{\var\vartwo}})\mathsf{I}\ctxhole} 
		{\lpos_\varthree}{\lpos_\var}\upp\\
		\iamujump&\nustatetab{\var} 
		{(\la\vartwo\la\var{\ctxhole\vartwo})\mathsf{I}(\la\varthree\varthree)} 
		{\lpos_\varthree} {\epsilon} \upp\\
		\iamuaplone&\ndstatetab{\vartwo} 
		{(\la\vartwo\la\var{\var\ctxhole})\mathsf{II}} 
		{\stempty} {\lpos_\varthree} \downp\\
		\iamdvar&\nustatetab{\la\vartwo\la\var{\var\vartwo}} 
		{\ctxhole\mathsf{II}} 
		{(\vartwo,(\la\vartwo\la\var{\var\ctxhole})\mathsf{II},\lpos_\varthree)}
		{\stempty} \upp\\
		\iamuaplone&\ndstatetab{\mathsf{I}} 
		{(\la\vartwo\la\var{\var\vartwo})\ctxhole\mathsf{I}} 
		{\stempty} 
		{(\vartwo,(\la\vartwo\la\var{\var\ctxhole})\mathsf{II},\lpos_\varthree)}
		\downp\\
\end{array}}
$\end{center}

Since the \LJAM is an optimization of the \LIAM, its final states have the 
same shape, namely $\dstate{\la\var\tmtwo}{\ctx}{\epsilon}{\tlog}$ (the fact that the log is always long 
enough to apply transition $\tomachvar$ is given by the \emph{position and log} invariant above). 
In the next section we shall prove that the \LIAM and the \LJAM are 
termination equivalent, obtaining as a corollary that the \LJAM implements \ccbn.

\paragraph{Cost of \LJAM Transitions} The cost of 
implementing \LJAM transitions and runs on RAM is exactly the same as for the 
\IAM: all transitions are atomic but \endgroup
for $\iamdvar$, whose cost is given by the level $n$ of the involved context $\ctxtwo_n$, itself bound by the size of 
the initial code $\tm$. Note that 
this means that in $\iamdvar$ the duplication of 
the log $\tlog$ amounts to the duplication of the pointer to the concrete representation of $\tlog$, and not of the 
whole of $\tlog$ (that would make the cost of $\iamdvar$ much higher, namely depending on the length of the whole run 
that led to the transition).

%%%%%%%%%%%%%%%%%%%%
% !TeX spellcheck = en_US
% !TEX root = main.tex
%%%%%%%%%%%%%%%%%%%%%%%%%%%%%%%%%%%%%%%%%%%%%%%%%%%%%%%%%%%%%%%%%%%%%
\section{Krivine Abstract Machine}
\label{sect:kam}
The Krivine abstract machine \cite{krivine_call-by-name_2007} (\KAM) is a 
standard environment machine for \ccbn 
whose time behavior has been studied thoroughly---in \refsect{jam-complexity} we recall the literature about it. In 
particular, it is a time reasonable implementation of \ccbn, where \emph{reasonably} means polynomially 
related to the time cost model of Turing machines.  To 
be uniform with respect to the other machines, we present the \KAM adding 
information about the 
context, which is not needed.

\begin{figure*}[t]
	\input{machines/KAM}
	\vspace{-8pt}
	\caption{Data structures and transitions of the Krivine Abstract Machine 
	(\KAM).}
	\label{fig:kam}
\end{figure*}

\paragraph{Hopping on Arguments} The \KAM (in \reffig{kam}) differs from the 
\LIAM and \LJAM as it \emph{does 
record} every $\beta$-redex that it encounters---thus 
explicitly entangling time and space consumption---using two data structures. Log and tape are replaced 
by the \emph{(local) environment} $\env$ and the \emph{(applicative) stack} $\stack$.
The basic idea is that, by saving encountered $\beta$-redexes in the environment $\env$, when the machine finds a 
variable occurrence $\var$ it simply looks up in $\env$ for the argument of the binder $\lambda \var$ binding $\var$, 
avoiding the $\upp$-phase of the \LJAM---note that the \KAM has no $\upp$ phase 
and no logs. Mimicking the 
\emph{jump} terminology, one may say that \KAM transition $\tomachvar$ 
\emph{hops} directly on the argument, skipping 
the search for it.
The stack $\stack$ is used to collect encountered arguments that still have to be paired to abstractions to form 
$\beta$-redexes, and then go into the environment $\env$. The intuition is that the stack has an entry for every 
occurrences of $\resm$ on the tape of the \LJAM in the $\downp$-phase, but such entries are more 
informative, they actually record the encountered argument (and a copy of the environment, explained next), and not 
just 
acknowledge its presence.

\paragraph{Closures, Stacks, and Environments.} The mutually recursive grammars for  \emph{closures} and 
\emph{environments}, plus the independent one for \emph{stacks} are defined in \reffig{kam}, together with the 
definition of states.
The idea is that every piece of code comes with an environment, forming a 
closure, which is why environments and closures are mutually 
defined. Also, when the machine executes a closed term $\tm$, every closure 
$(\tmtwo,\ctx, \env)$ in a reachable state is such that for any free variable 
$\var$ of $\tmtwo$ there is an entry $\esub\var\clos$ in $\env$ every, thus 
$\env$ 'closes' $\tmtwo$, whence the name \emph{closures}.

\paragraph{Transitions, Initial and Final States.} Initial states of the KAM 
are in the form $\state_\tm=\kamstate{\tm}{\ctxhole}{\stempty}{\stempty}$. The 
transitions of the \KAM are in 
\reffig{kam}---their union is noted $\tokam$. The idea is 
that the $\iamdvar$ transition looks in the environment for the argument of the 
variable under evaluation. 
As for the other machines, the \KAM
evaluates the term $\tm$ until the top abstraction of the weak head normal 
form of $\tm$ is found, that is a run  either 
never stops or ends in a state $\state$ of the shape 
$\state=\kamstate{\la\var\tmtwo}{\ctx}{\epsilon}{E}$. This is guaranteed by the 
mentioned and standard (but omitted) invariant ensuring that when the initial 
term is closed then every variable appearing in the code has an associated 
closure in the environment, so that the \KAM never gets stuck on a $\tomachvar$ 
transition.
In the next section we shall prove that the \LJAM and the \KAM
are termination equivalent. We show the \KAM execution trace of our running 
example. Initially, the \KAM looks for the head variable keeping track of the 
encountered arguments. We put 
$\mathsf{Q}\defeq\la\vartwo{\la\var{\var\vartwo}}$.
\begin{center}${\footnotesize
		\begin{array}{l@{\hspace{0.1cm}}l|c|c|c}
		&\mathsf{Sub}\mbox{-}\mathsf{term} & \mathsf{Context} & 
		\mathsf{Env.} & 
		\mathsf{Stack}
		\\
		\cline{1-5}
		&\dstatetab{(\la\vartwo{\la\var{\var\vartwo}})\mathsf{II}} 
		{\ctxhole} 
		{\epsilon} {\epsilon} \\
		\kamdap&\dstatetab{(\la\vartwo{\la\var{\var\vartwo}})\mathsf{I}} 
		{\ctxhole\mathsf{I}} 
		{(\mathsf{I},\mathsf{Q}\mathsf{I}\ctxhole,\stempty)}
		{\epsilon}\\
		\kamdap&\dstatetab{\la\vartwo{\la\var{\var\vartwo}}} 
		{\ctxhole\mathsf{II}} 
		{(\mathsf{I},\mathsf{Q}\ctxhole\mathsf{I},\stempty)\cdot(\mathsf{I},\mathsf{Q}\mathsf{I}\ctxhole,\stempty)}
		{\epsilon}\\
		\kamdlam&\dstatetab{\la\var{\var\vartwo}} 
		{(\la\vartwo\ctxhole)\mathsf{II}} 
		{(\mathsf{I},\mathsf{Q}\mathsf{I}\ctxhole,\stempty)} 
		{\esub\vartwo{(\mathsf{I},\mathsf{Q}\ctxhole\mathsf{I},\stempty)}}\\
		\kamdlam&\dstatetab{\var\vartwo} 
		{(\la\vartwo{\la\var\ctxhole})\mathsf{II}} 
		{\epsilon} 
		{\esub\var{(\mathsf{I},\mathsf{Q}\mathsf{I}\ctxhole,\stempty)}\cdot 
			\esub\vartwo{(\mathsf{I},\mathsf{Q}\ctxhole\mathsf{I},\stempty)}=\env}\\
		\kamdap&\dstatetab{\var} 
		{(\la\vartwo{\la\var\ctxhole\vartwo})\mathsf{II}} 
		{(\vartwo,(\la\vartwo{\la\var\var\ctxhole})\mathsf{II},\env)} {\env}
		\end{array}}$\end{center}
Thanks to the information saved in the environment, the \KAM is able to 
directly hop to the argument of $\var$, namely the second identity. Moreover, 
the environment is restored from the closure.
\begin{center}${\footnotesize
		\begin{array}{l@{\hspace{0.1cm}}l|c|c|c}
		&\mathsf{Sub}\mbox{-}\mathsf{term} & \mathsf{Context} & \mathsf{Env.} & 
		\mathsf{Stack}
		\\
		\cline{1-5}
		&\dstatetab{\var} 
		{(\la\vartwo{\la\var{\ctxhole\vartwo}})\mathsf{I}(\la\varthree\varthree)}
		{(\vartwo,(\la\vartwo{\la\var\var\ctxhole})\mathsf{II},\env)} 
		{\env}\\
		\iamdvar&\dstatetab{(\la\varthree\varthree)} 
		{(\la\vartwo{\la\var{\var\vartwo}})\mathsf{I}\ctxhole} 
		{(\vartwo,(\la\vartwo{\la\var\var\ctxhole})\mathsf{II},\env)} 
		{\stempty}\\
		\end{array}}$\end{center}
Then, the computation continues. All application arguments are saved in the 
stack as closures, \ie together with their environment, and then moved to the 
environment when a binder $\lambda\var$ is encountered (they are also linked to 
$\var$). Whenever a variable $\var$ is reached, its argument is retrieved, 
together with its environment from the closure linked to that variable $\var$.
\begin{center}${\footnotesize
		\begin{array}{l@{\hspace{0.1cm}}l|c|c|c}
		&\mathsf{Sub}\mbox{-}\mathsf{term} & \mathsf{Context} & \mathsf{Env.} & 
		\mathsf{Stack}
		\\
		\cline{1-5}
		&\dstatetab{(\la\varthree\varthree)} 
		{(\la\vartwo{\la\var{\var\vartwo}})\mathsf{I}\ctxhole} 
		{(\vartwo,(\la\vartwo{\la\var\var\ctxhole})\mathsf{II},\env)} 
		{\stempty}\\
		\kamdlam&\dstatetab{\varthree} 
		{(\la\vartwo{\la\var{\var\vartwo}})\mathsf{I}(\la\varthree\ctxhole)} 
		{\stempty} 
		{\esub\varthree{(\vartwo,(\la\vartwo{\la\var\var\ctxhole})\mathsf{II},\env)}}\\
		\iamdvar&\dstatetab{\vartwo} 
		{(\la\vartwo\la\var{\var\ctxhole})\mathsf{II}} 
		{\stempty} 
		{\esub\var{(\mathsf{I},\mathsf{Q}\mathsf{I}\ctxhole,\stempty)}\cdot 
			\esub\vartwo{(\mathsf{I},\mathsf{Q}\ctxhole\mathsf{I},\stempty)}}\\
		\iamdvar&\dstatetab{\mathsf{I}} 
		{(\la\vartwo\la\var{\var\vartwo})\ctxhole\mathsf{I}} 
		{\stempty} 
		{\stempty}
		\end{array}}$\end{center}

\paragraph{Cost of \KAM Transitions} The idea is that environments are implemented as linked lists, so that the 
duplication 
and insertion operations in transitions $\tomachapp$ and $\tomachabs$ can be implemented in constant time. Transition 
$\tomachvar$ needs to access the environment, whose size is bounded by 
$\size{\tm}$, the size of the initial term $\tm$ of the 
run. By adopting smarter implementations of envrionments, one $\tomachvar$ transition costs $\log 
\size{\tm}$---see 
\citet{DBLP:conf/ppdp/AccattoliB17} for discussions about implementations of the \KAM. Then implementing on RAM a 
\KAM 
run $\run$ from $\tm$ costs $\sizenotvar\run + \sizevar \run \cdot \log 
\size{\tm}$.

%%%%%%%%%%%%%%%%%%%%
% !TeX spellcheck = en_US
% !TEX root = main.tex
%%%%%%%%%%%%%%%%%%%%%%%%%%%%%%%%%%%%%%%%%%%%%%%%%%%%%%%%%%%%%%%%%%%%%
\section{The Exhaustible State Invariant}
\label{sect:invariant}

Here we present the \emph{exhaustible (state) invariant}. In 
\cite{IamPPDPtoAppear}, this is a key 
ingredient for the proof of the \LIAM implementation theorem. In this paper, 
we give it in various forms to 
establish the relationships between the various machines. Here we 
present the basic concepts.

The intuition behind the invariant is that whenever a
\trpos $\lpos$ occurs in a reachable state, it is there \emph{for a
  reason}: no \trpos occurs in initial states,
and transitions only add \trposs to which the machine may come back. 
In particular, if the state is set in the right way (to be explained), the 
\LIAM can reach $\lpos$, \emph{exhausting} it.
  
\medskip\emph{Preliminaries.} Exhaustible states rest on some \emph{tests} for their \trposs. More specifically, each 
\trpos $\lpos$ in a 
state $\state$ has an associated test state $s_\lpos$ that tunes the data structures of $\state$ as to test for the 
reachability of $\lpos$. Actually, there shall be \emph{two} 
classes of test states, one accounting for the \trposs in the tape of 
$s$, and one for the
 those in the log of $s$. The technical definition of log tests, however, is in 
\refapp{iam-jam-app}. They are essential for the proof of the exhaustible 
 invariant, but they are not needed for showing the main consequence of 
 interest in this section, that is, that backtracking always succeeds 
 (\reflemma{iam-backtracking-succeeds} below), which is why they are omitted.

\medskip\emph{Tape Tests.} Tape tests are easy to define. They focus on one of 
the \trposs in the tape, discarding everything that follows that position on the 
tape.
\begin{definition}[Tape tests]
	Let $\state = 
	\nopolstate{\tm}{\ctx_n}{\tapetwo\cons\lpos\cons\tapethree}{\tlog_n}{\pol}$ 
	be a state. Then the
	\emph{tape test of $\state$ of focus $\lpos$} is the state
	$\state_\lpos=\nopolstate{\tm}{\ctx_n}{\tapetwo 
	\cons\lpos}{\tlog_n}{\upp^{\sizee{\tapetwo \cons\lpos}} }$. 
\end{definition}
Note that the direction of tape tests is reversed with respect to what stated 
by the \emph{tape and direction} invariant (\reflemma{invarianttwo}), and so, in general, they are not reachable 
states. 
Such a counter-intuitive fact is needed for the invariant to go through, no more no less. When we shall use the 
properties of tests to prove properties of the \LIAM (\reflemma{iam-backtracking-succeeds} below), we shall extend 
their 
tape via the tape lifting property (\reflemma{iam-pumping}) as to satisfy the invariant and be 
reachable.
Exhausting a logged position $\lpos$ means backtracking to it. We then decorate the backtracking transition 
$\tomachbtone$ and $\tomachbttwo$ as $\tomachbtonedec$ and $\tomachbttwodec$ to specify the 
involved logged position $\lpos$. We also need a notion of state positioned in $\lpos$ and having an empty tape, 
which is meant to be the target state of $\tomachbttwodec$ when exhausting $\lpos$ starting on $\state_\lpos$.
\begin{definition}[State surrounding a position]
 Let $\lpos=(\tm,\ctxtwo,\tlogtwo)$ be a \trpos. A  state $\state$ surrounds 
 $\lpos$ if $\state = 
 \ustate{\tm}{\ctx_n\ctxholep{\ctxtwo}}{\stempty}{\tlogtwo\cdot\tlog_n}$ for 
 some 
 $\ctx_n$ and $\tlog_n$.
\end{definition}

\paragraph{The Exhaustibility Invariant.} After having introduced all the 
necessary preliminaries, we can now  formulate the property of states that we 
shall soon prove to be an invariant.
\begin{definition}[Exhaustible States]
   $\exstates$ is the smallest
   set of states $\state$ such that if $\state_\lpos$ is a tape or a log 
   test of $\state$ then there exists a run
   $\run:\state_\lpos\toliam^*\tomachbttwodec\statethree$, where
   $\statethree$ surrounds $\lpos$ and for the shortest of such runs $\run$ it 
   holds that $\statethree\in\exstates$. States in $\exstates$ are called 
   \emph{exhaustible}.
\end{definition}
Informally, exhaustible states are those for which every \trpos can be 
successfully tested,  that is, the \LIAM can backtrack to (an exhaustible 
state surrounding) it, if properly initialized. Roughly, a state is exhaustible 
if the backtracking information encoded in its \trposs is coherent. The set
$\exstates$ being the \emph{smallest} set of such states implies
that checking that a state is exhaustible can be finitely certified, \ie there must be a finitary proof.
%\begin{lemma}
%\label{l:recursive-good}
%Let $\state = \nopolstate{\tm}{\ctx_n}{\stme}{\lposn}{\pol}$ be a exhaustible
%state such that $n>0$. Then $\outsi i \state$ is exhaustible for all $i<n$.
%\end{lemma}
%\begin{proof}
%By induction on $n-i$. If $n-i=1$ then $\outsi i \state = \outsi {n-1}
%\state = \outs \state$, which is exhaustible because $\state$ is exhaustible
%(\emph{outer recursion}). If $n-i >1$ then by \ih, $\outsi {n-i-1}
%\state$ is exhaustible. Then $\outsi {n-i} \state = \outs {\outsi {n-i-1}
%  \state}$ is exhaustible because $\outsi {n-i-1} \state$ is exhaustible.
%\end{proof}
\begin{proposition}[Exhaustible invariant \cite{IamPPDPtoAppear}]
\label{prop:good-invariant}
   Let $\state$ be a \LIAM reachable state. Then $\state$ is exhaustible.
\end{proposition}

A key consequence is the fact that backtracking always succeeds, as it amounts to exhausting the first logged 
position on the log.

\begin{lemma}[Backtracking always succeeds]
\label{l:iam-backtracking-succeeds}
 Let $\state$ a \LIAM reachable state. If $\state \tomachbtonedec \statetwo$ 
 then there exists $\statethree$ such that 
$\statetwo\toliam^*\tomachbttwodec \statethree$.
\end{lemma}

\begin{proof}
Consider $\state=\ustate{ \tm }{ \ctxp{\tmtwo\ctxhole} }{ \tape }{ \lpos\cdot\tlog } 
	\tomachbtonedec 
	\dstate{ \tmtwo }{ \ctxp{\ctxhole\tm} }{ \lpos\cdot\tape }{ \tlog } = \statetwo$. 
		Since $\statetwo$ is reachable then it is exhaustible, and so its tape test $\statetwo_\lpos 
		\defeq 
		\dstate{ \tmtwo }{ \ctxp{\ctxhole\tm} }{ \lpos }{ \tlog }$ can be 
		exhausted, that is, there is a \LIAM run
		$\run:\statetwo_\lpos	\toliam^*\tomachbttwodec \statefour$ for a 
		state $\statefour$ surrounding $\lpos$. Note that $\statetwo_\lpos$ is 
		$\statetwo$ where the tape contains only $\lpos$. Now, we lift $\run$ to a run $\run^\tape: 
		\statetwo \toliam^*\tomachbttwodec \statethree$ using the tape 
		lifting lemma (\reflemma{iam-pumping}).
\end{proof}

%%%%%%%%%%%%%%%%%%%%
% !TeX spellcheck = en_US
% !TEX root = main.tex
%%%%%%%%%%%%%%%%%%%%%%%%%%%%%%%%%%%%%%%%%%%%%%%%%%%%%%%%%%%%%%%%%%%%%
\section{Relating the $\lambda$-IAM and the $\lambda$-JAM: Jumping is Exhausting}
\label{sect:iam-jam}
In this section we prove that the \LJAM is a time optimization of the \LIAM via an adaptation of the exhaustible 
invariant. Our proof is based on the construction of a bisimulation which also provides, as a corollary, the 
implementation theorem for the \LJAM.
The basic idea is that the two machines are equivalent \emph{modulo backtracking}. Indeed, the \LJAM 
evaluates terms as the \LIAM, but for the backtracking phase, which is 
short-circuited and done with just one \emph{jump} transition. Then one has to show 
that the \emph{jump} is actually simulated by the \LIAM.

\medskip\emph{Log Tests.} For simulating jumps we need log tests.
The idea is the same underlying tape tests: they focus on a given \trpos in the log. Their definition however requires more than simply stripping down the log, as the new log and the position still have to form a \trpos---said differently, the 
\emph{position and the log} invariant (\reflemma{invarianttwo}) has to be preserved. 
Roughly, the log test $\state_{\lpos_m}$ focussing on the $m$-th \trpos $\lpos_m$ in the log of a state $ 
\nopolstate{\tm}{\ctx_n}{\tape}{\lpos_n\cdots \lpos_2\cdot\lpos_1}{\pol}$ is 
obtained by removing the prefix $\lpos_n\cdots \lpos_{m+1}$ (if any), and 
moving the current 
position up by $n-m$ levels. Moreover, the tape is emptied and 
the direction is set to $\upp$. 

In the argument for the simulation of jumps given below, we need only log tests of a very simple form. Namely, given a state $\state = \ustate{ \tm }{ \ctxp{\tmtwo\ctxhole} }{ \tape }{ \lpos\cdot\tlog }$ from which the \LJAM jumps, we shall consider the log test $\state_\lpos \defeq \ustate{ \tm }{ \ctxp{\tmtwo\ctxhole} }{ \stempty }{ \lpos\cdot\tlog }$, that is, the tape is emptied and (in this case) the position does not change.
The more general form of log tests needing the position change is technical and defined in \refapp{iam-jam-app}---it is 
unavoidable for proving the invariant, but we fear that giving it here would obfuscate the use of the exhaustible 
technique, whose idea is instead quite simple.

\paragraph{I-Exhaustible Invariant} The \LIAM exhaustible invariant proves 
that backtracking phases always succeeds, and it is the key ingredient to relate the 
\LIAM and the \LJAM. While the underlying idea is clear, there is an 
important detail that has to be addressed: to establish the simulation, we have 
to prove that the \LIAM can exhaust logged positions \emph{of the \LJAM}, rather than its 
owns. 

Since the two machines use logs differently, we have to use a function $\JAMtoIAM{\cdot}$ that maps the log-related 
notions of the \LJAM to those of the \LIAM (where $\Gamma$ ranges over both logs and tapes):
\begin{center}$
\begin{array}{rcl}
	\textsc{Logged positions} 
	&&
	\JAMtoIAM{\var, \ctxp{\la\var\ctxtwo_n},\tlogn{\cdot}\tlog}\defeq(\var, 
	\la\var\ctxtwo_n,\JAMtoIAM{\tlogn})
	\\[3pt]
	\textsc{Tapes and Logs} 
	&&
	\JAMtoIAM{\stempty}\defeq\stempty\qquad\qquad
	\JAMtoIAM{\lpos{\cdot}\Gamma}\defeq\JAMtoIAM{\lpos}{\cdot}\JAMtoIAM{\Gamma}\qquad\qquad
	\JAMtoIAM{\resm{\cdot}\tape}\defeq\resm{\cdot}\JAMtoIAM{\tape}
	\\[3pt]
	\textsc{States} 
	&&
	
\JAMtoIAMstate{\nopolstate{\tm}{\ctx}{\tape}{\tlog}{\pol}}\defeq\nopolstate{\tm}{\ctx}{\JAMtoIAM{\tape}}{\JAMtoIAM{\tlog
}}{\pol}
\end{array}$
\end{center}
Another point is that the state surrounding the exhausted position now is 
uniquely determined by the logged position. Given a logged position 
$\lpos=(\var,\ctxtwo,\tlog)$, the \emph{state induced by} $\lpos$ is 
$\indstate\lpos \defeq \ustate{\var}{\ctxtwo}{\stempty}{\tlog}$.

\begin{definition}[I-Exhaustible States]
	$\exstates_I$ is the smallest
	set of \LJAM states $\state$ such that for any tape or log test 
	$\state_\lpos$ of	$\state$ of focus $\lpos$, there 
exists a run 
		$\run:\JAMtoIAM{\state_\lpos}\toliam^*\tomachbttwodecp{\JAMtoIAM\lpos} 
		\JAMtoIAM{\indstate\lpos}$ such that $\indstate\lpos\in\exstates_I$.
	States in $\exstates_I$ are called \emph{I-exhaustible}.
\end{definition}

\begin{lemma}[I-exhaustible invariant]
\label{l:good-invariant-jam}
		Let $\state$ be a \LJAM reachable state. Then $\state$ is I-exhaustible. 
\end{lemma}

\paragraph{Jumping is Exhausting} From the invariant and the tape lifting 
property of the \LIAM, it follows easily that hops can be simulated via 
backtracking, from which the relationship between the \LIAM and the \LJAM 
immediately follows. We write $\tomachhole{\jumpsym,\lpos}$ for a $\tomachjump$ 
transition jumping to $\lpos$.

\begin{lemma}[Jumps simulation via backtracking]
\label{l:jumps-simulation}
Let $\state$ be a \LJAM reachable and $\state \tomachhole{\jumpsym,\lpos} \statetwo$. Then 
$\JAMtoIAM\state\tomachbtonedecp{\JAMtoIAM\lpos} 
\toliam^*\tomachbttwodecp{\JAMtoIAM\lpos}\JAMtoIAM\statetwo$.
\end{lemma}
\begin{proof}
Let $\lpos \defeq (\var,\ctxtwo,\tlogtwo)$ and consider $\state=\ustate{ \tm }{ \ctxp{\tmtwo\ctxhole} }{ \tape }{ 	
		(\var,\ctxtwo,\tlogtwo)\cdot\tlog }
		\tomachhole{\jumpsym,\lpos} 
		\ustate{ \var }{ \ctxtwo }{ \tape }{ \tlogtwo } = \statetwo$.
		Since $\state$ is reachable then it is I-exhaustible, so its log test $\state_\lpos 
		\defeq 
		\ustate{ \tm }{ \ctxp{\tmtwo\ctxhole} }{ \epsilon }{ 
			\lpos\cdot\tlog }$ can be exhausted, that is, there is a \LIAM run
		$\run:\JAMtoIAM{\state_\lpos}	
		\toliam^*\tomachbttwodecp{\JAMtoIAM\lpos} 
\JAMtoIAMstate{\ustate{\var}{\ctxtwo}{\stempty}{\tlogtwo}}=\statethree$. Note that the first transition of $\run$ is 
necessarily $\tomachbtonedecp{\JAMtoIAM\lpos}$. Moreover, $\JAMtoIAM{\state_\lpos}$ and $\statethree$ are exactly 
$\JAMtoIAM{\state}$ and $\JAMtoIAM{\statetwo}$ with empty tape.	
			We lift $\run$ to a run $\run^{\JAMtoIAM\tape}: \JAMtoIAM{\state_\lpos}^{\JAMtoIAM\tape}
\tomachbtonedecp{\JAMtoIAM\lpos} \toliam^*\tomachbttwodecp{\JAMtoIAM\lpos}
		\statethree^{\JAMtoIAM\tape}$ using \reflemma{iam-pumping}. Now, $\run^{\JAMtoIAM\tape}$ is exactly the \LIAM 
simulation of the jump, because $\JAMtoIAM{\state_\lpos}^{\JAMtoIAM\tape}=\JAMtoIAM{\state}$ and 
$\statethree^{\JAMtoIAM\tape} = \JAMtoIAM{\state}$.
\end{proof}

From the lemma it easily follows a bisimulation between the \LIAM and the \LJAM, showing that the latter is faster. 
In \refapp{iam-jam-app}, there is a general theorem relating also potentially diverging runs. Here we give only the 
more concise statement about complete runs.

\begin{theorem}[\LIAM and \LJAM relationship]
\label{thm:ij-concise}
There is a complete \LJAM run $\run_J$ from  $\tm$ if and only if there 
is a complete \LIAM 
run $\run_I$ from $\tm$. In particular, the \LJAM implements \ccbn. Moreover, 
$\size{\run_J}\leq \size{\run_I}$ and $\sizevar{\run_J}\leq \sizevar{\run_I}$.
\end{theorem}

\paragraph{Exponential Gap} The time gap between the \LIAM and the \LJAM can be exponential, as it is shown by the 
family of terms $\tm_n$ ($\tm_1 \defeq \Id$ and $\tm_{n+1} \defeq \tm_n \Id$) mentioned in the introduction. The results 
of this paper provide a nice high-level proof. Next section shows that the time of the \LJAM is polynomial in the time 
of the \KAM, that takes time polynomial in the number of $\beta$-steps to 
evaluate $\tm_n$, that is, $n$. The study of 
multi types in \refsect{types} instead shows that the time of the \LIAM depends on the size of the smallest type 
$\ty_n$ of $\tm_n$, which is easily seen to be exponential in $n$. In fact, 
using the notation of \refsect{types}, $\ty_1 \defeq 
\arr{\mset\initty}\initty$, and $\ty_{n+1} \defeq \arr{\mset{\ty_n}}{\ty_n}$.

%%%%%%%%%%%%%%%%%%%%
% !TeX spellcheck = en_US
% !TEX root = main.tex
%%%%%%%%%%%%%%%%%%%%%%%%%%%%%%%%%%%%%%%%%%%%%%%%%%%%%%%%%%%%%%%%%%%%%
\section{Entangling the \LJAM and the \KAM: the \HAM}
\label{sect:jam-kam}

Now we turn to the relationship between the \LJAM and the \KAM. We prove that \KAM runs can be obtained from \LJAM 
ones via \emph{hops} that 
short-circuit the search for arguments realised by the blue transitions. It then follows that the \KAM can be seen as 
a time improvement of the \LJAM.

%\begin{example}
%	We show concretely that the two aforementioned cases can occur.
%	\begin{itemize}
%		\item Let us consider the term 
%		$\tm=(\la\var\la\vartwo\vartwo)(\la\varthree\varthree)$. The \LJAM 
%		does 
%		not consume any space to evaluate $\tm$ since no variables are reached 
%		during the computation. Instead, the \KAM consumes one unit of space 
%		since it builds one closure, namely 
%		$(\la\varthree\varthree,(\la\var\la\vartwo\vartwo)\ctxhole,\stempty)$. 
%		Thus $\spacem{\tm}^\JAM=0<1=\spacem{\tm}^\KAM$.
%		\item Let us consider the term 
%		$\tmtwo=(\la\var\var\var)(\la\vartwo\vartwo)$. Only two applications 
%		are crossed during the evaluation of $\tmtwo$ on the \KAM, while three 
%		variables are encountered by the \LJAM evaluation. Thus 
%		$\spacem{\tmtwo}^\KAM=2<3=\spacem{\tmtwo}^\JAM$.
%	\end{itemize}
%\end{example}

\begin{figure}[t]
	\input{machines/HAM}
	\vspace{-8pt}
	\caption{Data structures and transitions of the Hopping Abstract Machine (\HAM).}
	\label{fig:jkam}
\end{figure}
\paragraph{The \HAM} To prove that the \KAM is a time improvement of the \LJAM, we 
introduce an intermediate machine, the \emph{Hopping Abstract Machine} (\HAM) in \reffig{jkam}, that merges the two. 
The \HAM is a technical tool addressing an inherent difficulty: the \LJAM  and the \KAM use different data 
structures and it is impossible to turn a \KAM state into a \LJAM state without having to look at the whole run that 
led to that state, as it is instead possible for the \LJAM and the \LIAM. 

The idea behind the \HAM is to entangle the data structures of both machines (so that their states get paired by 
construction), and to allow it to behave non-deterministically either as the 
\LJAM or as
the \KAM. 
The \HAM deals with two enriched objects, \emph{logged closures} $\lclos$ and \emph{closed (logged) positions} $\clpos$
(defined in 
\reffig{jkam}, overloading some of the notations of the previous sections), obtained by adding a log to closures and an 
environment to logged positions. Of course, environments and logs have to be redefined as containing these enriched 
objects. There is also a \emph{(closed) tape} $\tape$, that is, a data structure obtained by merging the 
roles of the stack and the tape and containing both logged closures and closed positions. In fact the closed tape is 
obtained from the \LJAM tape by upgrading every $\resm$ entry to a logged closure $\lclos$, and every logged position 
$\lpos$
to a closed one $\clpos$. 
Note that logged closures and closed positions contain the same information (a term, a 
context, a log, and an environment) but they play different roles.

The non-determinism of the machine amounts to the presence of \emph{two} transitions $\tomachvarj$ and $\tomachvark$ for 
the variable case, that are simply the $\varsym$ transitions of the \LJAM and the \KAM, lifted to the new data 
structures. In particular, transition $\tomachvark$ short-circuits a whole $\upp$ phase of the \LJAM \emph{hopping} 
directly to the argument.

It is evident that by removing environments, turning every logged closure into $\resm$, and removing $\tomachvark$ we 
obtain the \LJAM. Similarly, by removing logs, $\tomachvarj$, and the $\upp$ transitions, one obtains the \KAM. We 
avoid spelling out these immediate projections. Instead, we see \KAM runs inside the \HAM as given by the transition 
$\tohamk \defeq \tomachdotoneapp \cup \tomachdottwoabs \cup \tomachvark$. Similarly, the \LJAM is seen as transition 
$\tohamj$, defined as the union of all \HAM transitions but $\tomachvark$.

The \HAM verifies the same basic properties of the \LJAM, simply lifted to 
the enriched data structures. Moreover, it verifies a tape lifting property.
\begin{lemma}[\HAM tape lift]
\label{l:ham-pumping}
Let $ \run: \state = \hamstatenopol{\tm}{\ctx}{\tlog}{\env}{\tapetwo}{\pol} \toham^n
	\hamstatenopol{\tmtwo}{\ctxtwo}{\tlogtwo}{\envtwo}{\tapethree}{\poltwo} = \statetwo$ be a run and $\tape$ be a tape. 
Then there is a run 
	$\run^\tape: \state^\tape = \hamstatenopol{\tm}{\ctx}{\tlog}{\env}{\tapetwo\cons\tape}{\pol}\toham^n
	\hamstatenopol{\tmtwo}{\ctxtwo}{\tlogtwo}{\envtwo}{\tapethree\cons\tape}{\poltwo} = \statetwo^\tape$.
\end{lemma}

\section{Hopping is Also Exhausting} 
\label{sect:hopping}
Since jumping and hopping amount to a similar idea, the proof technique that we use to relate the \LJAM and the \KAM 
is obtained by another variation on the exhaustible invariant.

\paragraph{Testing Logged Closures} The main difference is that now we exhaust \emph{logged closures} instead of logged 
positions. Via the $\upp$-exhaustible invariant below we shall show that the \HAM can exhaust a logged closure---that 
is it can recover the argument in the closure---by using only \LJAM $\upp$ transitions. This capability shall then be 
used to show that the \LJAM can simulate hops.

Since logged closures are both in the environment and in the tape, we have two kinds of test. The definition 
of tape tests is in the Appendix. They are essential for the proof of 
the $\upp$-exhaustible invariant, but 
they are not needed for the argument at work in the simulation, spelled out below.

\paragraph{Environment Tests} Given a \HAM state 
$\hamstatenopol{\tm}{\ctx}{\tlog}{\env}{\tape}{\pol}$ consider an entry 
$\esub\var\lclos$ in $\env$. The idea is that one wants to exhaust $\lclos$ to 
return to the state saved in $\lclos$. Remember that the \LJAM looks for the 
argument starting from the binder of $\var$. Then, the test associated to 
$\lclos$ is obtained by positioning the machine on the binder $\lambda\var$ for 
$\var$, and modifying the log and the environment accordingly. Moreover, the 
tape is emptied.

\begin{definition}[\HAM Environment tests]
	Let $\state =	\hamstatenopol{\tm}{\ctxp{\la\var\ctxtwo_{n}}}{\tlog_{n}\cons \tlog}{E'\cons\esub\var \lclos\cons 
E}{\tape}{\pol}$ be a
	state. Then, $\state_{\lclos} \defeq \hamstateu{\la\var\ctxtwo_n\ctxholep{\tm}}{\ctx}{\tlog}{E}{\stempty}$
	is an environment test for $\state$ of focus $\lclos$.
\end{definition}

As in the previous section, we need a notion of state induced by a logged closure $\lclos$, that is the state reached 
by a run exhausting $\lclos$. The definition may seem wrong, an explanation follows.
\begin{definition}[\HAM State induced by a logged closure]
	Given a logged closure 
	$\lclos=(\tmtwo,\ctxtwop{\tm\ctxhole},\env)^\tlog$,
	the state $\indstate{\lclos}$ induced  by $\lclos$ is defined as
	$\indstate{\lclos} \defeq 
	\hamstateu{\tm}{\ctxtwop{\ctxhole\tmtwo}}{\tlog}{\env}{\stempty}$.
\end{definition}
The previous definition is counter-intuitive, as one would expect $\indstate{\lclos}$ to rather be the state $\statetwo 
\defeq \hamstated{\tmtwo}{\ctxtwop{\tm\ctxhole}}{\tlog}{\env}{\stempty}$, but for technical reasons this is not 
possible. In the simulations of hops below, however, $\indstate{\lclos}$ is tape lifted to a state that makes a 
$\tomacharg$ transition to (a tape lifting of) $\statetwo$, as one would expect. We set $\tomachup \defeq 
\tomachhole{\resm 3, \resm 4 , \argsym, \jumpsym}$.
\begin{definition}[\HAM $\upp$-Exhaustible states]
	$\exstates_{\upp}$ is the smallest
	set of those states $\state$ such that for any tape or environment test $\state_{\lclos}$ of $\state$,
		there exists a run
		$\run_{\upp}:\state_{\lclos}\tomachup^*\indstate\lclos$ and 
		$\indstate\lclos \in\exstates_{\upp}$.
	States in $\exstates_{\upp}$ are called \emph{$\upp$-exhaustible} (pronounced \emph{up}-exhaustible).
\end{definition}

\begin{lemma}

		Let $\state$ be a \HAM reachable state. Then $\state$ is $\upp$-exhaustible.
\end{lemma}

\paragraph{Simulating Hops} From the invariant and the tape lifting property of the 
\HAM, it follows easily that hops can be simulated via $\tomachup$, as the next lemma shows.

\begin{lemma}[Hops simulation via $\upp$]
\label{l:hops-simulation}
Let $\state$ be a \HAM reachable state and $\state \tomachvark \statetwo$. Then 
$\state\tomachvarj\statethree\tomachup^+\statetwo$.
\end{lemma}
\begin{proof}
The hypothesis is: $\state=\hamstated{ \var }{ \ctx }{ \tlog}{ \env }{ \tape } 
		\tomachvark
	\hamstated{ \tmtwo}{ \ctxtwop{\tm\ctxhole} }{ \clpos\cons \tlogtwo }{ F }{ \tape } = \statetwo$
	where 
	$\env=\envtwo\cons\esub\var{\lclos}\cons \envthree$ with $\lclos 
	=(\tmtwo,\ctxtwop{\tm\ctxhole},F)^{\tlogtwo}$ and $\clpos\defeq(\var, 
	\ctx,\tlog)^{\env}$. From $\state$ the \HAM can also do a $\tomachvarj$ transition:
		$\state = \hamstated{ \var }{ \ctx'\ctxholep{\la\var\ctxtwo_{n}'} } { \tlogn\cons\tlogthree}{ \envtwo\esub\var 
\lclos\envthree }{\tape }    
		\tomachvarj 
		\hamstateu{ \la\var\ctxtwo_{n}'\ctxholep\var}{ \ctx' }{ \tlogthree }{ \envthree }{ \clpos\cons\tape} = 
\statethree$
where $\tlog =  \tlogn\cons\tlogthree$ and 
$\ctx=\ctx'\ctxholep{\la\var\ctxtwo_{n}'}$.
	Now consider the environment test $\statetwo_{\lclos}= \hamstateu{ 
	\la\var\ctxtwo_{n}\ctxholep\var}{ \ctx }{ \tlogthree }{ \envthree }{ \stempty }$. By $\upp$-exhaustibility we obtain 
	$\run: \statetwo_{\lclos} \tomachup^* \indstate{\lclos} = \hamstateu{ \tm}{ \ctxtwop{\ctxhole\tmtwo} }{ \tlogtwo }{ F 
}{ \stempty }$
	Then, lifting $\indstate{\lclos}$ with the tape $\clpos\cons\tape$, one has
	$
	\indstate{\lclos}_{\clpos\cons\tape}=\hamstateu{ \tm}{ 
		\ctxtwop{\ctxhole\tmtwo} }{ \tlogtwo }{ F }{ \clpos\cons\tape } 
		\iamuaplone 
		\hamstated{ \tmtwo }{ \ctxtwop{\tm\ctxhole} }{ \clpos\cons\tlogtwo }{ F }{ \tape }$
	Thus, $\state\jamdvar\statethree\tomachup^*\indstate{\lclos}_{\clpos\cons\tape} 
	\iamuaplone\statetwo$.
\end{proof}
	From the lemma it easily follows a bisimulation between the \LJAM and the \KAM, showing that the latter is 
faster. In \refapp{hopping-app}, there is a theorem relating their runs inside the \HAM, considering also 
potentially diverging runs. Here we give only the 
more concise statement about complete runs.

\begin{theorem}[\LJAM and \KAM relationship]
\label{thm:jk-concise}
There is a complete \LJAM run $\run_J$ from 
$\tm$ if and only if 
there is a complete \KAM run $\run_K$ from $\tm$. Moreover, 
$\size{\run_J}= \size{\run_K} + \sizeup{\run_J}$ and $\sizevar{\run_J}= \sizevar{\run_K}$.
\end{theorem}

%%%%%%%%%%%%%%%%%%%%
% !TeX spellcheck = en_US
% !TEX root = main.tex
%%%%%%%%%%%%%%%%%%%%%%%%%%%%%%%%%%%%%%%%%%%%%%%%%%%%%%%%%%%%%%%%%%%%%
\section{The $\lambda$-JAM is Slowly Reasonable}
\label{sect:jam-complexity}
In this section we provide bounds for the complexity of the \LJAM. First, we show that it is quadratically slower 
than the \KAM, and then, by using results from the literature about the \KAM, we obtain bounds with respect to the 
two parameters for complexity analyses of abstract machines, namely, the 
size $\size\tm$ of the evaluated term and the number $\#\beta$ of $\towh$-steps to 
evaluate $\tm$.

\paragraph{Locating the \LJAM} We have proved in the previous two sections that a run $\run_J$ of the \LJAM 
from $\tm$ is such that $\size{\run_K}\leq\size{\run_J}\leq\size{\run_I}$, 
where $\run_K$ and $\run_I$ are the runs from $\tm$ respectively of the \KAM 
and of the \LIAM. However, this tells nothing about the inherent complexity of 
evaluating a term with the \LJAM. In fact, it is well known that 
$\size{\run_K}$ is polynomial in $\#\beta$ and $\size\tm$ (namely quadratic in $\#\beta$ and linear in $\size\tm$), 
while $\size{\run_I}$ can be exponential in both $\#\beta$ and $\size\tm$ (the 
typical example of exponential behavior being the family of terms $\tm_n$ defined as $\tm_1 \defeq \Id$ and $\tm_{n+1} 
\defeq \tm_n \Id$). What about the \LJAM? Is it polynomial or 
exponential? It turns out that the \LJAM is polynomial, and precisely at most quadratically slower than the \KAM.

\paragraph{Bounding $\upp$ Phases} Since the \KAM is the \LJAM less the (\blue{blue}) $\upp$ phases, and the 
complexity of the \KAM is known, we only have to study the length of $\upp$ phases. The length 
of a $\upp$ phase extending a run $\run$ from $\tm$ is bound by $\sizevar{\run}\cdot\size{\tm}$, and the length of all 
$\upp$ phases together is bound by $\sizevar{\run}^2\cdot\size{\tm}$. The proof is in three steps. First, we show that 
in absence of jumps a 
$\upp$ phase cannot be longer than $\size\tm$. An immediate induction on $\size\ctx$ proves the following lemma.
\begin{lemma}\label{l:boundC}
	Let $\run:\ustate{\tm}{\ctx}{\tape}{\tlog}\tomachhole{\resm 3, \resm 4}^*\state$. Then 
	$\size\run\leq \size\ctx\leq \size{\ctxp\tm}$.
\end{lemma}
Second, we need an invariant. To estimate the number of jumps in a $\upp$ phase, we 
need to link the structure of logs with the number of $\iamdvar$ transitions encountered so far.
We introduce the notion of \emph{depth} of a tape/log 
$\Gamma$, defined in 
the following way:
\[
\begin{array}{rcl@{\hspace{.7cm}}rcl}
\spdepth\stempty &\defeq &0&
\spdepth{\resm\cdot\tape} &\defeq& \spdepth\tape
\\
\spdepth{\lpos\cdot\Gamma} &\defeq& \spdepth\lpos&
\spdepth{(\var,\ctx,\tlog)} &\defeq& 1+\spdepth\tlog
\\
\spdepthnopar{\nopolstate{\tm}{\ctx}{\tape}{\tlog}{\upp}} &\defeq &\spdepth\tape
&
\spdepthnopar{\nopolstate{\tm}{\ctx}{\tape}{\tlog}{\downp}} &\defeq 
&\spdepth\tlog

\end{array}\]
\begin{proposition}[Depth invariant]
\label{l:ljam-var-invariant}
	Let $\run:\state_\tm\toljam^*\state$ be 
	an initial run of the \LJAM. Then 
	$ \spdepth\state = \sizevar{\run}$. Moreover $\spdepth\state \geq \spdepth\lpos$ for every logged position $\lpos$ in 
$\state$.
\end{proposition}

Third, we bound $\upp$ phases. The number of jumps in a single phase $\state\tomachup^*\statetwo$ of $\upp$ transitions 
is bound by $\spdepth\state$, and 
pairing it up with \reflemma{boundC} we obtain a bound on the phase. By the depth invariant the bound can be given 
relatively to $\sizevar{\run}$, and a standard argument extends the bound 
to all $\upp$ phases in a run, adding a quadratic dependency. Let $\sizeup\run$ be the number of $\tomachup$ 
transitions in 
$\run$.
\begin{lemma}[Bound on $\upp$ phases]
\label{l:bound-up-phase} %\reflemmap{bound-up-phase}{global}
\hfill
\begin{enumerate}
\item \emph{One $\upp$ phase}: if $\state=\ustate{\tm}{\ctx}{\tape}{\tlog}$ is a reachable state and 
	$\run:\state\tomachup^*\statetwo$ then 
	$\size\run\leq \spdepth\state\cdot\size{\ctxp\tm}$.
	
	\item \label{p:bound-up-phase-global}
	\emph{All $\upp$ phases}:	if $\run:\state_\tm\toljam^*\state$ then 
	$\sizeup\run \leq \sizevar\run^2\cdot\size\tm$. 
\end{enumerate}
\end{lemma}

\paragraph{The Complexity of the \KAM} We need to recall the complexity analysis of the \KAM from 
\cite{DBLP:conf/rta/AccattoliL12,DBLP:conf/icfp/AccattoliBM14,DBLP:conf/ppdp/AccattoliB17}. The length of a complete 
\KAM run $\run$ from $\tm$ verifies $\size\run = \sizevar\run + 2\cdot\sizeabs\run$ and we have 
$\sizevar\run = \bigo{\sizeabs\run^2}$ (the bound is tight, as there are examples reaching it). Since 
$\sizeabs\run$ is exactly the number $\#\beta$ of $\towh$ steps to evaluate $\tm$ (the cost model of reference for 
\ccbn), we have that $\size\run = \bigo{\#\beta^2}$. Now since the cost of implementing \KAM single transitions on RAM 
is 
bound by $\size\tm$, the complexity of implementing the \KAM is $\bigo{\#\beta^2\cdot\size\tm}$, that is, the \KAM is a 
reasonable machine.

\paragraph{The Complexity of the \LJAM} From the complexity of the \KAM, the fact 
that the \LJAM and the \KAM do exactly the same number of $\tomachvar$ transitions, and that the number of $\upp$ 
transition of the \LJAM are bound by $\sizevar\run^2\cdot\size\tm$, we easily obtain the following results.

\begin{theorem}[\LJAM complexity]
	Let $\tm$ be a closed term such that $\tm \towh^n \tmtwo$, $\tmtwo$ be 
	$\towh$ normal, and $\run_J$ and $\run_K$ be the complete \LJAM and  \KAM runs from $\tm$. Then:
	\begin{enumerate}
		\item \emph{The \LJAM is quadratically slower than the \KAM}: $\size{\run_K} \leq \size{\run_J} 
		= \bigo{\size{\run_K}^2\cdot \size\tm}$.

		\item \emph{The \LJAM is (slowly) reasonable}: $\size{\run_J} = \bigo{n^4\cdot \size\tm}$, and the cost of 
		implementing $\run_J$ on a RAM is also $\bigo{n^4\cdot 
			\size\tm}$.
	\end{enumerate}
\end{theorem}
% In \refapp{jam-complexity-app}, we sketch a folklore optimization that lowers the complexity of the \KAM to 
% $\bigo{\#\beta\cdot\size\tm}$ and how to extend it to the \LJAM, lowering its complexity to $\bigo{n^2\cdot 
% \size\tm}$.

%%%%%%%%%%%%%%%%%%%%
% !TeX spellcheck = en_US
% !TEX root = main.tex
%%%%%%%%%%%%%%%%%%%%%%%%%%%%%%%%%%%%%%%%%%%%%%%%%%%%%%%%%%%%%%%%%%%%%
\section{The Pointer Abstract Machine, Revisited}
\label{sect:pam}

The Pointer Abstract Machine (\PAM), due to Danos and Regnier \cite{DBLP:conf/lics/DanosHR96,Danos04headlinear}, 
gives an operational account of the 
interaction process at work in Hyland and Ong game semantics. The machine 
is always described rather 
informally via a pseudo-code algorithm. Here we define it according 
to our syntactic style, calling it \LPAM, and provide its first formal and manageable
presentation as an actual abstract machine.

Our result concerning the \LPAM is that it is strongly bisimilar to the 
\LJAM. Roughly, the two are the same machine, with exactly the same time 
behavior, they just use different data structures. This connection is mentioned 
 in \cite{DR99}, but not proved. We find it instructive to 
spell it out, as the connection is elegant but far from being evident.

\paragraph{{Fragmented} $vs$ {Monolitic} Run Traces} Both machines jump and 
need to store 
information about the run, to jump to the right place. They differ on how they 
represent this information. The \LJAM uses logged positions, that is, 
positions coming with the information to be restored after the jump. The 
approach can be seen as \emph{fragmented}, as the trace of the run is 
distributed 
among {all the} logged positions {in the state}. The \LPAM adopts a 
\emph{monolitic} approach, storing all the information in a unique 
\emph{history} $\history$, a new data structure encoding the whole run 
in a minimalistic and sophisticated way. {Roughly, the history $\history$ saves 
all the variable positions $\pos$ for which an argument as been found, each one 
with a pointer (under the form of an index $i$) to a previous variable position 
$\postwo$ in $\history$. The index $i$ intuitively realizes a mechanism to 
retrieve the log associated to $\pos$ by the \LJAM.} We first define the 
machine and then 
explain the relationship between the two approaches. 

\paragraph{Data Structures}  All the data structures of the \PAM are defined 
in \reffig{pam}. Positions are no longer logged, and noted with $\pos$, 
$\postwo$, etc. An \emph{index} $i$ is simply a natural number. \emph{Indexed positions} are 
pairs $\indp \pos i$. A history $\history$ is a sequence of indexed variable 
positions ({accumulated from right to left}). The idea is that {indices are 
pointers to entries in the history, 
that is,} if the $i$-th entry of $\history$ is $\indp \pos j$ 
then $j$ points to a previous entry in $\history$, that is, $j<i$.  The tape of 
the \LPAM is a  stack containing variable positions and occurrences of 
$\resm$. 

\paragraph{Transitions and Look-Up} Initial states have the form $\state_\tm\defeq 
\pamstated{\tm}{\ctxhole}{\stempty}{0}{\stempty}$, the transitions of the \LPAM are in 
\reffig{pam}, they are labeled exactly as in the \LJAM, and their union is noted $\tolpam$. Transitions $\tomachvar$ and $\tomachjump$ need to retrieve information from the history $\history$, for which there is some dedicated notation. We use $i_k^\history,\var_k^\history,\ctxtwo_k^\history$ to 
denote, respectively, the index, variable, and context of the $k$-th indexed position in $\history$. 

Transition $\tomachvar$ moreover looks up into $\history$ in an unusual way. The idea is that it accesses $\history$ 
$n$ times to retrieve an index. The first time it retrieves the indexed position $\indp {\pos_1} {j_1}$ of index $i$, 
to then retrieve the position $\indp {\pos_2} {j_2}$ of index $j_1$, and so on, until it retrieves $j_n$ and makes it 
the new state index. This is formalized using the look-up function $\phi_\history:\mathbb{N}\to\mathbb{N}$ defined as 
$\phi_\history(k)=i_k^\history$, and whose powers $\phi_\history^n$ are defined as 
$\phi_\history^n(k)=\phi_\history(\phi_\history^{(n-1)}(k))$, where $\phi_\history^0(k)=k$. Note that implementing 
$\tomachvar$ on RAM then costs $n$, that is bound by the size $\size\tm$ of the initial term, exactly as for the 
\LJAM, while all other transitions have constant cost.

\paragraph{Final States and Invariants} Final states of the \LPAM have, as 
expected, shape $\pamstatenopol {\red{\underline{\la\var\tm}}}\ctx  \history i 
\stempty\downp$. This follows from the fact that the machine is never stuck on 
$\tomachvar$ steps because $\phi_\history^n(i)$ is undefined. Note indeed a 
subtle point: $\phi_\history(0)$ is undefined, so, potentially, 
$\phi_\history^n(i)$ may  be undefined. We then need an invariant ensuring that---in the source state of $\tomachvar$---
$\phi_\history^n(i)$ is always defined. The next statement collects also other 
minor invariants of the \LPAM.

We say that $\history$ \emph{has depth} (at least) $n\in \nat$ at $i$ if $n=0$ or if $n>0$ and $\phi_\history^m(i)>0$ for every $m<n$. 

\begin{lemma}[\LPAM invariants]
\label{l:pam-simple-tape}
    Let $\state=\pamstatenopol\tm{\ctx_n} \history i \tape\pol$ be a reachable 
    \PAM state. 
    Then:
  \begin{enumerate}
  	\item \emph{Depth}: $\history$ has depth $n$ at $i$. Moreover, if $\indp {(\tmtwo,\ctxtwo_m)} j$ is the $k$-th 
indexed position of $\history$, with $k>0$, then $\history$ has depth $m$ at $k-1$.
	\item \emph{Tape, index, and direction}: if $\pol=\downp$, then 
$i = \size\history$ and $\tape$ does not contain any \trpos, otherwise if $\pol=\upp$ then $\tape$ contains exactly 
one position.	
  \end{enumerate}
\end{lemma}

\begin{wrapfigure}{r}{0pt}${\footnotesize
		\begin{array}{l|c|c|c|c|c|c}
		&\mathsf{Sub}\mbox{-}\mathsf{term} & \mathsf{Context} & \mathsf{Hist.} 
		& \mathsf{Index} & \mathsf{Tape} & \mathsf{Dir}
		\\
		\cline{1-7}
			&\pamstatedtab{(\la\vartwo{\la\var{\var\vartwo}})\mathsf{II}} 
			{\ctxhole} 
			{\epsilon}{0} {\epsilon}& 
			\downp\\
			\iamdap&\pamstatedtab{(\la\vartwo{\la\var{\var\vartwo}})\mathsf{I}} 
			{\ctxhole\mathsf{I}} {\epsilon}{0}{\resm}&  \downp\\
			\iamdap&\pamstatedtab{\la\vartwo{\la\var{\var\vartwo}}} 
			{\ctxhole\mathsf{II}} 
			 {\epsilon}{0} {\resm\cdot\resm}& \downp\\
			\iamdlamone&\pamstatedtab{\la\var{\var\vartwo}} 
			{(\la\vartwo\ctxhole)\mathsf{II}} 
			{\epsilon}{0}{\resm} &\downp\\
			\iamdlamone&\pamstatedtab{\var\vartwo} 
			{(\la\vartwo{\la\var\ctxhole})\mathsf{II}} 
			{\epsilon}{0} {\epsilon}&\downp\\
			\iamdap&\pamstatedtab{\var} 
			{(\la\vartwo{\la\var\ctxhole\vartwo})\mathsf{II}} 
			{\epsilon}{0} {\resm}& \downp
	\end{array}}
	$
\end{wrapfigure}

\paragraph{An Example} As for the other machines we have considered in this 
paper, we give the execution trace of the \LPAM on the term 
$(\la\vartwo{\la\var{\var\vartwo}})\mathsf{II}$. The reader can 
grasp some intuition considering that the \PAM is strongly bisimilar to the 
\LJAM. In particular, the \LPAM considers explicit pointers. Indeed, as we have 
already pointed out, \LJAM logs are not actually copied in the \LJAM $\iamdvar$ 
transition: what is duplicated is just a pointer to them. The \LPAM 
handles this mechanism directly in its definition, and can thus be considered 
as a low-level 
implementation of the \LJAM. In the following we will explain this in more 
detail. After having looked for the head variable through the spine of the 
term, the \LPAM, now in $\upp$ mode, queries the argument of $\var$, namely 
$\la\varthree\varthree$, that then explores. The argument of its head variable 
$\varthree$ is $\vartwo$, that has to be found via \emph{backtracking} or 
\emph{jumping}. We put 
$\pos_\var=(\var,(\la\vartwo{\la\var{\ctxhole\vartwo}}) 
\mathsf{I}(\la\varthree\varthree))$.

\begin{center}$
	{\footnotesize
	\begin{array}{l|c|c|c|c|c|c}
	&\mathsf{Sub}\mbox{-}\mathsf{term} & \mathsf{Context} & \mathsf{Hist.} 
	& \mathsf{Index} & \mathsf{Tape} & \mathsf{Dir}
	\\
	\cline{1-7}
	&\pamstatedtab{\var} 
	{(\la\vartwo{\la\var{\ctxhole\vartwo}})\mathsf{I}(\la\varthree\varthree)}
	 {\epsilon}{0}{\resm} &\downp\\
	\iamdvar&\pamstateutab{\la\var{\var\vartwo}} 
	{(\la\vartwo\ctxhole)\mathsf{I}(\la\varthree\varthree)} 
	{\epsilon}{0}{\pos_\var\cons\resm} &\upp\\
	\iamulam&\pamstateutab{\la\vartwo{\la\var{\var\vartwo}}} 
	{\ctxhole\mathsf{I}(\la\varthree\varthree)} 
	{\epsilon}{0}{\resm\cons\pos_\var\cons\resm} &\upp\\
	\iamuapltwo&\pamstateutab{(\la\vartwo{\la\var{\var\vartwo}})\mathsf{I}} 
	{\ctxhole(\la\varthree\varthree)} 
	{\epsilon}{0}{\pos_\var\cons\resm} &\upp\\
	\iamuaplone&\pamstatedtab{(\la\varthree\varthree)} 
	{(\la\vartwo{\la\var{\var\vartwo}})\mathsf{I}\ctxhole} 
	 {(\pos_\var,0)}{1}{\resm}&\downp\\
	\iamdlamone&\pamstatedtab{\varthree} 
	{(\la\vartwo{\la\var{\var\vartwo}})\mathsf{I}(\la\varthree\ctxhole)}  {(\pos_\var,0)}{1}{\stempty}&\downp\\
	\iamdvar&\pamstateutab{\la\varthree\varthree} 
	{(\la\vartwo{\la\var{\var\vartwo}})\mathsf{I}\ctxhole} 
	{(\pos_\var,0)}{1}{(\varthree,(\la\vartwo{\la\var{\var\vartwo}})\mathsf{I} 
		(\la\varthree\ctxhole))}&\upp\\
	\end{array}}
$\end{center}

The jump is simulated by the \LPAM retrieving the position saved in the history at the current index, and then updating the index accordingly, \ie diminishing it by one. Intuitively, this corresponds to the ``unpacking'' made by the \LJAM in the $\iamujump$ transition. We put 
$\pos_\varthree=(\varthree,(\la\vartwo{\la\var{\var\vartwo}})\mathsf{I} 
(\la\varthree\ctxhole))$ and $\pos_\vartwo=(\vartwo,(\la\vartwo\la\var{\var\ctxhole})\mathsf{II})$.
\begin{center}${\footnotesize
	\begin{array}{l|c|c|c|c|c|c}
	&\mathsf{Sub}\mbox{-}\mathsf{term} & \mathsf{Context} & \mathsf{Hist.} 
	& \mathsf{Index} & \mathsf{Tape} & \mathsf{Dir}
	\\
	\cline{1-7}
	&\pamstateutab{\la\varthree\varthree} 
	{(\la\vartwo{\la\var{\var\vartwo}})\mathsf{I}\ctxhole} 
	{(\pos_\var,0)}{1}{\pos_\varthree}&\upp\\
	\iamujump&\pamstateutab{\var} 
	{(\la\vartwo\la\var{\ctxhole\vartwo})\mathsf{I}(\la\varthree\varthree)} 
	{(\pos_\var,0)}{0}{\pos_\varthree} & \upp\\
	\iamuaplone&\pamstatedtab{\vartwo} 
	{(\la\vartwo\la\var{\var\ctxhole})\mathsf{II}} 
	 {(\pos_\varthree,0)\cons(\pos_\var,0)}{2}{\stempty} &\downp\\
	\iamdvar&\pamstateutab{\la\vartwo\la\var{\var\vartwo}} 
	{\ctxhole\mathsf{II}} 
	{(\pos_\vartwo,0)\cons(\pos_\varthree,0)\cons(\pos_\var,0)}{0}{\pos_\vartwo}&\upp\\
	\iamuaplone&\pamstatedtab{\mathsf{I}} 
	{(\la\vartwo\la\var{\var\vartwo})\ctxhole\mathsf{I}} 
	{(\pos_\vartwo,0)\cons(\pos_\varthree,0)\cons(\pos_\var,0)}{3} {\stempty}&\downp\\
	\end{array}}
$\end{center}

\paragraph{History, Indices, and Logs} The history $\history$ essentially stores the sequence of $\tomachvar$ queries, consisting of the position of a variable needing an argument, that the \LPAM has completed, that is, for which it has found the argument. The key point is that it stores them with an index $i$. Indices are a low-level mechanism to retrieve logs, that are crumbled and shuffled all over $\history$. 

Let us explain how a log $(\pos_1,\tlog_1)\cons\ldots\cons(\pos_n,\tlog_n)$ of a reachable \LJAM state  is represented by the index $i_1$ and the history $\history$ of the corresponding \LPAM state. There are two ideas:
\begin{itemize}
\item \emph{The sequence of positions}: $\pos_1$ is in the $i_1$-th entry $\indp {\pos_1}  {i_2}$ of $\history$, 
$\pos_2$ is in the $i_2$-th entry $\indp {\pos_2} {i_3}$, and so on. 
\item \emph{The log of each position}: the log $\tlog_1$ of $\pos_1$ is represented in $\history$ (recursively 
following the same principle) starting from index $i_1-1$, the log $\tlog_2$ starting from index $i_2-1$, and so on.
\end{itemize}
\begin{figure}[t]
	\input{machines/LPAM}
	\vspace{-8pt}
	\caption{Data structures and transitions of the $\lambda$ Pointer Abstract Machine (\LPAM).}
	\label{fig:pam}
\end{figure}

\paragraph{The Bisimulation} The given explanation underlies the following definition 
of relations $\bisimtape$ ,$\bisimlog$ and $\bisimstate$ between data 
structures and states of the \LJAM and the \LPAM, 
that induce a strong bisimulation. {The intended meaning of the relation $\tlog 
\bisimlog (\history, i)$ is that the log $\tlog$ is represented in the history 
$\history$ starting from index $i$.}
\begin{definition}
	The relations $\bisimtape,\bisimlog$ and $\bisimstate$ are defined as follows.
\begin{center}
\begin{tabular}{rccc}
	\textsc{Tapes}
	&
	 $\infer{\epsilon\bisimtape\epsilon}{}$
	&
	$\infer{\resm\cons\tape_J\bisimtape\resm\cons\tape_P}{\tape_J\bisimtape\tape_P}$
	&
	$\infer{(\var,\ctx,\tlog)\cons\tape_J\bisimtape(\var,\ctx)\cons\tape_P}{\tape_J\bisimtape\tape_P}$
	\\[4pt]
	\textsc{Log-Histories}
	&
	$\infer{\epsilon\bisimlog(\history, 0)}{}$
	& 
	\multicolumn{2}{c}{
	$\infer{(\var,\ctx,\tlogtwo)\cons\tlog\bisimlog(\history,i)}{
    (\var,\ctx)=(\var_i^\history,\ctxtwo_i^\history) 
		\qquad \tlog\bisimlog(\history,\phi_\history(i))
		\qquad \tlogtwo\bisimlog(\history,i-1)}$}
	\\[4pt]
	\textsc{States}
	&
	\multicolumn{3}{c}{$
	\infer{(\tm,\ctx,\tlog,\tape_J,\pol)\bisimstate(\tm,\ctx,\history,i,\tape_P,\pol)}
	{\tape_J\bisimtape\tape_P\qquad\tlog\bisimlog(\history,i)}
	$}
	\end{tabular}
\end{center}
\end{definition}
Note that {in the second rule for $\bisimlog$ the index $i$ is $\geq 1$, and 
that} $\bisimstate$ contains all pairs of 
initial states. Note also that the (logged) positions case of $\bisimtape$ 
(rightmost rule for $\bisimtape$) the 
log $\tlog$ has no matching construct on the \LPAM side. This is why the next 
theorem is stated together with an invariant (the \emph{moreover} part), 
allowing to retrieve that log from the history. 
%We also need an \ben{easy lemma about $\bisimlog$ and its relationship with 
%log splittings (at work in transition $\tomachvar$ of the \LJAM) and history 
%extensions (transition $\tomacharg$ of the \LPAM)}.
%
%\ben{
%\begin{lemma}[Logs and histories]
%\label{l:pam-bisim-aux} % \reflemmap{pam-bisim-aux}{log}
%Let $\tlog \bisimlog (\history,i)$.
%\begin{enumerate}
%	\item \label{p:pam-bisim-aux-log}	\emph{Log splitting}:
%	if $\tlog = \tlog_n\cons\tlogtwo$ then 
%$\tlogtwo\bisimlog(\history,\phi^n_\history(i))$.
%	\item \label{p:pam-bisim-aux-history} \emph{History extension}:
%	if $\indp \pos j$ be an indexed position then $\tlog \bisimlog (\indp \pos 
%j \cons \history, i)$.
%\end{enumerate}
%\end{lemma}
%}

\begin{theorem}[$\bisimstate$ is a strong bisimulation]
	\hfill
	\begin{enumerate}
		\item for every run $\run_J: \state_{\tm}^{\text{\LJAM}} 
		\toljam^* s_J$  there exists a run $\run_P: 
\state_{\tm}^{\text{\LPAM}} \tolpam^* s_P$ such that $s_J\bisimstate s_P$ and $\size{\run_J} = \size{\run_P}$ and 
performing 
exactly the same transitions;
		\item for every run $\run_P: 
\state_{\tm}^{\text{\LPAM}} \tolpam^* s_P$ there exists a  run $\run_J: 
\state_{\tm}^\text{\LJAM} \toljam^* s_J$ such 
that $s_J\bisimstate s_P$ and $\size{\run_J} = \size{\run_P}$ and performing 
exactly the same transitions.
	\end{enumerate}
	Moreover, if $s_J = (\tm,\underline{\blue\ctx},\tlog,\tape_J,\upp) 
	\bisimstate (\tm,\underline{\blue\ctx},\history,i,\tape_P,\upp) = s_P$ and 
$(\var,\ctxtwo,\tlogtwo)$ is the unique logged position in $\tape_J$ then $\tlogtwo\bisimlog(\history,|\history|)$.
\end{theorem}

Strong bisimulations trivially preserve termination and the length of runs.

\begin{corollary}[Termination and \LPAM implementation]
$\terminates{\text{\LJAM}}\tm$ if and only if $\terminates{\text{\LPAM}}\tm$, and the 
two runs use exactly the same transitions. 
Therefore, the \LPAM implements \ccbn.
\end{corollary}

%%%%%%%%%%%%%%%%%%%%
% !TeX spellcheck = en_US
% !TEX root = main.tex
%%%%%%%%%%%%%%%%%%%%%%%%%%%%%%%%%%%%%%%%%%%%%%%%%%%%%%%%%%%%%%%%%%%%%
\section{Sequence Types}
\label{sect:types}
Here we introduce a type system that we shall use to measure the length of \LIAM runs.
%, via yet 
%another abstract machine, the $\typeIAM$, that manipulates type derivations 
%%%and is isomorphic to the \LIAM on \ccbn normalizable terms.

\paragraph{Intersections, Multi Sets, and Sequences} The framework that we adopt is the one of intersection types. As 
many recent works, we use the non-idempotent variant, where the type $A\wedge A$ is not equivalent to $A$, and which has 
 stronger ties to linear logic and time analyses, because it takes into account how many times a resource/type $A$ is 
used, and not just whether $A$ is used or not. Non-idempotent intersections are multi sets, which is why these types are 
sometimes called \emph{multi types}. Here we add a further change, we also consider \emph{non-commutative} multi types. 
Removing commutativity turns multi sets into lists, or sequences---thus, we call them \emph{sequence types}. Adopting 
sequences is an inessential tweak. Our study does not really depend on their sequential structure, we only constantly 
need to use bijections between multi sets, to describe the $\SIAM$, and these bijections are just more easily managed 
using sequences rather than multi sets. This \emph{rigid} approach has been 
already used 
fruitfully by \citet{ongrigid} and \citet{MazzaPellissierVial}.
\begin{figure}[t]
	\[
	\begin{array}{c@{\hspace{1cm}}c@{\hspace{1cm}}c}
	\infer[\tyvar]{\tjudg{\var:\mset{\ty}}{\var}{\ty}}{}
	&
	\infer[\tylam]{\tjudg{\tye}{\lambda\var.\tm}{\arr{\mty}{\ty}}} 
	{\tjudg{\tye,\var:\mty}{\tm}{\ty}}
	&
	\infer[\tylamstar]{\tjudg{\tye}{\lambda\var.\tm}{\initty}}{}
	\\[8pt]
	\multicolumn{3}{c}{\infer[\tyapp]{\tjudg{\tye\uplus \sum_{i\in\mset{1,\ldots,n}} \tyetwo_i  }{\tm\tmtwo}{\ty 
		}}{\tjudg{\tye}{\tm}{\arr{\mset{\tytwo_1,\ldots,\tytwo_n}}{\ty}}
		& \mset{\tjudg{\tyetwo_i}{\tmtwo}{\tytwo_i}}_{i\in\mset{1,\ldots,n}}}}
	\end{array}
	\]
	\vspace{-8pt}
	\caption{The sequence type system.}
	\label{fig:asstypesystem}
	\end{figure}
\paragraph{Basic Definitions} As for multi types, there are two layers of types, \emph{linear types} and \emph{sequence 
types}, mutually defined as follows.
	\begin{center}$
	\begin{array}{rrcl@{\hspace{1cm}}rrcl}
	\textsc{Linear types}&\ty,\tytwo&\grameq&\initty\grammarpipe\arr{\mty}{\ty} &
	\textsc{Sequence types}&\mty,\mtytwo&\grameq&\mset{\ty_1,\ldots,\ty_n}
	\end{array}$
	\end{center}
Since commutativity is ruled out, we have, e.g., $\mset{\ty,\tytwo} \neq 
\mset{\tytwo,\ty}$. We shall use $\mset\cdot$ as a 
generic list constructor not limited to types, thus writing $\mset{2,1,12,4}$ for a list of natural numbers, and also 
use it for lists of judgments or type derivations. Note that there is a ground 
type $\initty$, which can be thought as 
the type of normal forms, that in \ccbn are precisely abstractions. Note also that arrow (linear) types $\arr{\mty}{\ty}$ 
can have a sequence only on the left. The empty sequence is noted $\emmset$, and the concatenation of two sequences 
$\mty$ and $\mtytwo$ is noted
$\mty\uplus\mtytwo$.

Type 
judgments have the following form $\tjudg{\tye}{\tm}{\ty}$, where $\tye$ is a type environment, defined below. The 
typing rules are in \reffig{asstypesystem}, type derivations are noted $\tyd$ and we write 
$\tyd\pof\tjudg{\tye}{\tm}{\ty}$ for a type derivation $\tyd$ of ending 
judgement $\tjudg{\tye}{\tm}{\ty}$.
Type environments, ranged over by $\tye,\tyetwo$ are total maps
from variables to sequence types such that only finitely
many variables are mapped to non-empty sequence types, and we write $\tye = \var_1:\mty_1,\ldots,\var_n:\mty_n$ if 
$\dom\tye = \set{\var_1,\ldots,\var_n}$---note that type environments are commutative, what is non-commutative is only 
the sequence constructor $\mset\cdot$.
Given two type environments $\tye,\tyetwo$, the expression
$\tye\uplus\tyetwo$ stands for the type environment
assigning to every variable $\var$ the list
$\tye(\var)\uplus\tyetwo(\var)$.
A sequence $\tye_{i_1}, \ldots, \tye_{i_k}$ of type environments is also noted $\{\tye_i\}_{i\in 
\mset{i_1,\ldots,i_k}}$, or $\{\tye_i\}_{i\in I}$ with $I = \mset{i_1,\ldots,i_k}$. Moreover, we use
$\sum_{i\in I}\tye_i$ for the type environment defined as $\sum_{i\in I}\tye_i \defeq \emmset$ if $I=\emmset$, and 
$\sum_{i\in I}\tye_i\defeq \tye_{i_1} \uplus \sum_{i\in I'}\tye_i$ if $I = \mset{i_1}\uplus I'$.

In the following we use two basic properties of the type system, collected in the following straightforward lemma. One is the absence of 
weakening, and the other is a correspondence between sequence types and axioms. We write $\size\mty$ for the length of 
$\mty$ as a sequence.
\begin{lemma}[Relevance and axiom sequences]
\label{l:relevance}
If $\tyd \pof \tjudg{\tye}{\tm}{\ty}$ then $\dom\tye \subseteq \fv\tm$, thus if $\tm$ is closed then $\tye$ is empty. 
Moreover, there are exactly $\size{\tye(\var)}$ axioms typing $\var$ in $\tyd$, which appear from left to right as 
leaves of $\tyd$ (seen as an ordered tree) in the order given by $\tye(\var)= \mset{\ty_1,\ldots, \ty_k}$ and that the $i$-th axiom types $\var$ with $\ty_i$.
\end{lemma}
	
\paragraph{Characterization of Termination} It is well-known that intersection and multi types characterize \ccbn 
termination, that is, they type \emph{all} and only those $\l$-terms that 
terminate with respect to weak head reduction. If 
terms are closed, the same result smoothly holds for sequence types, as we now explain. The only point where 
non-commutativity is delicate for the characterization is in the proof of the typed substitution lemma for subject 
reduction (and the dual lemma for subject expansion), as substitution may change the order of concatenation in type 
environments. In our simple setting where terms are closed, however, the term to substitute is closed\footnote{It is 
well known that in \ccbn the substitutions  $\tm\isub\var\tmtwo$ associated to reduced $\beta$-redexes are such that 
$\tmtwo$ is closed. The term $\tm$ is of course (potentially) open, and its type derivation has a type environment 
$\tye$, but the important point here is that the type derivation of $\tmtwo$ has no type environment, so that the 
substitution does not concatenate sequence types.} and---by the relevance lemma---its type derivation comes with no type 
environment, so the order-of-concatenation problem disappears. Therefore, sequence types characterize termination in 
\ccbn too. Thus from now on we essentially identify multi and sequence types.
\begin{theorem}
	A closed term $\tm$ has weak head normal form if and only if $\ctjudg{\tm}{\initty}$.
\end{theorem}
\begin{figure}
\small
			\begin{center}
			\begin{tabular}{c|c}
			\KAM & \LIAM
			\\[3pt]
			$\begin{array}{c@{\hspace{.4cm}}c@{\hspace{.4cm}}c}
			\infer[\tyvar]{\wtjudgone{1}}{}
&
			\infer[\tylam]{\wtjudgone{w+1}}{\wtjudgone{w}}
			&
			\infer[\tylamstar]{\wtjudgone{0}}{}
\\[8pt]
\multicolumn{3}{c}{			\infer[\tyapp]{\wtjudgone{w+\sum v_i+1}}{\wtjudgone{w} & 
				\mset{\wtjudgone{v_i}}_{i\in\mset{1,\ldots,n}}}
}
			\end{array}$
			&
			$\begin{array}{c@{\hspace{.4cm}}c@{\hspace{.4cm}}c}
			\infer[\tyvar]{\wtjudgt{\mset{\ty}}{\occstar{\ty}}{}{\ty}}{}
			&
			\infer[\tylam]{\wtjudgt{\tye}{w+\occstar{\arr\mty\ty}} 
				{}{\arr{\mty}{\ty}}}
			{\wtjudgt{\tye,\mty}{w}{}{\ty}}
			&
			\infer[\tylamstar]{\wtjudgt{\tye}{0}{}{\initty}}{}
			\\[8pt]
			\multicolumn{3}{c}{\infer[\tyapp]{\wtjudgt{\tye\uplus\sum_{i\in\mset{1,\ldots,n}} \tyetwo_i  }{w+\sum 
v_i+\occstar{\ty}}{}{\ty 
					}}{\wtjudgt{\tye}{w}{}{\arr{\mset{\tytwo_1,\ldots,\tytwo_n}}{\ty}}
					& \mset{\wtjudgt{\tyetwo_i}{v_i}{}{\tytwo_i}}_{i\in\mset{1,\ldots,n}} }}
			\end{array}$
			\end{tabular}
			\end{center}
		\vspace{-8pt}
		\caption{The weight assignments 
			$\WeightTimeKAM{\cdot}$, on the left, and 
			$\WeightTimeIAM{\cdot}$, on the right.}\label{fig:aaweightKAMtime}
		\label{fig:weightskam}  
	\end{figure}
\paragraph{Sequence Types and \KAM Time} Multi types have been successfully applied in quantitative 
analyses of normalization, starting with \citet{Carvalho07,deCarvalho18} who 
used them to give a bound to the length of \KAM runs. 
De Carvalho's technique can be re-phrased and distilled as a decoration of  
type derivations with \emph{weights}, that is, cost annotations, following the scheme of \reffig{weightskam}. Please 
note that the 
weight assignment is blind to types, and thus relies only on the structure of 
the type derivation. 
%Note also that variables have weight $k$, 
%where $k$ is an external parameter. This is because the cost of a variable 
%transition depends on the size of the code, as discussed in \refsect{kam}. 
De Carvalho's result can be formulated as 
follows.
\begin{theorem}[De Carvalho]\label{thm:decarvalho}
There is a 
complete \KAM run of length $n$ from $\tm$ if 
and only if
	$\WeightTimeKAM{\tyd}=n$ for every $\tyd\pof\tjudg{}{\tm}{\initty}$.
\end{theorem}

\paragraph{Sequence Types and \LIAM Time} 
We use the same idea to capture the length of a \LIAM run. We keep the same type system but we change the weight 
assignment to typing rules. First, we 
have to define a norm on types and sequence types, counts the number of occurrences of $\initty$:
\begin{center}$
\occstar{\initty}=1\qquad
\occstar{\arr{\mty}{\ty}}=\occstar{\mty}+\occstar{\ty}\qquad
\occstar{\mset{\ty_1,\ldots,\ty_n}}=\sum_{1\leq i\leq n}\occstar{\ty_i}$
\end{center}
Then we define the weight system $\WeightTimeIAM{\cdot}$ in \reffig{weightskam}. Observe how this weight 
system is structurally very similar to $\WeightTimeKAM\cdot$, the only 
difference being the
fact that whenever the latter adds $1$ to the weight, the former adds
the number of occurrences of $\initty$ in the underlying type. The next section proves the following theorem, that is 
the \LIAM analogous of de Carvalho's theorem.
\begin{theorem}
There is a complete \LIAM run of length $n$ from $\tm$ if and only if
	$\WeightTimeIAM{\tyd}=n$ for every $\tyd\pof\tjudg{}{\tm}{\initty}$.
\end{theorem}

%%%%%%%%%%%%%%%%%%%%
% !TeX spellcheck = en_US
% !TEX root = main.tex
%%%%%%%%%%%%%%%%%%%%%%%%%%%%%%%%%%%%%%%%%%%%%%%%%%%%%%%%%%%%%%%%%%%%%
\section{The Sequence IAM}
\label{sect:SIAM}
This section introduces yet another machine, the \emph{Sequence} $\IAMold$, or \emph{\SIAM}, that mimics the \LIAM directly on top of a type derivation $\tyd$. It is the key tool used in the next section to show that the \LIAM weights on type derivations do measure the time cost of \LIAM runs.

\paragraph{\SIAM} The idea behind the \SIAM is simple but a formal definition is a technical nightmare. Let us 
explain the idea. The machine moves over a fixed type derivation $\tyd\pof\tjudg{}{\tm}{\initty}$, to be thought as the 
code. The position of the machine is expressed by an occurrence of a type judgement\footnote{A 
judgement may occur repeatedly in a derivation, which is why we talk about \emph{occurrences} of  
judgements. To avoid too many technicalities, we usually just write the judgement, leaving implicit that we refer to an 
occurrence of that judgement.}  $\ruleoc$ 
of $\tyd$. As the \LIAM, the \SIAM has two possible directions, noted $\downpt$ and $\uppt$\footnote{Type 
derivations are upside-down wrt to the term structure, then direction $\downp$ of the \LIAM becomes here $\uppt$, and 
$\upp$ is $\downpt$.}. In direction $\uppt$ the machine looks at the rule above the focused judgement, in direction 
$\downpt$ at the rule below. The only "data structure" is a type context $\tyctx$ isolating an occurrence of 
$\initty$ in the type $\ty$ of the focused judgement (occurrence) $\tjudg{\tye}{\tmtwo}{\ty}$, defined as follows (careful to not 
 confuse type contexts $\tyctx$ with type environments $\tye$):
\begin{center}$
\begin{array}{rl@{\hspace{.3cm}}|rlcccc}
\textsc{Type ctxs}&\tyctx \grameq  \ctxhole \grammarpipe \arr\mty\tyctx 
\grammarpipe 
\arr{\mtyctx}\ty
&
\textsc{Sequence ctxs}&\mtyctx\grameq\mset{\ty_1,..,\ty_k, \tyctx,\ty_{k+1}..,\ty_n}
\end{array}$
\end{center}
Summing up, a state $\state$ is a 
quadruple $(\tyd, \ruleoc, \tyctx, \pol)$. If $\ruleoc$ is in the form 
$\tjudg{\tye}{\tmtwo}{\ty}$, we often write $\state$ as 
$\tjudg{}{\tmtwo}{\tyctxp{\initty_\pol}}$, where $\tyctxp{\initty}=\ty$. In 
fact we will see soon that the type environment is not needed.
% As for the previous machines, the direction $\pol$ is usually omitted and 
%represented via colors and underlining for $\downpt$ and overlining for 
%$\uppt$.
\begin{figure}[t]
  \begin{center}\footnotesize
\begin{tabular}{c}
$\begin{array}{ccc||ccc}
	\infer{\tjudg{}{\red{\tm\tmtwo}}{\tyctxp{\initty_\uppt} (=\ty)}} 
	{\tjudg{}{\tm}{\arr{\mty}{\ty}} & \mset{\vdash}} 
	&\tomachdotone&
	\infer{\tjudg{}{\tm\tmtwo}{\ty 
		}}{\tjudg{}{\red\tm}{\arr{\mty}{\tyctxp{\initty_\uppt}}} & 
		\mset{\vdash}}
		&
	\infer{\tjudg{}{\red{\lambda\var.\tm}}{\arr{\mty} 
			{\tyctxp{\initty_\uppt}}}} 
	{\tjudg{}{\tm}{\ty (= \tyctxp{\initty})}}
	& \tomachdottwo &
	\infer{\tjudg{}{\lambda\var.\tm}{\arr{\mty}{\ty}}} 
	{\tjudg{}{\red\tm}{\tyctxp{\initty_\uppt}}}
	 \\[8pt]\hhline{======}&&&\\

	\infer{\tjudg{}{\tm\tmtwo}{\ty(= \tyctxp{\initty}) 
		}}{\tjudg{}{\blue\tm}{\arr{\mty}{\tyctxp{\initty_\downpt}}} & 
		\mset{\vdash}}
		
	&\tomachdotthree&
		\infer{\tjudg{}{\blue{\tm\tmtwo}}{\tyctxp{\initty_\downpt}}} 
		{\tjudg{}{\tm}{\arr{\mty}{\ty}} & \mset{\vdash}}
		&
		\infer{\tjudg{}{\lambda\var.\tm}{\arr{\mty}{\ty}}} 
		{\tjudg{}{\blue\tm}{\tyctxp{\initty_\downpt} (=\ty)}}
		 & \tomachdotfour &
		\infer{\tjudg{}{\blue{\lambda\var.\tm}}{\arr{\mty} 
				{\tyctxp{\initty_\downpt}}}} 
		{\tjudg{}{\tm}{\ty}}
		\\[8pt]\hhline{======}\\
		\end{array}$
		\\
	$\begin{array}{ccccccc }
	\infer*{\infer{\tjudg{}{\la\var\ctxp{\var}} 
			{\arr{\mset{\myldots\ty_i\myldots}}\tytwo}}{}}
	{\infer[i]{\tjudg{}{\red\var}{\tyctxp{\initty_\uppt}_i (= \ty_i)}}{}}   
	&\tomachvar&
	 \infer*{\infer{\tjudg{}{\blue{\la\var\ctxp{\var}}} 
			{\arr{\mset{\myldots\tyctxp{\initty_\downpt}_i\myldots}}\tytwo}}{}}
	{\infer[i]{\tjudg{}{\var}{\ty_i}}{}}
	\\[8pt]\hhline{===}\\
	
	\infer*{\infer{\tjudg{}{\red{\la\var\ctxp{\var}}} 
			{\arr{\mset{\myldots\tyctxp{\initty_\uppt}_i\myldots}}\tytwo}}{}}
	{\infer[i]{\tjudg{}{\var}{\ty_i (=\tyctxp{\initty}_i)}}{}}
	 & \tomachbttwo &
	 \infer*{\infer{\tjudg{}{\la\var\ctxp{\var}} 
			{\arr{\mset{\myldots\ty_i\myldots}}\tytwo}}{}}
	{\infer[i]{\tjudg{}{\blue\var}{\tyctxp{\initty_\downpt}_i}}{}} 
	\\[8pt]\hhline{===}\\
		\infer{\tjudg{}{\tm\tmtwo}{\tytwo}} 
		{\tjudg{}{\blue\tm}{\arr{\mset{\myldots 
						\tyctxp{\initty_\downpt}_i\myldots}}{\tytwo}}
			& \tjudgi{}{\tmtwo}{\ty_i (=\tyctxp{\initty}_i)}}
		& \tomacharg &
		\infer{\tjudg{}{\tm\tmtwo}{\tytwo}} 
		{\tjudg{}{\tm}{\arr{\mset{\myldots 
						\ty_i\myldots}}{\tytwo}}
			& \tjudgi{}{\red\tmtwo}{\tyctxp{\initty_\uppt}_i}}
		\\[8pt]\hhline{===}\\

		\infer{\tjudg{}{\tm\tmtwo}{\tytwo}} 
		{\tjudg{}{\tm}{\arr{\mset{\myldots 
						\ty_i\myldots}}{\tytwo}}
			& \tjudgi{}{\blue\tmtwo}{\tyctxp{\initty_\downpt}_i (=\ty_i)}}
		 & \tomachbtone &
		\infer{\tjudg{}{\tm\tmtwo}{\tytwo}} 
		{\tjudg{}{\red\tm}{\arr{\mset{\myldots 
						\tyctxp{\initty_\uppt}_i\myldots}}{\tytwo}}
			& \tjudgi{}{\tmtwo}{\ty_i}}
		\end{array}$
		\end{tabular}
		
		\end{center}

	\vspace{-8pt}
	\caption{The transitions of the Sequence $\IAMold$ (\SIAM).}\label{fig:assweightIAMtime} 
	\label{fig:multiiam}  
\end{figure}
\paragraph{Transitions} The \SIAM starts on the final judgement of $\tyd$, with empty type context $\tyctx = \ctxhole$ 
and direction $\uppt$. It moves from judgement to judgement, following occurrences of $\initty$ around $\tyd$. The 
transitions are in \reffig{multiiam}, their union noted $\tosiam$, as we now 
explain them---the transitions have the 
 labels of \LIAM transitions, because they correspond to each other, as we shall show. 

Let's start with the simplest, $\tomachdottwo$. The state focusses on the conclusion judgement $\ruleoc$ of a $\tylam$ 
rule with direction $\uppt$. The eventual type environment $\tye$ is omitted because the transition does not depend 
on it---none of the transitions does, so type environments are omitted from all transitions. The judgement assigns  
type $\arr\mty\ty$ to $\la\var\tm$, and the type context is $\arr\mty\tyctx$, that is, it selects an occurrence of 
$\initty$ in the target type $\ty = \tyctxp\initty$. The transition then simply moves to the judgement above, stripping 
down the type context to $\tyctx$, and keeping the same direction. Transition $\resm 4$ does the opposite move, in 
direction $\downpt$, and transitions $\resm 1$ and $\resm 3$ behave similarly on $\tyapp$ rules: 
$\mset\vdash$ simply denotes the right premise that is left unspecified since not relevant to the transition.

Transitions $\tomacharg$: the focus is on the left premise of a $\tyapp$ rule, of type $\arr\mty{\tytwo}$ isolating $\initty$ 
inside the $i$-th type $\ty$ in $\mty$. The transition then moves to the final judgement of the $i$-th derivation in 
the right premise, changing direction. Transition $\tomachbtone$ does the opposite move.

Transitions $\tomachvar$ and $\tomachbttwo$ are based on the axiom sequences property of  \reflemma{relevance}. Consider 
a $\tylam$ rule occurrence whose right-hand type of the conclusion is $\arr\mty{\tytwo}$. The premise has shape 
$\tjudg{\tye, \var:\mty}{\tm}{\tytwo}$, and by the lemma there is a bijection between the sequence of linear types in 
$\mty$ and the axioms on $\var$, respecting the order in $\mty$. The left side of $\tomachbttwo$ focuses on the $i$-th 
type $\ty$ in $\mty$ and the \SIAM moves to the judgement of the axiom corresponding to that type, which is exactly 
the $i$-th from left to right seeing the derivation as a tree where the children of nodes are ordered as in the typing 
rules. Transition $\tomachvar$ does the opposite move, which can always happen because the code is the type derivation 
of a closed term. 

The only typing rule not inducing a transition is $\tylamstar$. Accordingly, when the \SIAM reaches one of these rules 
it is in a final state. Exactly as the \LIAM, the \SIAM is bi-deterministic.

\begin{proposition}
	The \SIAM is bi-deterministic for each type derivation $\tyd\pof\tjudg{}{\tm}{\initty}$.
\end{proposition}

\paragraph{An example} We present below the very same example analyzed in 
Section~\ref{sec:IJK}. We have reported its type derivation, with the 
occurrences of $\initty$ on the right of $\vdash$ annotated with increasing 
integers and a direction. The occurrence of $\initty$ marked with 1 represents 
the first state, and so on.
\begin{center}
$
\infer{\tjudg{}{(\la\vartwo{\la\var{\var\vartwo}})\mathsf{I}(\la\varthree\varthree)}
	{\initty_{\uppt{\red 1}}}}{
	\infer{\tjudg{}{(\la\vartwo{\la\var{\var\vartwo}})\mathsf{I}} 
		{\arr{\mset{\arr{\mset{\initty_{\uppt{\red 
									{13}}}}}\initty_{\downpt{\blue 
									9}}}}\initty_{\uppt{\red 
					2}}}}{
		{\infer{\tjudg{}{\la\vartwo{\la\var{\var\vartwo}}}{\arr{\mset{\initty_{\downpt{\blue
									{18}}}}}{ 
						\arr{\mset{\arr{\mset{\initty_{\uppt{\red 
												{14}}}}}\initty_{\downpt{\blue 
										8}}}}\initty_{\uppt{\red 
								3}}}}}{
				\infer{\tjudg{\vartwo:\mset\initty}{\la\var{\var\vartwo}} 
					{\arr{\mset{\arr{\mset{\initty_{\uppt{\red 
												{15}}}}}\initty_{\downpt{\blue 
												7}}}}\initty_{\uppt{\red 
								4}}}}{
					\infer{\tjudg{\vartwo:\mset\initty,\var:\mset{\arr{\mset\initty}\initty}}{\var\vartwo}
						{\initty_{\uppt{\red 5}}}}{
						\infer{\tjudg{\var:\mset{\arr{\mset\initty}\initty}} 
							{\var}{\arr{\mset{\initty_{\downpt{\blue 
												{16}}}}}\initty_{\uppt{\red 
												6}}}}{}\qquad 
						\infer{\tjudg{\vartwo:\mset\initty}{\vartwo}{\initty_{\uppt{\red
										{17}}}}}{}}}}\qquad 
			\infer{\tjudg{}{\mathsf{I}}{\initty_{\uppt{\red {19}}}}}{}}}\qquad 
	\infer{\tjudg{}{\la\varthree\varthree}{\arr{\mset{\initty_{\downpt{\blue 
							{12}}}}}\initty_{\uppt{\red 
					{10}}}}}{
		\tjudg{\varthree:\mset\initty}{\varthree}{\initty_{\uppt{\red {11}}}}}}$
\end{center}
One can immediately notice that every occurrence of $\initty$ is visited 
exactly once. Moreover, the sequence of the visited subterms is the same as the 
one obtained in the example of Section~\ref{sec:IJK}.
\section{$\lambda$-IAM Time via Exhausting Sequence Types}
\label{sect:typed-invariant}
The aim of this section is to explain the strong bisimulation between the \SIAM and the 
\LIAM, that, once again, is based on a variation on the exhaustible invariant. 
A striking point of the \SIAM is that it does not have the log nor the tape. 
They are encoded in the 
judgement occurrence $\ruleoc$ and in the type context $\tyctx$ of its states, as we shall show. But first, let's make 
a step back.

\paragraph{Handling Duplications} $\beta$-reducing a $\l$-term (potentially) duplicates arguments, whose different 
copies may be used differently, typically being applied to different further arguments. The machines in this paper never 
duplicates arguments, but have nonetheless to distinguish different uses of a same piece of code. This is why the 
\LIAM uses \emph{logged} positions instead of simple 
positions: the log is a trace of (part of) the previous run that allows to distinguishing different uses of the position---the closures of the \KAM or the history mechanism of the \LPAM are alternatives.

The key point of multi/sequence type derivations is that duplication is explicitly accounted for, somewhat \emph{in 
advance}, by multi-set/sequences: all arguments come with as many type derivations as the times they are 
duplicated during evaluation. Note indeed that the type derivation may be way bigger than the term itself, 
while this is not possible with, say, simple types. Therefore, there is no need to resort to logs, closures, or 
histories to distinguish copies, because all copies are already there: simple positions in the type 
derivation (not in the term!) are informative enough.

In the Appendix~\ref{ap:dup} we provide the execution of the term 
$(\la\var\var\var)\mathsf{I}$, that actually 
duplicates the sub-term $\mathsf{I}$, for all the machines presented in the 
paper. 

\paragraph{Relating Logs and Tapes with Typed Positions} In the \LIAM, the log 
$\tlog=\lpos_1\cons\ldots\cons\lpos_n$ has a logged position for every argument $\tmtwo_1, \ldots, \tmtwo_n$ in which 
the position of the current state is contained. In the \SIAM this is simply given by the sub-derivations for 
$\tmtwo_1, \ldots, \tmtwo_n$ in which the current judgement occurrence $\ruleoc$ is contained---the way in which 
$\lpos_k$ identifies a copy of $\tmtwo_k$ in the \LIAM corresponds on the type derivation $\pi$ to the index $i$ of the sub-derivation (in the sequence of 
sub-derivations) typing $\tmtwo_k$ in which $\ruleoc$ is located. Note that the \LIAM manipulates the log only via 
transitions $\tomacharg$ and $\tomachbtone$, that on the \SIAM correspond exactly to entering/exiting derivations for 
arguments.

The tape, instead, contains logged positions for which the \LIAM either has not yet found the associated argument, or 
it is backtracking to. Note that the \LIAM puts logged positions on the tape via transitions $\tomachvar$ and 
$\tomachbtone$, and removes them using $\tomacharg$ and $\tomachbttwo$. By looking at \reffig{assweightIAMtime}, it is 
evident that there is a logged position on the \LIAM tape for every type sequence $\mty$ in which it lies the hole 
$\ctxhole$ of the current type context $\tyctx$ of the \SIAM.

These ideas are used to extract from every \SIAM state $\state$ a \LIAM state $\extr\state$ in a quite technical 
way. A notable point is that the extraction procedure is formally defined by means of yet another reformulation on the 
\SIAM of the exhaustible invariant, called \emph{S-exhaustibility}, relying on \emph{typed} tape and log tests built 
following the explained correspondence. For lack of space the technical development is in \refapp{types-appendix}. The 
extraction process induces a relation $\state \bisimtypes \extr\state$ that is easily proved to be a strong 
bisimulation between the \SIAM and the \LIAM.

\begin{proposition}
\label{prop:str-bisim-typed}
Let $\tm$ a closed and $\towh$-normalizable term, and $\tyd\pof\tjudg{}{\tm}{\initty}$ a type derivation. Then 
$\bisimtypes$ is a strong bisimulation between \SIAM states on $\tyd$ and \LIAM states on $\tm$.
\end{proposition}

\paragraph{Weights and the Length of \SIAM Runs via Acyclicity} {We now turn to 
the proof of the correctness of 
the weight assignment $\WeightTimeIAM{\tyd}$, that is, the fact that it correctly measures the length of \LIAM 
complete runs. While the weight assignment for the \LIAM is similar to de Carvalho's one for the \KAM, the proof of its 
correctness is completely different, and it must be, as we know explain. 

The \KAM performs an 
evaluation that essentially mimics cut-elimination and so the number of \KAM transitions to normal form is 
obtained via a refined, quantitative form of subject reduction. One may say that it is obtained in a 
\emph{step-by-step} manner. The \LIAM, instead, does not mimic subject reduction. It walks over the type derivation 
\emph{without ever changing it}, potentially passing many times over the same judgement (because of backtracking). 
Correctness of weights cannot then be obtained via a refined subject reduction property, because the reduced 
derivation gives rise to a different run, and not to a sub-run. It must instead follow from a \emph{global} analysis of 
a fixed derivation, that we now develop. The proof technique is an original contribution of this paper.
}

Weights as in $\WeightTimeIAM{\tyd}$ count the number 
of occurrences of $\initty$ in $\tyd$, and every such occurrence corresponds to 
a state of the \SIAM. Proving 
the correctness of the weight system amounts to showing that every state of the 
\SIAM is reachable, and reachable exactly once. In order to do so, we have to show that the \SIAM 
never loops on typed derivations. 

Note a subtlety: by the bisimulation with the \LIAM 
(\refprop{str-bisim-typed}) we know that the run of the \SIAM terminates, but we do not know whether it reaches all 
states. What we have to prove, then, is that there are no unreachable loops, that is, loops that are not reachable 
from an initial state. The next easy lemma guarantees that this is enough.

\begin{lemma}
	Let $\tsys{}$ be an acyclic bi-deterministic transition system on a finite set of states $\mathcal{S}$ and with only 
one 	initial state $\state_i$. Then all states in $\mathcal{S}$ are reachable from $\state_i$, and reachable only once.
\end{lemma}

We show the absence of loops using a sort of subject reduction property. We first show that if 
the \SIAM loops on $\tyd\pof\tjudg{}{\tm}{\initty}$ and $\tm \towh\tmtwo$, then 
there is a type derivation $\tydtwo\pof\tjudg{}{\tmtwo}{\initty}$ on which the \SIAM loops---that 
is, \SIAM looping is preserved by reduction of the underlying term. This is done by defining a relation $\relf$ 
between the \SIAM states on $\tyd$ and on $\tydtwo$---see Appendix~\ref{ssect:acyclic-app}. 

\begin{proposition}
	$\relf$ is a loop-preserving bisimulation between \SIAM states.
\end{proposition}

Then, by the trivial fact that the 
\SIAM does not loop on $\towh$-normal terms (as they are typed using just one rule, namely $\tylamstar$), we obtain 
that it never loops.

\begin{corollary}
	Let $\tyd\pof\tjudg{}{\tm}{\initty}$ be a type derivation. Then the \SIAM does not loop on $\tyd$.
\end{corollary}

The correctness of the weights for the length of \SIAM runs immediately follows, and, via the strong bisimulation 
in \refprop{str-bisim-typed}, it transfers to the \LIAM.

\begin{theorem}[\LIAM time via sequence types]\label{thm:iamtimetypes}
	Let $\tm$ be a closed term that is $\towh$-normalizable, $\runtwo$ the 
	complete \LIAM run from $\state_\tm$, and 
$\tyd\pof\tjudg{}{\tm}{\initty}$ a type derivation for $\tm$. Then $\size\runtwo = \WeightTimeIAM{\tyd}$.
\end{theorem}

%%%%%%%%%%%%%%%%%%%%
% !TeX spellcheck = en_US
% !TEX root = main.tex
%%%%%%%%%%%%%%%%%%%%%%%%%%%%%%%%%%%%%%%%%%%%%%%%%%%%%%%%%%%%%%%%%%%%%
\section{Our Two Cents About Space}
\label{sect:space}
Here we provide an interesting example about space usage, with the only purpose of stressing that the situation is 
subtler than for time. Among the machines we have presented, the \LIAM is the 
only one tuned 
for space efficiency, as shown by the literature 
\cite{ghica_geometry_2007,DBLP:journals/entcs/GhicaS10,bllspace, 
DBLP:conf/esop/LagoS10,DBLP:conf/csl/Mazza15,DBLP:conf/icalp/MazzaT15}. 
In fact, the space used by the \LJAM (thus the \LPAM) and the KAM is 
proportional to their time, \ie their space usage is inflationary. Nonetheless, 
there are terms 
for which  the \LJAM outperforms the \LIAM in space consumption, showing that 
the space relationship between the 
\LIAM and the \LJAM is less smooth than the time one.
\begin{proposition}
	Let $\tmthree_k^h$ be defined as 
	$\tmthree_k^h\defeq (\l{\var_1}...\l{\var_k}.\la\vartwo\vartwo 
	(\l{\varthree_1}...\l{\varthree_h}.\la\varthree\varthree))\tm_1...\tm_k(\la w{w\tmtwo_1\,...\tmtwo_h})$. The the 
\LIAM space 
	consumption for the evaluation of $\tmthree_k^h$ is 2 
	logged positions plus $h+k$ occurrences of $\resm$, while the \LJAM needs 
	2 logged positions plus $\max\{h,k+1\}$ occurrences of	$\resm$.
\end{proposition}
%Concretely, the family of terms $\tm_n$ defined as $\tm_1=\mathsf{I}$ and 
%$\tm_{n+1}=\tm_n\tm_n$ has space consumption which is linear in $n$ on the 
%\LIAM, while it is exponential on both the \LJAM and the \KAM. 
%
% Moreover, comparing the \LJAM and \KAM is also not straightforward. They are both 
% inflationary in their space consumption, but memory is allocated at different 
% times. The \LJAM creates logged positions when it hits \emph{variables}, while the 
% \KAM creates closures  when \emph{applications} are encountered. Estimating the 
% space consumption of the two machines on normalizing terms is then easy, since 
% the number of variables encountered is $\bigo{\#\beta^2}$, while applications 
% are $\bigo{\#\beta}$. However, this upper bounds cannot translate into
% any result about the \emph{relative} behaviour on concrete terms. 
% 
% 

%%%%%%%%%%%%%%%%%%%%
% !TeX spellcheck = en_US
% !TEX root = main.tex
%%%%%%%%%%%%%%%%%%%%%%%%%%%%%%%%%%%%%%%%%%%%%%%%%%%%%%%%%%%%%%%%%%%%%
\section{Conclusions}
\label{sect:conclusions}
In this paper, we analysed the relative time performances of three
game machines, namely the IAM the JAM and the PAM, establishing a series
of results which can be summarized as follows:
\begin{center}
  \includegraphics[scale=1]{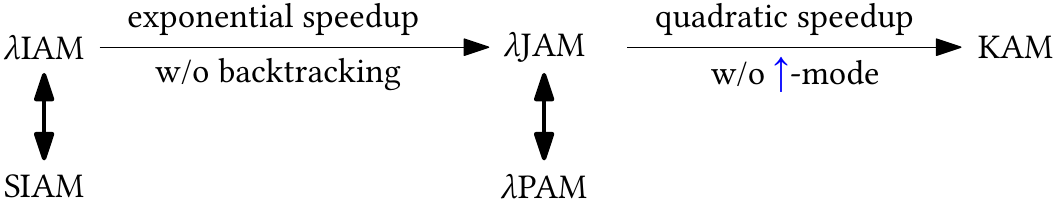}
\end{center}
Here, thicker arrows represent a stronger correspondence, i.e.,
the \LPAM and \LJAM are isomorphic, the \LJAM improves
the \LIAM with a possibly exponential advantage, while the
KAM improves the \LJAM (thus the \LPAM) with a quadratic
advantage.

Besides settling the question about the relative efficiency of the
main game machines, we also prove non-idempotent intersection types to
be able to precisely characterize the time performance of the \LIAM
when run on the typed term, in analogy with the classic results on
environment machines by~\citet{deCarvalho18}.  This way, the time
behavior of two heterogeneous machines, namely the KAM and the
\LIAM, on a given normalizing term $\tm$ can be captured by just
comparing two different ways of weighting the same sequence type derivation, 
the former attributing weight $1$ to any instance rule in the
type derivation, the latter taking into account the size of the underlying
type in an essential way. In other words, \emph{the bigger the types,
  the more inefficient the \LIAM.}

 Among the topics for future work, we can certainly mention the extension
 of the results obtained here to \emph{call-by-value} game machines, which
 seems within reach. A study on the relative \emph{space} efficiency
 of game machines is more elusive, as the partial results in Section
 \ref{sect:space} show.

%%%%%%%%%%%%%%%%%%%%

%% Acknowledgments
\begin{acks}                            %% acks environment is optional
                                        %% contents suppressed with 'anonymous'
  %% Commands \grantsponsor{<sponsorID>}{<name>}{<url>} and
  %% \grantnum[<url>]{<sponsorID>}{<number>} should be used to
  %% acknowledge financial support and will be used by metadata
  %% extraction tools.
  The second author is funded by the ERC CoG ``DIAPASoN'' (GA
  818616). This work has been partially funded by the ANR JCJC grant
  ``COCA HOLA'' (ANR-16-CE40-004-01).
\end{acks}

%% Bibliography
\bibliography{main}

%% Appendix
\appendix
\newpage
% !TeX spellcheck = en_US
% !TEX root = main.tex
% !TeX spellcheck = en_US
% !TEX root = main.tex

% !TeX spellcheck = en_US
% !TEX root = main.tex
%%%%%%%%%%%%%%%%
\section{Proof of \refsect{iam-jam} ($\lambda$-IAM and the $\lambda$-JAM: Jumping is Exhausting)}
\label{sect:iam-jam-app}
First of all, we have to finally define log tests.

\paragraph{Log Tests} We need to define the log test focussing on the $m$-th \trpos $\lpos_m$ in the log of a state $ 
\state = \nopolstate{\tm}{\ctx_n}{\tape}{\lpos_n\cdots \lpos_2\cdot\lpos_1}{\pol}$. We 
remove the prefix $\lpos_n\cdots \lpos_{m+1}$ (if any), and move the current 
position up by $n-m$ levels. Moreover, the tape is emptied and 
the direction is set to $\upp$. Let us define the (technical) position change.

Let $(\tmtwo,\lctx{n+1})$ be a position. Then, for 
every
decomposition of $n$ into two natural numbers $m,k$ with
$m+k=n$, we can find contexts $\lctx{m}$ and $\lctx{k}$ such that $\tm = \lctxp{m}{\tmthree \lctxp{k}\tmtwo}$.
  Then, the \emph{$m+1$-outer context} of the position $(\tmtwo,\lctx{n+1})$
  is the context $\octx{m+1} \defeq\lctxp{m}{r\ctxhole}$ 
  of level $m+1$ and the \emph{$m+1$-outer position} is
  $(\lctxp{k}\tmtwo,\octx{m+1})$.

  Note that the $m$-outer context and the $m$-outer position (of a given position) have level $m$.
It is easy to realize that any position having level $n$ has \emph{unique}
$m$-outer context and $m$-outer position, for every $1\leq m\leq n+1$, and that, moreover, outer positions are 
hereditary, in the following sense: the $i$-outer position of the $m$-outer position of $(\tmtwo,\lctx{n+1})$ is 
exactly 
the $i$-outer position of $(\tmtwo,\lctx{n+1})$.
\begin{definition}[Log tests]
  Let $\state =
  \nopolstate{\tm}{\ctx_n}{\tape}{\lpos_n\cdots \lpos_2\cdot\lpos_1}{\pol}$ be a
   state with $1\leq m\leq n$, and $(\tmtwo,\octx{m})$ be the $m$-outer 
  position of
  $(\tm,\lctx{n})$. The \emph{$m$-log test  of $\state$ of focus $\lpos_m$} is 
  $\state_{\lpos_m}\defeq\nopolstate{\tmtwo}{\underline{\blue{\octx{m}}}}{\stempty}{\lpos_m\cdots
   \lpos_2\cdot \lpos_1}{\upp}$.
\end{definition}

We also need to recall a lemma about log tests from \cite{IamPPDPtoAppear}, to be used in the proof of the I-exhaustible invariant.

\begin{lemma}[Invariance properties of log tests]
  \label{l:outer-inv} % \reflemmap{outer-inv}{two}

  Let $\state = \nopolstate{\tm}{\ctx_n}{\tape}{\tlog_n}{\pol}$ be a \LIAM
   state. Then:
  \begin{varenumerate}
  \item \label{p:outer-inv-one} \emph{Direction}: the dual
    $\nopolstate{\tm}{\ctx_n}{\tape}{\tlog_n}{\pol^1}$ of $\state$
    induces the same log tests;
  \item \label{p:outer-inv-two} \emph{Tape}: the state
    $\nopolstate{\tm}{\ctx_n}{\tapetwo}{\tlog_n}{\pol}$ obtained from
    $\state$ replacing $\tape$ with an arbitrary tape $\tapetwo$
    induces the same log tests;
  \item \label{p:outer-inv-three} \emph{Head translation}: if $\tm =
    \hctxp\tmthree$ then the head translation\\
    $\nopolstate{\tmthree}{\ctx_n\ctxholep\hctx}{\tapetwo}{\tlog_n}{\pol}$
    of $\state$ induces the same log tests.
  \item \label{p:outer-inv-four} \emph{Inclusion}: if $\ctx_n = 
  \ctx_m\ctxholep{\ctx_i}$ and $\tlog_n = \tlog_i\cdot \tlog_m$
  then the log tests of 
  $\nopolstate{\ctx_i\ctxholep\tm}{\ctx_m}{\tapetwo}{\tlog_m}{\pol}$
  are log tests of $\state$.
  \end{varenumerate}
\end{lemma}

\begin{lemma}[I-exhaustible invariant]
\label{l:good-invariant-jam-app}
	Let $\tm$ be a closed term and $\run: \state_\tm \toljam^k 
\state$ a \LJAM run. Then $\state$ is I-exhaustible. 
\end{lemma}

% !TeX spellcheck = en_US
% !TEX root = ../main.tex
\begin{proof}
	By
	induction on $k$. For $k=0$ there is nothing to prove because the tape has no logged positions (so it does not decompose) and $\state$ has no outer
	state. Then suppose
	$\run':\state_0\toliam^{k-1}\statetwo$ and that the run continues with $\statetwo\toljam\state$. By \ih, $\statetwo$ is 
I-exhaustible.

\emph{Terminology}: when a test state satisfies the clause in the definition of I-exhaustible states we say that it is 
\emph{positive}. 

	There are many cases to take into account, depending on the transition used 
	to 	move from $\statetwo$ to $\state$---the cases for $\resm$ are rather trivial, the other ones instead are subtle, 
	the subtlest one being the jump, that is, transition $\iamujump$ (it is the last case). First, suppose that $\pol = 
\downp$. Cases of 
	$\statetwo\toliam\state$:
	\begin{enumerate}
		\item Transition $\tomachdotone$:		
		$$
		\statetwo = \dstate{ \tmthree\tmfour }{ \ctx }{ \tape }{ \tlog } 
		\iamdap 
		\dstate{ \tmthree }{ \ctxp{\ctxhole\tmfour} }{ \resm\cons \tape }{ \tlog 
		} = \state.
		$$		
		\begin{itemize}
		\item \emph{Log tests}. Positivity of log tests follows from	
		\reflemmap{outer-inv}{three} and the \ih: $\statetwo$ is a head 
		translation of $\state$, and the	lemma states that they have the 
		same log tests, which are positive because $\statetwo$ is I-exhaustible 
		by \ih
		
		\item \emph{Tape tests}. The direction is $\downp$ and by \reflemma{jam-simple-tape}, the tape of  $\state$ has no logged positions, and so there are not tape tests.
		\end{itemize}
		 
		% ABSTRACTION
		\item Transition $\tomachdottwo$:
			$$\statetwo = \dstate{ \la\var\tmthree }{ \ctx }{ \resm\cdot\tape }{ \tlog } 
				\iamdlamone 
				\dstate{ \tmthree }{ \ctxp{\la\var\ctxhole} }{ \tape }{ \tlog } = \state$$
		Exactly as the previous case.

		% VARIABLE DOWN
		\item Transition $\tomachvar$:
		$$
		\statetwo = \dstate{ \var }{ \ctxp{\l\var.\ctxtwo_n} }{ \tape }{ 
		\tlog_n\cdot\tlog } 
		\iamdvar 
		\ustate{ \l\var.\ctxtwo_n\ctxholep\var}{ \ctx }{ 
		(\var,\ctxp{\l\var.\ctxtwo_n},\tlog_n\cdot\tlog)\cdot\tape }{ \tlog } = 
		\state
		$$			
		\begin{itemize}
		\item \emph{Log tests}. By \reflemmap{outer-inv}{four}, all 
		log tests of $\state$
		are also log tests of $\statetwo$. Since the latter is
		I-exhaustible by \ih, then all these tests are positive.
			
		\item \emph{Tape tests}. Let $\lpos\defeq (\var,\ctxp{\l\var.\ctxtwo_n},\tlog_n\cdot\tlog)$.  The only tape state 
of $\state$ is 
		$\state_{\lpos}\defeq \dstate{ \l\var.\ctxtwo_n\ctxholep\var}{ \ctx }{ 
			\lpos}{ \tlog }$ and the one-step run
		\[\begin{array}{rllll}
		\runtwo: \JAMtoIAM{\state_{\lpos}}=&
		\dstate{ \l\var.\ctxtwo_n\ctxholep\var}{ \ctx 	}{ \JAMtoIAM\lpos}{ \JAMtoIAM\tlog }
		\\
		=&
		\dstate{ \l\var.\ctxtwo_n\ctxholep\var}{ \ctx 	}{ (\var,\l\var.\ctxtwo_n,\JAMtoIAM{\tlog_n})}{ \JAMtoIAM\tlog }
		\\
		\tomachbttwodecp{\JAMtoIAM\lpos}
		&
		\ustate{ \var }{ \ctxp{\l\var.\ctxtwo_n} }{ 
			\stempty }{ 
				\JAMtoIAM{\tlog_n}\cdot\JAMtoIAM\tlog } 
		\\
		=&\JAMtoIAMstate{\ustate{ \var }{ \ctxp{\l\var.\ctxtwo_n} }{ 
			\stempty }{ 
				\tlog_n\cdot\tlog } } &=\JAMtoIAM{\indstate\lpos}
				\end{array}
		\]
		exhausts $\lpos$ as required. 		
		Now, we prove that $\indstate\lpos$ is I-exhaustible.  Note that $\indstate\lpos$ is $\statetwo$ with 
empty tape, so they have the same log tests, which are positive because $\statetwo$ is I-exhaustible by \ih, and 
$\indstate\lpos$ has no tape test.
		\end{itemize}		
	\end{enumerate}
	%%
	%% UP CASES
	%%
	Now, suppose that $\pol = \upp$. Cases of
	$\statetwo\toliam\state$:
	\begin{enumerate}
			% UP - APPLICATION LEFT WITH BULLET (->o3)
		\item Transition $\tomachdotthree$:
		$$
		\statetwo = \ustate{ \tmtwo }{ \ctxtwop{\ctxhole\tmthree} }{ \resm\cdot\tape }{ \tlog } 
		\iamuapltwo 
		\ustate{ \tmtwo\tmthree }{ \ctxtwo }{ \tape }{ \tlog } = \state.
		$$
		\begin{enumerate}
			\item	\emph{Log tests}. 	Positivity of log tests follows from \reflemmap{outer-inv}{three} and the \ih: 
$\statetwo$ is a head translation of $\state$, and the lemma states that they have the same log tests, which are 
positive because $\statetwo$ is I-exhaustible by \ih 

			\item \emph{Tape tests}. The direction of $\state$ is $\upp$ and by \reflemma{jam-simple-tape}, the tape of 
$\state$ has exactly one logged positions $\lpos$, and so just one tape test $\state_\lpos$. Note that $\statetwo$ also 
has a tape test $\statetwo_\lpos$ and that by \ih it is positive, that is, there is a run 
$\runtwo:\JAMtoIAM{\statetwo_\lpos} \toliam^*\tomachbttwodecp{\JAMtoIAM\lpos} \JAMtoIAM{\indstate\lpos}$ with 
$\indstate\lpos$ I-exhaustible. Since the direction of $\state_\lpos$ and $\statetwo_\lpos$ is 
$\downp$, we have a run $\runthree: \JAMtoIAM{\state_\lpos} \iamdap \JAMtoIAM{\statetwo_\lpos} \toliam^*\tomachbttwodecp{\JAMtoIAM\lpos}
\JAMtoIAM{\indstate\lpos}$ prefixing $\runtwo$ with a step and exhausting $\state_\lpos$.
		\end{enumerate}
		
		% UP - COMING FROM UNDER AN ABSTRACTION
		\item Transition $\tomachdotfour$
		$$
		\statetwo = \ustate{ \tmtwo }{ \ctxtwop{\la\var\ctxhole} }{ \tape }{ \tlog } 
		\iamulam 
		\ustate{ \la\var\tmtwo }{ \ctxtwo }{ \resm\cdot\tape }{ \tlog } = \state.
		$$
		This case is exactly as the previous one.
		
		% UP - APPLICATION LEFT WITH POSITION (->arg)
		\item Transition $\tomacharg$:
		$$
		\statetwo = \ustate{ \tmtwo }{ \ctxtwop{\ctxhole\tmthree} }{ 
		\lpos\cdot\tape }{ \tlog } 
		\iamuaplone 
		\dstate{ \tmthree }{ \ctxtwop{\tmtwo\ctxhole} }{ \tape }{ 
		\lpos\cdot\tlog } = \state.
		$$
		
		\begin{enumerate}
			\item
			\emph{Log tests}. The log tests of
			$\state$ are those of $\statetwo$ plus
			$\state_\lpos=\ustate{ \tmthree }{ \ctxtwop{\tmtwo\ctxhole} }{ 
			\epsilon}{ 
			\lpos\cdot\tlog }$.
			The former are positive because of the \ih, while
			about the latter, observe that
			\begin{equation}
			\label{eq:step-in-proof-JAMtoIAM}
			\begin{array}{rll}
				\JAMtoIAM{\state_\lpos}=
				&
				\JAMtoIAMstate{\ustate{ \tmthree }{ \ctxtwop{\tmtwo\ctxhole} }{ 	\epsilon}{ \lpos\cdot\tlog }}
				\\
				=&
				\ustate{ \tmthree }{ \ctxtwop{\tmtwo\ctxhole} }{ 	\epsilon}{ \JAMtoIAM\lpos\cdot\JAMtoIAM\tlog }
				\\
				\tomachbtone&
				\dstate{ \tmtwo }{ \ctxtwop{\ctxhole\tmthree} }{ \JAMtoIAM\lpos }{ \JAMtoIAM\tlog }
				\\
				=&
				\JAMtoIAMstate{\dstate{ \tmtwo }{ \ctxtwop{\ctxhole\tmthree} }{ \lpos }{ \tlog }}=\statetwo_\lpos.
				\end{array}				
		\end{equation}
			Note that $\statetwo_\lpos $ is a tape test of $\statetwo$. By \ih, there is a run $\runtwo: 
\JAMtoIAM{\statetwo_\lpos}\toliam^*\tomachbttwodecp{\JAMtoIAM\lpos} \JAMtoIAM{\indstate\lpos}$ such that $\indstate\lpos$ is I-exhaustible. Now, the run 
for the test of interest is $\runthree: \JAMtoIAM{\state_\lpos}\toliam\JAMtoIAM{\statetwo_\lpos}\toliam^*\tomachbttwodecp{\JAMtoIAM\lpos}
\JAMtoIAM{\indstate\lpos}$, obtained by prefixing $\runtwo$ with the step in \refeq{step-in-proof-JAMtoIAM}. 
			
			\item \emph{Tape tests}. The direction is $\downp$ of $\state$ and by \reflemma{jam-simple-tape}, the tape of  
$\state$ has no logged positions, and so there are not tape tests for of $\state$.
		\end{enumerate}
		
		% UP - JUMP - FROM APPLICATION RIGHT
		\item Transition $\tomachjump$:
		$$    \statetwo = \ustate{ \tmtwo }{ \ctxtwop{\tmthree\ctxhole} }{ 
		\tape }{ (\var,\ctx,\tlogtwo)\cdot\tlog }
		\iamujump 
		\ustate{ \var }{ \ctx }{ \tape }{ 
		\tlogtwo } = \state.
		$$	
		Let $\lpos \defeq (\var,\ctx,\tlogtwo)$.
		\begin{enumerate}
			\item	\emph{Log tests}:	By \ih, $\statetwo$ is 
I-exhaustible, and since $\statetwo_{\lpos}=\ustate{ 	\tmtwo }{ \ctxtwop{\tmthree\ctxhole} }{ 				
\epsilon }{ (\var,\ctx,\tlogtwo)\cdot\tlog }$ is a log test of 
				$\statetwo$, then it is positive and there exist a run
				\[\begin{array}{rllll}
					\runtwo:\JAMtoIAM{\statetwo_{\lpos}}&					
					\toliam^*\tomachbttwodecp{\JAMtoIAM\lpos}&
					\JAMtoIAM{\indstate{\lpos}} 
				\end{array}
			\]
			where $\indstate\lpos$ is I-exhaustible. 
By \reflemmap{outer-inv}{two}, 
				$\state$ and $\indstate{\lpos}$ have the same log tests, which are then positive.

			\item \emph{Tape tests}. 
			Since the direction of $\state$ is $\upp$, 
			by \reflemma{jam-simple-tape} $\sizee{\tape}=1$, there is 
			only one possible decomposition: $\tape=\tapetwo\cdot \lpos\cdot \tapethree$. Then 
			the only tape test of $\state$ is
			\[
			\state_{\lpos}=\dstate{\var}{\ctx}{\tapetwo\cdot \lpos}{\tlogtwo}
			\]
			and the only tape test of $\statetwo$ is
			\[
			\statetwo_{\lpos}=\dstate{ \tmtwo }{ \ctxtwop{\tmthree\ctxhole} }{ \tapetwo\cdot \lpos }{ 
(\var,\ctx,\tlogtwo)\cdot\tlog }
			\]
			that by \ih is positive and so there is a run $\runtwo:\JAMtoIAM{\statetwo_\lpos}\toliam^*\tomachbttwodecp{\JAMtoIAM\lpos} 
\JAMtoIAM{\indstate\lpos}$ with $\indstate\lpos$ I-exhaustible. 
			
			Now, we show that $\JAMtoIAM{\state_\lpos}\toliam^+\JAMtoIAM{\statetwo_\lpos}$, that proves the positivity of the 
tape tests, using an argument analogous to the one for the log tests. Let $\lpostwo \defeq (\var,\ctx,\tlogtwo)$ and 
consider the state $\statetwo_{\lpostwo} \defeq \ustate{ \tmtwo }{ \ctxtwop{\tmthree\ctxhole} }{ \stempty }{ 
\lpostwo\cdot\tlog }$, that it is a log test of $\statetwo$. By \ih, it is positive, thus there is a run 
$\runthree:\JAMtoIAM{\statetwo_{\lpostwo}} 
\toliam^*\tomachbttwodecp{\JAMtoIAM{\lpostwo}} \JAMtoIAM{\indstate\lpostwo}$. 
By {reversibility},  we obtain 
a run $\runthree': \JAMtoIAM{\indstate\lpostwo}^\bot \toliam^+\JAMtoIAM{\statetwo_{\lpostwo}}^\bot$, where $\cdot^\bot$ 
is the operation on states that changes the direction. Explicitly, we have:
			\[
			\runtwo':\JAMtoIAM{\indstate\lpostwo}^\bot=\JAMtoIAMstate{\dstate{ \var }{ \ctx }{ \stempty }{ \tlogtwo }}
			\toliam^+
			\JAMtoIAMstate{\dstate{ \tmtwo }{ \ctxtwop{\tmthree\ctxhole} }{ 	\stempty }{ (\var,\ctx,\tlogtwo)\cdot\tlog }}
			=\JAMtoIAM{\statetwo_{\lpostwo}}^\bot
		\]
		By \reflemma{iam-pumping}, we can lift the run to states extended with the tape $\tapetwo\cdot \lpos$, obtaining:
		\[
			\runthree'':\JAMtoIAM{\state_\lpos}
			=\JAMtoIAMstate{\dstate{ \var }{ \ctx }{ \tapetwo\cdot \lpos }{ \tlogtwo }}
			\toliam^+
			\JAMtoIAMstate{\dstate{ \tmtwo }{ \ctxtwop{\tmthree\ctxhole} }{ 	\tapetwo\cdot \lpos }{ (\var,\ctx,\tlogtwo)\cdot\tlog }}
			=\JAMtoIAM{\statetwo_\lpos}
		\]
				
The run $\JAMtoIAM{\state_\lpos}\toliam^+ \JAMtoIAM{\indstate\lpos}$ obtained by concatenating $\runthree''$ and 
$\runtwo$ exhausts $\state_\lpos$.
		\end{enumerate}
		
\end{enumerate}
\end{proof}

\begin{theorem}[\LIAM and \LJAM relationship]
\label{thm:ij-relationship}
\hfill
\begin{enumerate}
	\item \emph{\LJAM to \LIAM}: for every \LJAM run $\run_J: 
	\state_\tm^{\text{\LJAM}} 
	\toljam^* \state$ there exists a \LIAM run\\ $\JAMtoIAM{\run_J}: 	
	\JAMtoIAM{\state_\tm^\text{\LJAM}} \toliam^* \JAMtoIAM\state$ such that 
	$\size{\run_I} \geq 
\size{\run_J}$ and $\sizevar{\run_I} \geq 
\sizevar{\run_J}$.
	
	\item \emph{\LIAM to \LJAM}: for every \LIAM run $\run_I: 
	\state_\tm^\text{\LIAM} \toliam^* \state$ there exist a 
\LJAM run\\ $\run_J: \state_\tm^\text{\LJAM} \toljam^* \statetwo$ and a \LIAM run 
$\runtwo_I:\state 
\toliam^*\JAMtoIAM\statetwo$ such that $\run_I \runtwo_I = \JAMtoIAM{\run_J}$.
	
	\item \emph{Termination and \LJAM implementation}: 
	$\terminates{\text{\LIAM}}\tm$ if and only if 
	$\terminates{\text{\LJAM}}\tm$. 
Therefore, the \LJAM implements \ccbn.
\end{enumerate}
\end{theorem}

\begin{proof}\hfill
	\begin{enumerate}
	\item We proceed by induction on the length of $\run_J$. If $\size{\run_J}=0$ 
	there is nothing to prove. Now, let us consider 
	$\run_J:\state_\tm\toljam^*\statetwo\toljam\state$. Considering the property 
	true for the reduction $\runtwo_J:\state_\tm\toljam^*\statetwo$, we prove that 
	it is 
	true for $\run_J$. In particular, there exists a reduction 
	$\runtwo_I:\JAMtoIAM{\state_\tm}\toliam^*\JAMtoIAM\statetwo$ such that 
	$\size{\runtwo_I} \geq \size{\runtwo_J}$ and $\sizevar{\runtwo_I} \geq \sizevar{\runtwo_J}$ . We 
	proceed 
	considering all the possible transitions from 
	$\statetwo$ to $\state$.
	\begin{itemize}
		%%%%%%%%%%%%%%%%%%
		\item Transitions $\tomachdotone,\tomachdottwo,\tomachdotthree,\tomachdotfour, 
		\tomacharg$. This group of transitions behaves identically, modulo 
		$\JAMtoIAM\cdot$ in the two machines. Then $\size{\run_J}=1+\size{\runtwo_J}\leq_\ih
			1+\size{\runtwo_I}=\size{\run_I}$ and $\sizevar{\run_J}=\sizevar{\runtwo_J}\leq_\ih
			\sizevar{\runtwo_I}=\sizevar{\run_I}$.

		%%%%%%%%%%%%%%%%%%
		\item Transition $\iamdvar$. This transition behaves identically, 
		modulo 
		$\JAMtoIAM\cdot$, in the two machines. 
Therefore, $\size{\run_J}=1+\size{\runtwo_J}\leq_\ih
			1+\size{\runtwo_I}=\size{\run_I}$, and exactly the same sequence of (in)equalites holds with respect to 
$\sizevar\cdot$.

		%%%%%%%%%%%%%%%%%%
		\item Transition $\iamujump$. This is the only non trivial case. If $\state \tomachhole{\jumpsym,\lpos} \statetwo$ 
then by the simulation of jumps via backtracking (\reflemma{jumps-simulation}) we have a run 
$\runthree_I:\JAMtoIAM\state\tomachbtonedecp{\JAMtoIAM\lpos} 
\toliam^*\tomachbttwodecp{\JAMtoIAM\lpos}\JAMtoIAM\statetwo$. Then we define 
$\run_I$ as $\runtwo_I$ followed by 
$\runthree_I$, so that $\size{\run_J}=1+\size{\runtwo_J}\leq_\ih 1+\size{\runtwo_I}< \size{\runthree_I} + 
\size{\runtwo_I} = \size{\run_I}$ and $\sizevar{\run_J}=\sizevar{\runtwo_J}\leq_\ih \sizevar{\runtwo_I}\leq 
\sizevar{\runthree_I} + 
\sizevar{\runtwo_I} = \sizevar{\run_I}$.
	\end{itemize}

	%%%%%%%%%%%%%%%%%%	
	%% Reverse Simulation
	%%%%%%%%%%%%%%%%%%
	\item By induction on the length of $\run_I$. If $\timem{\run_I}=0$ 	there is nothing to prove. Now, let us consider 	
$\run_I:\state_\tm\toliam^*\state_1\toliam\state$. Considering the property 
	true for the reduction $\run_I':\state_\tm\toliam^*\state_1$, we prove that 
	it is true for $\run_I$. By \ih, there are 
runs $\run_J': \state_\tm \toljam^*\state_2$ and $\runtwo_I': \state_1 \toliam^* 
\JAMtoIAM{\state_2}$ such that 
$\run_I'\runtwo_I' = \JAMtoIAM{\run_J'}$. If $\runtwo_I'$ is non-empty then by determinism of the \LIAM we are done, 
because $\runtwo_I'$ has to pass through $\state$ and the suffix $\runtwo_I$ of $\runtwo_I'$ starting on $\state$ 
proves 
the statement. If $\runtwo_I'$ is empty then $\state_1 = \JAMtoIAM{\state_2}$. Then consider the cases of transition 
$\state_1\toliam\state$:
	\begin{itemize}
		%%%%%%%
		\item Transitions $\iamdap,\iamdlamone,\iamuapltwo,\iamulam, \iamuaplone, \iamdvar$. The \LJAM can do the same step 
and close the diagram, as these transitions behaves identically, modulo $\JAMtoIAM\cdot$, in the two machines. 
		%%%%%%%
		\item Transition $\iamdlamtwo$. Impossible because then the state $\state_1 = \JAMtoIAM{\state_2}$ would have 
direction $\downp$, have a logged position on the tape, and be the projection of a \LJAM state---by the direction and 
tape invariant of the \LJAM such states have no logged positions on the tape.
		%%%%%%%
		\item Transition $\iamuapr$. Then the \LJAM can make a jump and we can close the diagram using the simulation of 
jumps via backtracking (\reflemma{jumps-simulation}), as in the previous point of the theorem.	
	\end{itemize}
	
	%%%%%%%%%%%%%%%%%%
	%% Termination Equivalence
	%%%%%%%%%%%%%%%%%%
	\item The first two points of the theorem provide the proof that $I$ is a 
	bisimulation between the \LIAM and the \LJAM. Clearly, $I$ preserves 
	termination.
\end{enumerate}
\end{proof}
\refthm{ij-relationship} immediately implies the following more concise statement given in the body of the 
paper (as \refthm{ij-concise}).

\begin{corollary}[\LIAM and \LJAM relationship]
\label{coro:ij}
There is a complete \LJAM run $\run_J$ from  $\tm$ if and only if there 
is a complete \LIAM 
run $\run_I$ from $\tm$. In particular, the \LJAM implements \ccbn. Moreover, 
$\size{\run_J}\leq \size{\run_I}$ and $\sizevar{\run_J}\leq \sizevar{\run_I}$.
\end{corollary}

% !TeX spellcheck = en_US
% !TEX root = main.tex
%%%%%%%%%%%%%%
\section{Proofs of \refsect{hopping} (Hopping is Also Exhausting)}
\label{sect:hopping-app}
\begin{lemma}[\HAM basic invariants]\label{l:ham-simple-tape}
  Let $\state = \hamstatenopol{\tm}{\ctx_n}{\tlog}{\env}{\tape}{\pol}$ be a \HAM reachable state. Then 
  \begin{enumerate}
  	\item \emph{Position and log}: $(\tm,\ctx_n, \tlog)^\env$ is a closed position, and 
	\item \emph{Tape and direction}: if $\pol=\downp$, then $\tape$ does not contain any closed positions, 
    otherwise, if $\pol=\upp$, then $\tape$ contains exactly one closed position.
  \end{enumerate}
\end{lemma}

\begin{proof}  
By induction on the length of the run reaching $\state$, together with an immediate inspection of the transitions  using the \ih
\end{proof}

\paragraph{Tape Tests} By the \emph{tape and direction} invariant, there is exactly one closed position $\clpos$ on the tape in direction $\upp$ and none in direction $\downp$. Essentially, we test only the logged closures added to the tape in a $\downp$ phase, which---in $\upp$-states---are those on the right of $\clpos$ on the tape. Moreover, in a $\upp$-state we test them starting on $\clpos$ (which records a $\downp$-state), not from the current position.
\begin{definition}[\HAM Tape tests]
	Let $\state=\hamstatenopol{\tm}{\ctx}{\tlog}{\env}{\tape}{\pol}$ be a \HAM state. 
	Tape tests of $s$ are defined depending on whether there is a closed position $\clpos$ on $\tape$, that is, on 
whether $\sizeclpos{\tape}$ is $1$ or $0$.
	\begin{itemize}
		\item If $\pol= \downp$ then $\state_{\lclos} \defeq \hamstateu{\tm}{\ctx}{\tlog}{E}{\tapetwo}$ is a 
		tape test of $\state$ of focus $\lclos$ for each decomposition $\tape=\tapetwo\cons \lclos \cons\tapethree$ of the 
tape.
		\item If $\pol =\upp$ then $\state_{\lclos}\defeq \hamstateu{\var}{\ctxtwo}{\tlogtwo}{E'}{\tape_1}$ is a tape test 
		of 	$\state$ of focus $\lclos$ for each decomposition
		$\tape=\tapetwo\cons\clpos\cons\tape_1\cons\clos\cons \tape_2$ with $\clpos = (\var,\ctxtwo,\tlogtwo)^{\envtwo}$ of the tape.
	\end{itemize} 
\end{definition}

\begin{lemma}[Invariance properties of \HAM environment tests]
	\label{l:env-inv} % \reflemmap{outer-inv}{three}
	Let $\state = \hamstatenopol{\tm}{\ctx_n}{\tlog_n}{\env}{\tape}{\pol}$ be a
	state. Then:
	\begin{varenumerate}
		\item \label{p:env-inv-one} \emph{Direction}: the dual
		$\hamstatenopol{\tm}{\ctx_n}{\tlog_n}{\env}{\tape}{\pol^1}$ of $\state$
		induces the same environment tests;
		
		\item \label{p:env-inv-two} \emph{Tape}: the state
		$\hamstatenopol{\tm}{\ctx_n}{\tlog_n}{\env}{\tapetwo}{\pol}$ obtained from
		$\state$ replacing $\tape$ with an arbitrary tape $\tapetwo$
		induces the same environment tests;
		
		\item \label{p:env-inv-three} \emph{Weak shift}: let weak contexts be defined by $W\grameq\ctxhole\grammarpipe 
		W\tmtwo\grammarpipe\tmtwo W$. Then
		\begin{enumerate}
		\item  if $\tm =
		W\ctxholep\tmthree$, then for every $\tlogtwo$ and $\tapetwo$ the state
		$\hamstatenopol{\tmthree}{\ctx_n\ctxholep W}{\tlogtwo}{\env}{\tapetwo}{\pol}$
		 induces the same environment tests of $\state$.
		
		\item if $\ctx_n = \ctxtwo_h\ctxholep{\ctxthree_k}$ with $\ctxthree_k$ weak context, then 
		for every $\tlogtwo$ and $\tapetwo$ the state
		$\hamstatenopol{\ctxthree_k\ctxholep{\tm}}{\ctx_h}{\tlogtwo}{\env}{\tapetwo}{\pol}$
		induces the same environment tests of $\state$.
		\end{enumerate}
		
		\item \label{p:env-inv-four} \emph{Inclusion}: if $\ctx_n = 
		\ctx_m\ctxholep{\la\var\ctx_i}$, $\tlogn = \tlog_i\cdot \tlog_m$ and 
		$\env=\envtwo\cdot\esub\var \lclos\cdot\envthree$
		then the environment tests of 
		$\hamstatenopol{\la\var\ctx_i\ctxholep\tm}{\ctx_m}{\tlog_m}{\envthree}{\tapetwo}{\pol}$
		are environment tests of $\state$.
	\end{varenumerate}
\end{lemma}
%\begin{definition}
%	Given an environment test $s_c=\nustate{\tm} 
%	{\ctx}{\stempty}{\tlog}{E}$, where 
%	$c=(\tmthree,\ctxtwop{\tmtwo\ctxhole},\tlogtwo,\envtwo)$,
%	the state generated by $s_c$ is 
%	$\nustate{\tmtwo}{\ctxtwop{\ctxhole\tmthree}}{\stempty}{\tlogtwo}{\envtwo}$.
%\end{definition}
%
%\begin{definition}
%	Given a tape test $s_c=\nustate{\tm} 
%	{\ctx}{\tape}{\tlog}{E}$, where 
%	$c=(\tmthree,\ctxtwop{\tmtwo\ctxhole},\tlogtwo,\envtwo)$,
%	the state generated by $s_c$ is 
%	$\nustate{\tmtwo}{\ctxtwop{\ctxhole\tmthree}}{\stempty}{\tlogtwo}{\envtwo}$.
%\end{definition}

%We define a relation $\ejamtojam{\cons}$ that maps \HAM states to \LJAM 
%states.
%\[
%\begin{array}{rcl}
%\ejamtojam{\epsilon}&\defeq&\epsilon\\
%\ejamtojam{\lclos\cons\tape}&\defeq&\resm\cons\ejamtojam{\tape}\\
%\ejamtojam{\exps\cons\tlog}&\defeq&\ejamtojam{\exps}\cons\ejamtojam{\tlog}\\
%\ejamtojam{(\tm,\ctx,\tlog,E)}&\defeq&(\tm,\ctx,\ejamtojam{\tlog})\\
%\ejamtojam{\nnopolstate{\tm}{\ctx}{\tape}{\tlog}{E}{\pol}}&\defeq&
%\nopolstate{\tm}{\ctx}{\ejamtojam{\tape}}{\ejamtojam{\tlog}}{\pol}
%\end{array}
%\]
As in \refsect{jam-kam}, we consider the \LJAM and the \KAM as special 
instances of the \HAM. In particular the \LJAM always uses the 
$\tomachvarj$ 
transition, while the \KAM always the transition $\tomachvark$. This way, 
states can be 
compared without any kind of projection.

\begin{lemma}[Logged closures and closed positions were visited]
\label{l:lclosure-were-visited}
Let $\run:\state_\tm \toham^* \state$.
\begin{enumerate}
\item \emph{Logged closures}: if $\lclos = (\tmtwo,\ctxp{\tm\ctxhole},\env)^{\tlog}$ is a logged closure in $\state$ 
then $\run$ passes through a state $\hamstated{\tm}{\ctxp{\ctxhole\tmtwo}}{\tlog}{\env}{\tape}$ for some tape $\tape$.
\item \emph{Closed positions}: if $\clpos = (\tmtwo,\ctx,\tlog)^{\env}$ is a closed position in $\state$ then $\run$ 
passes through a state $\hamstated{\tmtwo}{\ctx}{\tlog}{\env}{\tape}$ for some tape $\tape$.
\end{enumerate}
\end{lemma}

\begin{proof}  
By induction on the length of $\run$, together with an immediate inspection of the transitions  using the \ih
\end{proof}

\begin{lemma}[$\upp$-exhaustible invariant]
		Let $\state$ be a \HAM reachable state. Then $\state$ is $\upp$-exhaustible.
\end{lemma}

% !TeX spellcheck = en_US
% !TEX root = ../main.tex
%%%%%%%%%%%%%%
\begin{proof}
	By 	induction on $k$. For $k=0$ there is nothing to prove because $\state=\state_\tm$ has 
	no tests. Then suppose
	$\state_\tm\toham^{k-1}\statetwo\toham\state$. By \ih\ $\statetwo
	= \hamstatenopol{\tmtwo}{\ctx}{\tlog}{\env}{\tape}{\pol}$ is $\upp$-exhaustible, and 
	with
	this hypothesis we need to conclude that $\state$ is $\upp$-exhaustible, too.
	There are many cases to take into account, depending on the transition used 
	to
	move from $\statetwo$ to $\state$.
	
	\emph{Terminology}: when a test satisfies the clause for tests in the definition of $\upp$-exhaustibility, we say 
that 
it is \emph{positive}.
	
	 First,
	suppose that $\pol = \downp$. Cases of 
	$\statetwo\toham\state$:
	\begin{enumerate}
	%%%%%%%%%%%
		\item Transition $\tomachdotoneapp$:
		\[
		\statetwo = \hamstated{ \tm\tmtwo }{ \ctx }{ \tlog } { \env }{ \tape }
		\tomachdotoneapp
		\hamstated{ \tm }{ \ctxp{\ctxhole\tmtwo} }{ \tlog }{ \env }{(\tmtwo,\ctxp{\tm\ctxhole},\env)^\tlog\cons\tape 
}=\state
		\]
		\begin{itemize}
			\item \emph{Environment tests}. It follows by the \ih, since the environment tests 
			of 
			$\state$ are the same of those of $\statetwo$.
			\item \emph{Tape tests}. The first tape test of $\state$ is 
			trivially positive since $\state_{\lclos}=\indstate{\lclos}$. All 
			other tape tests of $\state$ are in the 
			form 
			$\state_{\lclos}=\hamstateu{ \tm }{ \ctxp{\ctxhole\tmtwo} } { \tlog }{ \env } 
{(\tmtwo,\ctxp{\tm\ctxhole},\env)^\tlog\cons\tapetwo }$, 
			where $\tapetwo\cons \lclos$ is a prefix of $\tape$. Clearly
			\[
			\state_{\lclos}=\hamstateu{ \tm }{ \ctxp{\ctxhole\tmtwo} } { \tlog }{ \env } 
{(\tmtwo,\ctxp{\tm\ctxhole},\env)^\tlog\cons\tapetwo }
			\tomachdotthree
			\hamstateu{ \tm\tmtwo }{ \ctx }{ \tlog } { \env }{ \tapetwo }=\statetwo_{\lclos}
		\]
		and $\statetwo_{\lclos}$ is a tape test for $\statetwo$. By 
		\ih, $\statetwo_{\lclos}$ is positive. Hence, since
		\[
		\state_{\lclos}\tomachup\statetwo_{\lclos}
		\]
		also $\state_{\lclos}$ is positive.
	\end{itemize}
			
	%%%%%%%%%%%%%%
	\item Transition $\tomachdottwoabs$:
	\[
	\statetwo= \hamstated{ \la\var\tm }{ \ctx }{ \tlog }{ \env }{ \lclos\cons\tape }
	\tomachdottwoabs
	\hamstated{ \tm }{ \ctxp{\la\var\ctxhole} }{ \tlog } { \esub\var \lclos\cons \env } { \tape }
	= \state
	\]
	\begin{itemize}
		\item \emph{Environment tests}. The environment tests of $\state$ 
		are 
		those of $\statetwo$ plus $\state_{\lclos} \defeq \hamstateu{ \la\var\tm }{ \ctx }{ \tlog }{ \env }{	\stempty }$. 
Note that $\state_{\lclos}$ is also a tape test of $\statetwo$, which by \ih is positive.
		\item \emph{Tape tests}. Each tape test of $\state$ is in the form 
		$\state_{\lclostwo}=\hamstateu{\tm}{\ctxp{\la\var\ctxhole}}{\tlog}{\esub \var {\lclos}\cons \env}{\tapetwo}$, 
		where $\tapetwo\cons \lclostwo$ is a prefix of $\tape$. Clearly
		\[
		\state_{\lclostwo}=\hamstateu{\tm}{\ctxp{\la\var\ctxhole}}{\tlog}{\esub \var \lclos\cons	\env}{\tapetwo}
			\tomachdotfour
			\hamstateu{\la\var\tm}{\ctx}{\tlog}{\env}{\lclos\cons\tapetwo}=\statetwo_{\lclostwo}
		\]
		and $\statetwo_{\lclostwo}$ is a tape test for $\statetwo$. Thus, by 
		\ih\ $\statetwo_{\lclostwo}$ is positive, and so is $\state_{\clos'}$.
	\end{itemize}
	
	%%%%%%%%%%%%%%
	\item \emph{Transition $\tomachvarj$.}
	\[ 	\hamstated{ \var }{ \ctxp{\la\var\ctxtwo_{n}} }{ \tlogn\cons\tlog}{ \envtwo\esub\var \lclos\env }{ \tape }    
		\tomachvarj 
		\hamstateu{ \la\var\ctxtwo_{n}\ctxholep\var}{ \ctx }{ \tlog }{ \env }{ \clpos\cons\tape } \]
		where $\clpos\defeq(\var, \ctxp{\la\var\ctxtwo},\tlogn\cons\tlog)^{\envtwo\esub\var {\lclos} \env}$.
	\begin{itemize}
		\item \emph{Environment tests}. It follows by the \ih, since all the environment tests 
		of 
		$\state$ are environment tests of $\statetwo$.
		\item \emph{Tape tests}. It follows by the \ih, since the tape tests of 
		$\state$ are the same of those of $\statetwo$.		
	\end{itemize}
	
	%%%%%%%%%%%%%%%%%%
	\item \emph{Transition $\tomachvark$.}
	\[
	\statetwo=\hamstated{ \var }{ \ctx }{ \tlog}{ \env }{ \tape } 
		\tomachvark
	\hamstated{ \tmtwo}{ \ctxtwop{\tm\ctxhole} }{ (\var,\ctx,\tlog)^\env\cons \tlogtwo }{ F }{ \tape }=\state
	\]
	where 
	$\env=\envtwo\cons\esub\var{(\tmtwo,\ctxtwop{\tm\ctxhole},F)^{\tlogtwo}}\cons \envthree$.
	\begin{itemize}
		\item \emph{Environment tests}. By \reflemma{lclosure-were-visited}, we have that the run $\run$ passed through a 
state $\statethree \defeq \hamstated{ \tm }{ \ctxtwop{\ctxhole\tmtwo} }{ \tlogtwo}{ F }{ \tapetwo } $ for some 
$\tapetwo$. Note that $\statethree$ is a weak shift of $\state$ as defined in \reflemmap{outer-inv}{three}, and so 
$\statethree$ and $\state$ have the same environment tests, which are then positive by \ih

		\item \emph{Tape tests}. Note that for each prefix $\tapetwo\cons\lclos$ of $\tape$ we have
		\[
		\state_{\lclos} = \hamstateu{ \tmtwo}{ \ctxtwop{\tm\ctxhole} }{ 
				(\var,\ctx,\tlog)^\env\cons \tlogtwo }{ F }{ \tapetwo }
			\tomachjump
			\hamstateu{ \var }{ \ctx }{ \tlog}{ \env }{ \tapetwo } = \statetwo_{\lclos}
	\]
	and by \ih\ $\statetwo_{\lclos}$ is positive. Then 
	$\state_{\lclos}$ is positive.
\end{itemize}
\end{enumerate}
Then, suppose that $\pol = \upp$. Cases of $\statetwo\toham\state$:
\begin{enumerate}
	%%%%%%%%%%%%%%%%%%
	\item Transition $\tomachdotthree$:
	\[	\hamstateu{ \tm }{ \ctxp{\ctxhole\tmtwo} }{ \tlog }{ \env }{ \lclos\cons\tape }
		\tomachdotthree 
		\hamstateu{ \tm\tmtwo }{ \ctx }{ \tlog }{ \env }{ \tape }  \]
	\begin{itemize}
		\item \emph{Environment tests}. It follows by the \ih, because all the environment tests of 
		$\state$ are environment tests of $\statetwo$ by \reflemmap{outer-inv}{three}.
		\item \emph{Tape tests}. By \ih\ since the tape tests of 
		$\state$ are the same of those of $\statetwo$.
	\end{itemize}
	
	%%%%%%%%%%%%%%%%
	\item Transition $\tomachdotfour$:
		\[	\hamstateu{ \tm }{ \ctxp{\la\var\ctxhole} }{ \tlog }{ \esub\var {\lclos}\cons \env }{ \tape }    
	\tomachdotfour 
	\hamstateu{ \la\var\tm }{ \ctx }{ \tlog }{ \env }{ \lclos\cons\tape }  \]
	\begin{itemize}
		\item \emph{Environment tests}. It follows by the \ih, because all the environment tests of 
		$\state$ are environment tests of $\statetwo$.
		\item \emph{Tape tests}. By \ih\ since the tape tests of 
		$\state$ are the same of those of $\statetwo$ ($\lclos$ appears on the left of the enriched logged position in the 
tape, and so needs not to be tested).
	\end{itemize}

	%%%%%%%%%%%%%%%%%
	\item Transition $\tomacharg$:
	\[
	\statetwo=\hamstateu{ \tm }{ \ctxp{\ctxhole\tmtwo} }{ \tlog } 	{ \env } { \lpos\cons\tape }
	\tomacharg
	\hamstated{ \tmtwo }{ \ctxp{\tm\ctxhole} }{ \lpos\cons\tlog } { \env }{ \tape }
	=\state
	\]
	\begin{itemize}
		\item \emph{Environment tests}. By \ih\ since the environment tests of 
		$\state$ are the same of those of $\statetwo$ by \reflemmap{outer-inv}{three}.
		\item \emph{Tape tests}. Since the direction of $\state$ is $\downp$, by the tape and direction invariant 
(\reflemma{ham-simple-tape}) there are no closed position on $\tape$, and the tape tests of 
		$\state$ are in the form 
		$\state_{\lclos} \defeq \hamstateu{\tmtwo}{\ctxp{\tm\ctxhole}}{\clpos\cons\tlog}{\env}{\tapetwo}$,
		where $\tapetwo\cons\lclos$ is a prefix of $\tape$.
		If $\clpos=(\var,\ctxtwo,\tlogtwo)^{\envtwo}$, then
		\[
		\state_{\lclos} = \hamstateu{ \tmtwo }{ \ctxp{\tm\ctxhole} }  {(\var,\ctxtwo,\tlogtwo)^{\envtwo}\cons\tlog 
} { \env }{ \tapetwo }
		\tomachjump
		\hamstateu{ \var }{ \ctxtwo }{ \tlogtwo }{ \envtwo}{\tapetwo }.
		\]
		Those states in the form $\hamstateu{ \var }{ \ctxtwo }{ \tlogtwo }{ \envtwo}{\tapetwo }$ are exactly the tape 
tests $\statetwo_{\lclos}$ of $\statetwo$. Thus, 
		by \ih they are positive, and so are the tests $\state_{\lclos}$.
	\end{itemize}

	%%%%%%%%%%%%
	\item \emph{Jumping.} 
	\[ 		\statetwo = \hamstateu{ \tm }{ \ctxp{\tmtwo\ctxhole} }{ \clpos\cons\tlog }{ \env }{ \tape }
	\iamujump 
	\hamstateu{ \var }{ \ctxtwo }{ \tlogtwo }{ \envtwo }{ \tape } = \state\]
	where $\clpos=(\var,\ctxtwo,\tlogtwo)^{\envtwo}$.
	\begin{itemize}
		\item \emph{Environment tests}. By \reflemma{lclosure-were-visited}, the run $\run$ passes through a state 
		$\statethree \defeq \hamstated{ \var }{ \ctxtwo }{ \tlogtwo }{ \envtwo }{ \tapetwo }$ for some $\tapetwo$. Note 
that $\statethree$ and $\state$ differ only for direction and tape, and so by \reflemma{env-inv} they have the same 
environment tests, which are positive by the \ih
		\item \emph{Tape tests}. It follows by the \ih, since the tape tests of 
		$\state$ are the same of those of $\statetwo$.
	\end{itemize} %
\end{enumerate}
\end{proof}

	\begin{theorem}[\LJAM and \KAM relationship via the \HAM]
	\label{thm:jk-relationship} % \refthmp{jk-relationship}{one}
	Let $\state_\tm$ be a \HAM initial state.
\begin{enumerate}
	\item \label{p:jk-relationship-one}
	\emph{\KAM to \LJAM}: for every run $\run_K: \state_\tm 
	\tohamk^* \state$ there exists a run $J(\run_K): 	
	\state_\tm \tohamj^* \state$ such that $\size{J(\run_K)} = \size{\run_K} + \sizeup{J(\run_K)}$ and 
$\sizeparam{J(\run_K)}{\varjsym} = 
\sizeparam{\run_K}{\varksym}$.
	
	\item \emph{\LJAM to \KAM}: for every run $\run_J: \state_\tm \tohamj^* 
	\state$ there exist a run $\run_K: \state_\tm \tohamk^* \statetwo$ and a 
	run $\runtwo_J:\state \tomachup^*\statetwo$ such that $\run_J \runtwo_J = 
	J(\run_K)$.
	
	\item \emph{Termination}: $\tohamk$ terminates if and only if $\tohamj$ terminates.
\end{enumerate}
	\end{theorem}

\begin{proof}\hfill
	\begin{enumerate}
	\item We proceed by induction on the length of $\run_K$. If $\size{\run_K}=0$ 
	there is nothing to prove. Now, let us consider 
	$\run_K:\state_\tm\tohamk^*\statetwo\tohamk\state$. Considering the property 
	true for the reduction $\runtwo_K:\state_\tm\tohamk^*\statetwo$, we prove that 
	it is 
	true for $\run_K$. In particular, there exists a reduction 
	$J(\runtwo_K): {\state_\tm}\tohamj^* \statetwo$ such that 
	$\size{J(\runtwo_K)} = \size{\runtwo_K}+\sizeup{J(\runtwo_K)}$ and $\sizeparam{J(\runtwo_K)}{\varjsym} = 
\sizeparam{\runtwo_K}{\varksym}$. We 
	proceed 
	considering all the possible $\tohamk$ transitions from $\statetwo$ to $\state$.
	\begin{itemize}
		%%%%%%%%%%%%%%%%%%
		\item Transitions $\tomachdotoneapp$ and $\tomachdottwoabs$. These transitions belong also to $\tohamj$, so the 
statement trivially holds. In particular, $\size{J(\run_K)} = 1+\size{J(\runtwo_K)} =_\ih 1+\size{\runtwo_K} +
\sizeup{J(\runtwo_K)}= 
\size{\run_K}$.

		%%%%%%%%%%%%%%%%%%
		\item Transition $\tomachvark$. By \reflemma{hops-simulation}, we have a run
$\runthree:\statetwo\tomachvarj\statethree\tomachup^+\state$. Then we define $J(\run_k)$ as the concatenation of 
$J(\runtwo_k)$ and $\runthree$, for which	
$$\size{J(\run_k)} = \size{J(\runtwo_k)}+1+\size{\runthree} =_\ih \size{\runtwo_K} +\sizeup{J(\runtwo_k)}+1 + 
\size{\runthree_J} = 
\size{\runtwo_K} +1 +\sizeup{J(\run_k)}= \size{\run_K} + \sizeup{J(\run_k)}$$
and $\sizeparam{J(\run_k)}{\varjsym} = 1+ \sizeparam{J(\runtwo_k)}{\varjsym} 
=_\ih 
1+ \sizeparam{\runtwo_K}{\varksym} = \sizeparam{\run_K}{\varksym}$.
	\end{itemize}
	
	%%%%%%%%%%%%%%%%%%	
	%% Reverse Simulation
	%%%%%%%%%%%%%%%%%%
	\item By induction on the length of $\run_J$. If $\size{\run_J}=0$ 	
	there is nothing to prove. Now, let us consider 	
	$\run_J:\state_\tm\tohamj^*\state_1 \tohamj \state$. Considering the property 
	true for the reduction $\run_J':\state_\tm\tohamj^*\state_1$, we prove that 
	it is true for $\run_J$. By \ih, there are runs $\run_K': \state_\tm 
	\tohamk^*\state_2$ and $\runtwo_J': \state_1 \tomachup^* \state_2$ 
	such that $\run_J'\runtwo_J' = J(\run_K')$. If $\runtwo_J'$ is 
	non-empty then by determinism of the \LJAM we are done, because 
	$\runtwo_J'$ has to pass through $\state$ and the suffix $\runtwo_J$ of 
	$\runtwo_J'$ starting on $\state$ proves the statement. If $\runtwo_J'$ is 
	empty then $\state_1 = \state_2$ and in particular $\state_2$ has direction $\downp$, because it is reached by 
$\tohamk$. Then consider the cases of 
	transition $\state_1\toljam\state$:
	\begin{itemize}
		%%%%%%%
		\item Transitions $\tomachdotoneapp$ and $\tomachdottwoabs$. These transitions belong also to $\tohamk$, so 
$\tohamk$ can
		do the same 
		step and close the 
		diagram.
		%%%%%%%		
		\item Transition $\tomachvarj$. Then we can 
		close the diagram via the reasoning used at the previous point of the 
		theorem, based on \reflemma{hops-simulation}.	
	\end{itemize}
	
	%%%%%%%%%%%%%%%%%%
	%% Termination Equivalence
	%%%%%%%%%%%%%%%%%%
	\item Two directions:
	\begin{itemize}
	 \item \emph{$\tohamk$ termination implies $\tohamj$ termination}: an omitted standard invariant 
ensures that if terms are closed then, whenever the code is a variable $\var$, the environment is defined on $\var$. 
This fact forbids $\tohamk$ to get stuck on $\tomachvark$ transitions. So $\tohamk$ final states have the shape 
$\hamstated{\la\var\tm}\ctx\tlog\env\stempty$, which are also $\tohamj$ final states. Then if $\tohamk$ 
terminates $\tohamj$ terminates.

  \item \emph{$\tohamj$ termination implies $\tohamk$ termination}: we prove the contrapositive statement. Suppose that 
$\tohamk$ diverges starting from $\state_\tm$. Note that it has to make an infinity of $\tomachvark$ transitions, 
because without them---that is considering only $\tomachdotoneapp$ and $\tomachdottwoabs$---the size of the code 
strictly decreases. By the first point of the theorem, projecting the diverging $\tohamk$ run we obtain a diverging 
$\tohamj$ run, because the projection maps the infinity of $\tomachvark$ transitions to an infinity of $\tomachvarj$ 
transitions.
	\end{itemize}

\end{enumerate}
\end{proof}

\refthm{jk-relationship} immediately implies the following more concise statement given in the body of the 
paper (as \refthm{jk-concise}).

\begin{corollary}[\LJAM and \KAM relationship]
\label{coro:jk}
There is a complete \LJAM run $\run_J$ from 
$\tm$ if and only if 
there is a complete \KAM run $\run_K$ from $\tm$. Moreover, 
$\size{\run_J}= \size{\run_K} + \sizeup{\run_J}$ and $\sizevar{\run_J}= \sizevar{\run_K}$.
\end{corollary}

\begin{proof}
 It follows immediately from the previous theorem by the two obvious (and omitted) strong bisimulations between the 
\KAM and the transition subrelation $\tohamk$ of the \HAM, and between the \LJAM and the transition subrelation 
$\tohamj$ of the \HAM.
\end{proof}

% !TeX spellcheck = en_US
% !TEX root = main.tex
%%%%%%%%%%%%%%
\section{Proofs of \refsect{jam-complexity} (The $\lambda$-JAM is Slowly Reasonable)}
\label{sect:jam-complexity-app}

\begin{proposition}[Depth invariant]
	Let $\run:\state_\tm\toljam^*\state$ be 
	an initial run of the \LJAM. Then 
	$ \spdepth\state = \sizevar{\run}$. Moreover $\spdepth\state \geq \spdepth\lpos$ for every logged position $\lpos$ in 
$\state$.
\end{proposition}

	\begin{proof}
		We proceed by induction on the length of the run $\run$. If 
		$|\run|=0$, 
		then $\state = \state_\tm$ and $\spdepth{\state_\tm} = \spdepthnopar{\dstate{\tm}{\ctx}{\stempty}{\stempty}} = 
\spdepth\stempty = 0 =  \sizevar\run$. If $|\run|\geq 1$, let $\runtwo$ be the prefix of $\run$ such that $\state_\tm 
\toljam^* \statetwo$, and let's consider the various cases of the last transition 
		$\statetwo \toljam \state$:
		\begin{itemize}
		%%%%%%%%
		\item Transitions $\iamdap$ or $\iamdlamone$: the result holds by \ih, since the 
		polarity 	has not changed and neither the depth of the log. 

		%%%%%%%%%%%%%%%
		\item Transition $\tomachvar$:
		\[\statetwo = \dstate{ \var }{ \ctxp{\la\var\ctxtwo_n} }{ \tape }{ 
		\tlog_n\cdot\tlog 
		} 
		\iamdvar
		\ustate{ \la\var\ctxtwo_n\ctxholep\var}{ \ctx }{ 
			(\var,\ctxp{\la\var\ctxtwo_n},\tlog_n\cdot\tlog)\cdot\tape }{ \tlog 
			} = \state\]
		Then $\spdepth\state = \spdepth{\tlog_n\cdot\tlog} + 1 = 
		\spdepth{\statetwo} +1 =_{\ih} \sizevar{\runtwo} +1 = 
\sizevar\run$. For the \emph{moreover} part, let $\lpos\defeq 
(\var,\ctxp{\la\var\ctxtwo_n},\tlog_n\cdot\tlog)$ and 
consider a logged position $\lpostwo \neq \lpos$ in $\state$. By \ih $\spdepth{\lpostwo} \leq \spdepth{\statetwo} < 
\spdepth\state$. For $\lpos$, instead, by definition of $\spdepth\cdot$ we have $\spdepth\lpos = \spdepth\state$.

		%%%%%%%%%%%%%
		\item Transitions $\iamuapltwo$, $\iamulam$, and $\iamujump$: the result holds by \ih, since 
		the polarity has not changed and neither has the depth of the tape. For the \emph{moreover} part, the every logged 
position of $\state$ is in $\statetwo$, and so it follows by the \ih
		
		\item Transition $\tomacharg$: the result follows by \ih, since the depth of the tape of $\state$ is the same of the 
depth 	of 	the log of $\statetwo$.
		\[
		\state=\ustate{ \tmtwo }{ \ctxp{\ctxhole\tm} }{ \lpos\cdot\tape }{ 
			\tlog } 
		\iamuaplone
		\dstate{ \tm }{ \ctxp{\tmtwo\ctxhole} }{ \tape }{ \lpos\cdot\tlog 
		}=\statetwo
		\]
		For the \emph{moreover} part, the every logged position of $\state$ is in $\statetwo$, and so it follows by the \ih
		\end{itemize}
	\end{proof}
	
	Remember that $\tomachup \defeq \tomachhole{\resm 3, \resm 4 , \argsym, \jumpsym}$. We also set $\tomachdown \defeq 
\tomachhole{\resm 1, \resm 2, \varsym}$.

	\begin{lemma}[Bound on $\upp$ phases]
\hfill
\begin{enumerate}
\item \emph{One $\upp$ phase}: if $\state=\ustate{\tm}{\ctx}{\tape}{\tlog}$ is a reachable state and 
	$\run:\state\tomachup^*\statetwo$ then 
	$\size\run\leq \spdepth\state\cdot\size{\ctxp\tm}$.
	
	\item 
	\emph{All $\upp$ phases}:	if $\run:\state_\tm\toljam^*\state$ then 
	$\sizeup\run \leq \sizevar\run^2\cdot\size\tm$. 
\end{enumerate}
\end{lemma}
\begin{proof}\hfill
\begin{enumerate}
\item 
	We can split $\run$ in many subruns $\run_1\ldots\run_n$ consisting only of 
	$\tomachhole{\resm 3, \resm 4}$ sequences and interleaved by $\iamujump$ 
	transitions, \ie $\run=\run_1\cdot\iamujump\run_2\cdot\iamujump\cdots 
	\run_n$. By \reflemma{boundC}, each $\run_i$ is such that 
	$\size{\run_i}\leq\size{\ctxp\tm}$. Moreover, note that the log is 
	untouched by $\run_i$ and that the number of $\iamujump$ transitions is bound by the depth of the first logged 
position in $\tlog$, itself bound by $\spdepth\state$ by \reflemma{ljam-var-invariant}. Then 	$\size\run\leq 
\spdepth\state\cdot\size{\ctxp\tm}$.
	\item The run $\run$ has shape 
$\run_1^{\downp}\run_1^{\upp}\run_2^{\downp}\run_2^{\upp}\ldots\run_n^{\downp}\run_n^{\upp}$ where $\run_i^{\downp}$ is 
made out of $\tomachdown$ transitions and $\run_i^{\upp}$ is made out of $\tomachup$ transitions. By the previous point, 
we have $\size{\run_i^{\upp}} \leq \spdepth{\state_i^{\upp}}\cdot \size\tm$ where $\state_i^{\upp}$ is the source state 
of $\run_i^{\upp}$. By \reflemma{ljam-var-invariant}, $\spdepth{\state_i^{\upp}} = 
\sum_{j=1}^i\sizevar{\run_j^{\downp}}$. Now, 
	$\sizeup\run = \sum_{i=1}^n \size{\run_i^{\upp}} \leq \sum_{i=1}^n \spdepth{\state_i^{\upp}}\cdot \size\tm = \size\tm 
\cdot \sum_{i=1}^n \sum_{j=1}^i\sizevar{\run_j^{\downp}} \leq \size\tm \cdot \sum_{i=1}^n \sizevar\run \leq \size\tm 
\cdot \sizevar\run^2.$
\end{enumerate}
\end{proof}

\begin{theorem}[\LJAM complexity]
	Let $\tm$ be a closed term such that $\tm \towh^n \tmtwo$, $\tmtwo$ be 
	$\towh$ normal, and $\run_J$ and $\run_K$ be the complete \LJAM and  \KAM runs from $\tm$. Then:
	\begin{enumerate}
		\item \emph{The \LJAM is quadratically slower than the \KAM}: $\size{\run_K} \leq \size{\run_J} 
		= \bigo{\size{\run_K}^2\cdot \size\tm}$.

		\item \emph{The \LJAM is (slowly) reasonable}: $\size{\run_J} = \bigo{n^4\cdot \size\tm}$, and the cost of 
		implementing $\run_J$ on a RAM is also $\bigo{n^4\cdot 
			\size\tm}$.
	\end{enumerate}
\end{theorem}
\begin{proof}
	\begin{enumerate}
		\item
 By \refthmp{jk-relationship}{one}, $\size{\run_J} = \size{\run_K} + \sizeup{\run_J}$ and $\sizevar{\run_K} = 
\sizevar{\run_J}$. By 
\reflemmap{bound-up-phase}{global}, $\sizeup{\run_J} = \sizevar{\run_J}^2\cdot\size\tm = 
\sizevar{\run_k}^2\cdot\size\tm \leq^* \size{\run_k}^2\cdot\size\tm$, from 
which the statement follows.
\item The previous point gives $\size{\run_J} = \bigo{\size{\run_K}^2\cdot 
\size\tm}$ where $\run_K$ is the corresponding run 
on the \KAM. As recalled in \refsect{jam-complexity}, $\size{\run_K} = 
\bigo{n^2}$, from which we obtain $\size{\run_J} = 
\bigo{n^4\cdot \size\tm}$.

To obtain the cost of implementing on a RAM, we need to consider the cost of 
implementing single transitions. They all 
have constant cost but for $\tomachvar$ that costs $\size\tm$. Now note that in 
the length bound $\size{\run_J} = 
\bigo{n^4\cdot \size\tm}$ the component $\size\tm$ comes from the $\upp$ 
transitions, not $\tomachvar$, so that the 
cost on RAM is not $\bigo{n^4\cdot \size\tm^2}$ but simply $\bigo{n^4\cdot 
\size\tm}$.
\end{enumerate}
\end{proof}

% !TeX spellcheck = en_US
% !TEX root = main.tex
%%%%%%%%%%%%%%%%%%%%%%%%%%%%%%%%%%%%%%%%%%%%%%%%%%%%%%%%%%%%%%%%%%%%%
\section{Proofs of \refsect{pam} (The Pointer Abstract Machine)}

\begin{lemma}[\LPAM invariants]
    Let $\state=\pamstatenopol\tm\ctx \history i \tape\pol$ be a reachable \LPAM state. 
    Then:
  \begin{enumerate}
  	\item \emph{Depth}: $\history$ has depth $n$ at $i$. Moreover, if $((\tmtwo,\ctxtwo_m), j)$ is the $k$-th indexed 
position of $\history$, with $k>0$, then $\history$ has depth $m$ at $k-1$.
	\item \emph{Tape, index, and direction}: if $\pol=\downp$, then 
$i = \size\history$ and $\tape$ does not contain any \trpos, otherwise if $\pol=\upp$ then $\tape$ contains exactly 
one position.	
  \end{enumerate}
\end{lemma}

\begin{proof}  
By induction on the length of the run reaching $\state$, together with an immediate inspection of the transitions using the \ih
\end{proof}

%\begin{proof}
%\ben{to do}.
%\end{proof}

\begin{lemma}[Logs and histories]\label{l:pam-bisim-aux}
Let $\tlog \bisimlog (\history,i)$.
\begin{enumerate}
	\item\label{p:pam-bisim-aux-log} \emph{Log splitting}:
	if $\tlog = \tlog_n\cons\tlogtwo$ then $\tlogtwo\bisimlog(\history,\phi^n_\history(i))$.
	\item\label{p:pam-bisim-aux-history} \emph{History extension}:
	if $\indp \pos j$ be an indexed position then $\tlog \bisimlog (\indp \pos j \cons \history, i)$.
\end{enumerate}
\end{lemma}

\begin{proof}\hfill
\begin{enumerate}
\item By induction on $n$:
\begin{itemize}
 \item \emph{Base case}: if $n=0$, then it is trivially satisfied since $\tlog_n=\epsilon$ and $\phi_\history^n(i)=i$, 
so that $\tlogtwo= \tlog_n\cons\tlogtwo \bisimlog (\history, i) = (\history, \phi_\history^n(i))$, as required. 

 \item \emph{Inductive case}: if $n>0$ first of all note that $\phi^m_\history(i)$ is defined for all $m\leq n$ by the 
depth invariant (\reflemma{pam-simple-tape}). Then, $\tlog=\tlog_{n-1}\cons e\cons\tlogtwo$ and by \ih 
$e\cons\tlogtwo\bisimlog(\history, \phi_\history^{n-1}(i))$. By definition of $\bisimlog$ this is possible only if 
$\tlogtwo\bisimlog(\history, \phi_\history(\phi^{n-1}_\history(i)))$, \ie 
$\tlogtwo\bisimlog(\history, \phi^n_\history(i))$.
\end{itemize}

%%%%%%%%%
\item  By induction on $\tlog$. Cases:
 \begin{itemize}
  \item \emph{Empty}, \ie $\tlog = \stempty$. We have that the hypothesis is $\stempty \bisimlog (\history, 
0)$, because it is the only derivable relation for empty logs. Then $\stempty \bisimlog (\indp \pos j \cons \history, 
0)$.

  \item \emph{Non-empty}, \ie $\tlog = (\var,\ctx,\tlogthree)\cons \tlogtwo$. By hypothesis, $\lpos\cons \tlogtwo 
\bisimlog (\history, i)$, which implies that 
\begin{enumerate}
 \item $(\var,\ctx)=(\var_i^\history,\ctxtwo_i^\history)$,
 \item $\tlogtwo\bisimlog(\history,\phi_\history(i))$, and
 \item $\tlogthree\bisimlog(\history,i-1)$.
\end{enumerate}
By \ih, we obtain $\tlogtwo\bisimlog(\indp \pos j \cons \history,\phi_\history(i))$ and $\tlogthree\bisimlog(\indp \pos 
j\cons \history,i-1)$, and clearly $(\var,\ctx)=(\var_i^{\indp \pos j \cons \history},\ctxtwo_i^{\indp \pos j \cons 
\history})$. Then $\tlog = (\var,\ctx,\tlogthree)\cons \tlogtwo \bisimlog (\indp \pos j \cons \history, i)$.
 \end{itemize}
\end{enumerate}
\end{proof}

\begin{theorem}[$\bisimstate$ is a strong bisimulation]
	\hfill
	\begin{enumerate}
		\item for every run $\run_J: \state_{\tm}^{\text{\LJAM}} 
		\toljam^* s_J$  there exists a run $\run_P: 
\state_{\tm}^{\text{\LPAM}} \tolpam^* s_P$ such that $s_J\bisimstate s_P$ and $\size{\run_J} = \size{\run_P}$ and 
performing 
exactly the same transitions;
		\item for every run $\run_P: 
\state_{\tm}^{\text{\LPAM}} \tolpam^* s_P$ there exists a  run $\run_J: 
\state_{\tm}^\text{\LJAM} \toljam^* s_J$ such 
that $s_J\bisimstate s_P$ and $\size{\run_J} = \size{\run_P}$ and performing 
exactly the same transitions.
	\end{enumerate}
	Moreover, if $s_J = (\tm,\underline{\blue\ctx},\tlog,\tape_J,\upp) 
	\bisimstate (\tm,\underline{\blue\ctx},\history,i,\tape_P,\upp) = s_P$ and 
$(\var,\ctxtwo,\tlogtwo)$ is the unique logged position in $\tape_J$ then $\tlogtwo\bisimlog(\history,|\history|)$.
\end{theorem}

\begin{proof}
We prove the first point, the second point is symmetrical (simply replacing the use of \reflemma{pam-simple-tape}---in the case of transition $\tomachvar$ below---with \reflemma{jam-simple-tape}). By induction on $\size{\run_J}$. If $\run_J$ is empty then 
simply take $\run_P$ as the empty run. Otherwise, by \ih there exists a \LPAM run $\run_P: 
\state_{\tm}^{\text{\LPAM}} \tolpam^* s_P$ such that $s_J\bisimstate s_P$ and $\size{\run_J} = \size{\run_P}$. Note 
that 
$s_J\bisimstate s_P$ implies $s_J=\nopolstate{\tm}{\ctxthree}{\tape_J}{\tlog}{\pol}$ and 
$s_P=\pamstatenopol{\tm}{\ctxthree}{\tape_P}{\history}{i}{\pol}$ with $\tape_J \bisimtape \tape_P$ and $\tlog 
\bisimlog 
(\history, i)$.

Let's 
consider the possible extensions of $\run_J$ with a further transition. Cases:
	\begin{itemize}
		\item Transitions $\tomachdotone$,$\tomachdottwo$,$\tomachdotthree$,$\tomachdotfour$: we show one such case, the 
other are analogous. 			
			\begin{center}\begin{tikzpicture}[node distance=30mm, auto, 
				transform shape]
				\node (p) at (0,0) 
				{$\dstate {\tmtwo\tmthree} \ctxthree {\tape_J} \tlog $};
				\node (q) at (6,0) 
				{$\dstate  \tmtwo {\ctxthreep{\ctxhole\tmthree}} {\resm\cons\tape_J} \tlog $};
				\node (w) at (0,-1) 
				{$\pamstated  {\tmtwo\tmthree} \ctxthree \history i {\tape_P}$};				
				\node (r) at (6,-1) 
				{$\pamstated \tmtwo {\ctxthreep{\ctxhole\tmthree}} \history i {\resm\cons\tape_P}$};
				\node at (0,-.5) {$\bisimstate$};
				\node at (6,-.5) {$\bisimstate$};
				\node at (3,0) {$\tomachdotone$};
				\node at (3,-1) {$\tomachdotone$};				
		\end{tikzpicture}\end{center}
		For $\tomachdotthree$ and $\tomachdotfour$ the \emph{moreover} part follows from the \ih
				
		%%%%%%%%%%%%%
		%%% ->var %%%
		%%%%%%%%%%%%%
		\item Transition $\tomachvar$. We are in the following situation:		
		$$\state_J = \dstate \var {\ctxp{\la\var\ctxtwo_n}} {\tape_J} {\tlog_n\cons\tlogtwo}
		\bisimstate \pamstated \var {\ctxp{\la\var\ctxtwo_n}} \history i {\tape_P} = \state_P$$		
		with $\tape_J \bisimtape \tape_P$ and $\tlog = \tlog_n\cons\tlogtwo \bisimlog (\history, i)$.
		The \LPAM can do a $\tomachvar$ transition (guaranteed by the depth invariant of \reflemma{pam-simple-tape}), but 
we have to verify that the two target states are 
still $\bisimstate$-related. By \reflemmap{pam-bisim-aux}{log}, we have $\tlogtwo\bisimlog(\history, \phi^n_\history(i))$. 
Then:
		\begin{center}
		\begin{tikzpicture}[node distance=30mm, auto, transform shape]
		\node (p) at (0,0) {$\dstate \var {\ctxp{\la\var\ctxtwo_n}} {\tape_J} {\tlog_n\cons\tlogtwo}$};
		\node (q) at (7,0) {$\ustate {\la\var\ctxtwo_n\ctxholep{\var}} \ctx {(\var,\ctxp{\la\var\ctxtwo_n},
\tlog_n\cons\tlogtwo)\cons\tape_J} {\tlogtwo}$};
		\node (w) at (0,-1) {$\pamstated \var {\ctxp{\la\var\ctxtwo_n}} \history i {\tape_P}$}; 
		\node (r) at (7,-1) 
{$\pamstateu {\la\var\ctxtwo_n\ctxholep{\var}} \ctx \history {\phi^n_\history(i)} {(\var,\ctxp{\la\var\ctxtwo_n}
)\cons\tape_P}$ }
;
		\node at (2.8,0) {$\tomachvar$};
		\node at (2.8,-1) {$\tomachvar$};
		\node at (0,-.5) {$\bisimstate$};
		\node at (7,-.5) {$\bisimstate$};
		\end{tikzpicture}
		\end{center}
		
		Now, the \emph{moreover} part. We have to prove that $\tlog_n\cons\tlogtwo \bisimlog (\history, \size\history)$. By 
hypothesis $\tlog_n\cons\tlogtwo \bisimlog (\history, i)$ and by \reflemma{pam-simple-tape}, we have $i = 
\size\history$.
		
		%%%%%%%%%%%%%
		%%% ->arg %%%
		%%%%%%%%%%%%%
		\item Transition $\tomacharg$. We are in the following situation:	
			$$\state_J = \ustate \tm {\ctxp{\ctxhole\tmtwo}} {\lpos\cons\tape_J'} \tlog
		\bisimstate \pamstateu \tm {\ctxp{\ctxhole\tmtwo}} \history i {\pos\cons\tape_P'} = \state_P$$		
		with $\lpos\cons\tape_J'\bisimtape\pos\cons\tape_P'$ and $\tlog\bisimlog(\history, i)$. The \LPAM can do a 
$\tomacharg$ transition, but we have to verify that the two target states are still $\bisimstate$-related. Namely, 
		we have to show that 
$\lpos\cons\tlog\bisimlog(\indp\pos i \cons \history,|\history|+1)$. Let us set $\historytwo \defeq 
\indp\pos i \cons \history$ 
and $j \defeq |\historytwo|=|\history|+1$. We check that all three hypothesis of the rule defining $\bisimlog$ hold:
		\begin{enumerate}
			\item Since $\tlog\bisimlog(\history, i)$, \reflemmap{pam-bisim-aux}{history} gives $\tlog\bisimlog(\indp\pos i \cons \history, 
i)$. Note that $\phi_{\historytwo}(j)=i$, that is,
$\tlog\bisimlog(\historytwo, \phi_{\historytwo}(j))$.
			\item Since $\lpos\cons\tape_J'\bisimtape \pos\cons\tape_P'$, if $\lpos=(\var,\ctxtwo,\tlogtwo)$, then 
$\pos=(\var,\ctxtwo)$ and thus $(\var,\ctxtwo)=(\var_j^{\historytwo},\ctxtwo_j^{\historytwo})=(\var,\ctxtwo)$.
			\item By \ih, the logged position $\lpos=(\var,\ctxtwo,\tlogtwo)$ on the \LJAM tape verifies 
$\tlogtwo\bisimlog(\history,|\history|)$. By \reflemmap{pam-bisim-aux}{history}, $\tlogtwo\bisimlog(\indp\pos i \cons 
\history,|\history|)$, that is, $\tlogtwo\bisimlog(\historytwo,j-1)$.
		\end{enumerate}
		Then the two target states match:
		\begin{center}\begin{tikzpicture}[node distance=30mm, auto, transform shape]
			\node (p) at (0,0) {$\ustate \tm {\ctxp{\ctxhole\tmtwo}} {\lpos\cons\tape_J'} \tlog$};
			\node (q) at (6,0) {$\dstate \tmtwo {\ctxp{\tm\ctxhole}} {\tape_J'} {\lpos\cons\tlog}$};
			\node (w) at (0,-1){$\pamstateu \tm {\ctxp{\ctxhole\tmtwo}} \history i {\pos\cons\tape_P'}$};
			\node (r) at (6,-1) {$\pamstated \tmtwo {\ctxp{\tm\ctxhole}} {\historytwo} j {\tape_P'}$};
			\node at (3,0) {$\tomacharg$};
			\node at (3,-1) {$\tomacharg$};
			\node at (0,-.5) {$\bisimstate$};
		\node at (6,-.5) {$\bisimstate$};
			\end{tikzpicture}\end{center}		
		
		%%%%%%%%%%%%%%
		%%% ->jump %%%
		%%%%%%%%%%%%%%
		\item Transition $\tomachjump$. We are in the following 
situation:	
$$\state_J = \ustate \tm {\ctxp{\tmtwo\ctxhole}} {\tape_J} {(\var,\ctxtwo,\tlogtwo)\cons\tlogtwo}
		\bisimstate \pamstateu \tm {\ctxp{\tmtwo\ctxhole}} \history i {\tape_P} = \state_P$$					
    with $\tape_J \bisimtape \tape_P$ and $(\var,\ctxtwo,\tlogtwo)\cons\tlogtwo \bisimlog (\history,i)$. The \LPAM can 
do a $\tomachjump$ transition, but we have to verify that the two target states are still $\bisimstate$-related. Note 
that, 
		since 
$(\var,\ctxtwo,\tlogtwo)\cons\tlogtwo\bisimlog(\history, i)$, we 
have $(\var,\ctxtwo)=(\var_i^\history,\ctxtwo_i^\history)$ and $\tlogtwo\bisimlog(\history, i-1)$. Therefore:
		
		\begin{center}\begin{tikzpicture}[node distance=30mm, auto, transform shape]
			\node (p) at (0,0) {$\ustate \tm {\ctxp{\tmtwo\ctxhole}} {\tape_J} {(\var,\ctxtwo,\tlogtwo)\cons\tlogtwo}$};
			\node (q) at (6,0) {$\ustate \var \ctxtwo {\tape_J} {\tlogtwo}$};
			\node (w) at (0,-1) 
{$\pamstateu \tm {\ctxp{\tmtwo\ctxhole}} \history i {\tape_P}$};
			\node (r) at (6,-1) {$\pamstateu {\var_i^\history} {\ctxtwo_i^\history} \history {i-1} {\tape_P}$};

			\node at (3,0) {$\tomachjump$};
			\node at (3,-1) {$\tomachjump$};
			\node at (0,-.5) {$\bisimstate$};
		\node at (6,-.5) {$\bisimstate$};
			\end{tikzpicture}\end{center}
			The \emph{moreover} part follows from the \ih		
	\end{itemize}
\end{proof}

% !TEX root = main.tex
% !TeX spellcheck = en_US
\section{Proofs of \refsect{typed-invariant} (Typed Interactions Are 
Exhausting)}
\label{sect:types-appendix}
In the first part of this section we prove the S-exhaustible state invariant for the \SIAM, then use it to extract 
\LIAM states from \SIAM ones, and finally prove the strong bisimulation between the two machines.

In the second part we deal with showing that the \SIAM never loops on type derivations. The key tool shall be a 
loop-preserving bisimulation between \SIAM states of the type derivation of $\tm$ and $\tmtwo$ if $\tm \towh \tmtwo$.

\subsection{S-Exhaustible Invariant}
We present an example of type derivation for the term 
$\tm=(\la\vartwo{\la\var{\var\vartwo}})\mathsf{I}(\la\varthree\varthree)$, the 
same 
example used in Section~\ref{sec:IJK}. We use it to explain the next 
technical definitions. We have annotated the occurrences of $\initty$ with 
natural numbers, so that they represent the run on the type derivation.
\[
\infer{\tjudg{}{(\la\vartwo{\la\var{\var\vartwo}})\mathsf{I}(\la\varthree\varthree)}
 {\initty_{\uppt{\red 1}}}}{
 \infer{\tjudg{}{(\la\vartwo{\la\var{\var\vartwo}})\mathsf{I}} 
 {\arr{\mset{\arr{\mset{\initty_{\uppt{\red 
 						{13}}}}}\initty_{\downpt{\blue 9}}}}\initty_{\uppt{\red 
 2}}}}{
 {\infer{\tjudg{}{\la\vartwo{\la\var{\var\vartwo}}}{\arr{\mset{\initty_{\downpt{\blue
 						{18}}}}}{ 
 	\arr{\mset{\arr{\mset{\initty_{\uppt{\red 
 						{14}}}}}\initty_{\downpt{\blue 
 	8}}}}\initty_{\uppt{\red 
 	3}}}}}{
 \infer{\tjudg{\vartwo:\mset\initty}{\la\var{\var\vartwo}} 
 {\arr{\mset{\arr{\mset{\initty_{\uppt{\red 
 						{15}}}}}\initty_{\downpt{\blue 7}}}}\initty_{\uppt{\red 
 4}}}}{
 \infer{\tjudg{\vartwo:\mset\initty,\var:\mset{\arr{\mset\initty}\initty}}{\var\vartwo}
  {\initty_{\uppt{\red 5}}}}{
 \infer{\tjudg{\var:\mset{\arr{\mset\initty}\initty}} 
 {\var}{\arr{\mset{\initty_{\downpt{\blue 
 				{16}}}}}\initty_{\uppt{\red 6}}}}{}\qquad 
\infer{\tjudg{\vartwo:\mset\initty}{\vartwo}{\initty_{\uppt{\red 
				{17}}}}}{}}}}\qquad 
\infer{\tjudg{}{\mathsf{I}}{\initty_{\uppt{\red {19}}}}}{}}}\qquad 
\infer{\tjudg{}{\la\varthree\varthree}{\arr{\mset{\initty_{\downpt{\blue 
{12}}}}}\initty_{\uppt{\red 
				{10}}}}}{
	\tjudg{\varthree:\mset\initty}{\varthree}{\initty_{\uppt{\red {11}}}}}}
\]
We start by defining the notions of typed tests used to define S-exhaustible 
states. Somewhat surprising, while in the 
\LIAM tape tests are easy to define and log tests require some syntactical gymnastics, here it is the other way 
around.

\paragraph{Type Positions and Generalized States} To define tests, we have to consider 
a slightly more general notion of \SIAM state. In \refsect{SIAM}, a state 
is a quadruple $(\tyd, \ruleoc, \tyctx, \pol)$ where $\ruleoc$ is an occurrence 
of a judgement 
$\tjudg{\tye}{\tmtwo}{\ty}$ in $\tyd$, $\pol$ is a direction, and $\tyctx$ is a type context isolating an occurrence of 
$\initty$ in $\ty$. The generalization simply is to consider type contexts $\tyctx$ such that $\tyctxp{\tytwo} = \ty$ 
for some $\tytwo$, that is, not necessarily isolating $\initty$. A pair $(\tytwo, \tyctx)$ such that $\tyctxp{\tytwo} = 
\ty$ is called a position in $\ty$. 

Note that the \SIAM can be naturally adapted to this more general notion of state, that follows an arbitrary formula 
$\tytwo$, not necessarily $\initty$
%---it can be found in \reffig{genSIAM}, 
and it amounts to simply replace $\initty$ 
with $\tytwo$.

To easily manage \SIAM states we also use a concise notations, writing $\tjudg{}{\tm}{\ty,\tyctx}$ for a state 
$\state= (\tyd, \ruleoc, (\ty,\tyctx), \pol)$ where $\ruleoc$ is 
$\tjudg{\tye}{\tm}{\tyctxp\ty}$ for some $\tye$, potentially 
specifying the direction via colors and under/over-lining.

%\begin{figure}[t]
%  \input{machines/SIAM_generalized}
%	\vspace{-8pt}
%	\caption{The transitions of the (Generalized) Sequence $\IAMold$ (\SIAM).}
%	\label{fig:genSIAM}  
%\end{figure}

\paragraph{\SIAM Tests} Given a \SIAM state $\state= (\tyd, \ruleoc, (\ty,\tyctx), \pol)$, the underlying idea is that 
the judgement occurrence $\ruleoc$ encodes the log of the \LIAM, while the type 
context $\tyctx$ encodes the 
tape. It is then natural to define two kinds of test, one for judgements and 
one for type contexts.

The intuition is that a test focuses on (the occurrence of) an element $\tytwo$ 
of a sequence $\mty$ related to 
$\state$, and that these sequence elements play the role of logged positions in the \LIAM. These sequence elements 
are of two kinds:
\begin{enumerate}
 \item \emph{Elements containing $\ruleoc$}: those in which the focused 
 judgment $\ruleoc$ itself is contained, 
corresponding to the logged positions in the log of the \LIAM. Note that the positions on the log are those for which 
the \LIAM has previously found the corresponding arguments. In the \SIAM these arguments are exactly those in which the 
focused judgment is contained.

\item \emph{Elements appearing in $\tyctx$}: those in 
the right-hand type of $\state$ in which the focused type $\ty$ is contained, 
corresponding to the logged positions on 
the tape of the \LIAM. They correspond to \LIAM queries for which the argument has not yet been found, or positions to 
which the \LIAM is backtracking to.
\end{enumerate}
Each one of these elements is then identified by a judgement occurrence 
$\tydtwo$ and a position $(\tytwo,\tyctxtwo)$ 
in the right-hand type of $\tydtwo$. 

\begin{definition}[Focus]
 A \emph{focus} $\focus$ in a derivation $\tyd$ is a pair $\focus = (\ruleoc, 
 (\ty,\tyctx))$ of a judgement occurrence 
$\ruleoc$ and of a type position $(\ty,\tyctx)$ in the right-hand type $\tyctxp\ty$ of $\ruleoc$.
\end{definition}

The intuition is that exhausting a test $\state_{\ruleoc, (\ty,\tyctx)}$ in $\tyd$ shall amount to retrieve the 
axiom of 
$\tyd$ of type $\ty$ that would be substituted by that sequence element of type $\ty$ by reducing $\tyd$ via 
cut-elimination---the 
definition of exhaustible tests is given below, after the definition of tests.

\begin{definition}[Judgement tests]
	Let $\state=(\tyd, \ruleoc, (\ty,\tyctx), \pol)$ be a \SIAM state. Let $r_i$ be $i$-th $\tyapp$ rule found traversing $\tyd$ 
by descending from the focused judgement $\ruleoc$ towards the final judgement 
of $\tyd$. Let $\ruleoc_i$ be the 
judgement of the sequence $\mty_i$ in the right premise of $r_i$ traversed in 
such a descent (careful: $\ruleoc_i$ is 
the $j$-th judgement of $\mty_i$ for some $j$, that is, the index $i$ denotes 
the connection with rule $r_i$, not the 
position in $\mty_i$). 
Let $\ruleoc_i$ be $\tjudg{\tye}{\tm}{\tytwo}$. Then 
$\state_{\focus}^i = (\tyd,\ruleoc_i,(\tytwo,\ctxhole), 
 \downpt)$ is the $i$-th judgement test of $\state$, having as focus $\focus 
 \defeq (\ruleoc_i,(\tytwo,\ctxhole))$. 
\end{definition}
We 
often
omit the judgement from the focus, writing simply $\state_{(\tytwo,\ctxhole)}$, 
and even concisely 
note $\state_{\focus}$ as 
$\tjudg{}{\blue\tm}{\tytwo,\ctxhole_\downpt}$.

Note that judgement tests always have type context $\ctxhole$. According to the 
intended correspondance judgement/ 
log and type context/tape between the \SIAM and the \LIAM, having type context $\ctxhole$ corresponds to the fact that 
the log tests of the \LIAM always have an empty tape.

\begin{example}[Judgement test]
Let us give an example of judgement test in the context of the given example of 
\SIAM run. If we 
consider the state $\uppt\red{11}$, we find its log tests going down in the 
type 
derivation for each $\tyapp$ rule traversed from the right hand side. In this 
case we immediately find the judgment 
$\tjudg{}{\la\varthree\varthree}{\arr{\mset\initty}\initty}$. Then, 
$\tjudg{}{\la\varthree\varthree}{\ctxholep{\arr{\mset\initty}\initty}_\downpt}$
 is a log test for 
$\uppt\red{11}$. Since between 
$\tjudg{}{\la\varthree\varthree}{\arr{\mset\initty}\initty}$
and the root of the derivation we do not cross any other suitable $\tyapp$ 
rule, there are no other log tests for $\uppt\red{11}$. 
\end{example}

\paragraph{Type (Context) Tests} While judgement tests depend only on the 
judgement occurrence $\ruleoc$ of a state 
$\state = (\tyd, \ruleoc, (\ty,\tyctx), \pol)$, type context tests---dually---fix $\ruleoc$ and depend only on the 
type 
context $\tyctx$ of $\state$, that is, they all focus on sequence elements of the form $(\ruleoc, (\tytwo,\tyctxtwo))$ 
where $\tyctxtwop\tytwo = \tyctxp\ty$ and $\tyctx = \tyctxtwop\tyctxthree$ for some type context $\tyctxthree$. Namely, 
there is one type context test (shortened to \emph{type test}) for every sequence 
in which the hole of $\tyctx$ is contained. We need some notions about type contexts, in particular a notion of level 
analogous to the one for term contexts.

\paragraph{Terminology About Type Contexts} Define type contexts $\tyctx_n$ of level $n\in\nat$ as follows:

\[\begin{array}{lclr}
	\tyctx_0 &\defeq &\ctxhole \mid \arr\mty\tyctx_0
	\\
	\tyctx_{n+1} &\defeq &\arr{\mset{\myldots\tyctx_n\myldots}}\ty \mid \arr\mty\tyctx_{n+1}	
	\end{array}\]
Clearly, every type context $\tyctx$ can be seen as a type context $\tyctx_n$ for a unique $n$, and viceversa a type 
context of level $n$ is also simply a type context---the level is then sometimes omitted.
A \emph{prefix} of a context $\tyctx$ is a context $\tyctxtwo$ such that $\tyctxtwop\tyctxthree = \tyctx$ for some 
$\tyctxthree$. Given $\tyctx$ of level $n>0$, there is a smallest prefix context $\tyctx|_i$ of level $0<i\leq 
n$, and it has the form $\tyctxtwo\ctxholep{\arr{\mset{\myldots\ctxhole\myldots}}\ty}$ for a type context $\tyctxtwo$ 
of level $i-1$.

%\ben{example: TODO}.

\begin{definition}[Type tests]
	Let $\state=(\tyd, \ruleoc, (\ty,\tyctx), \pol)$ be a \SIAM state and $n$ be the level of $\tyctx$. The sequence 
of directed prefixes $\DiPref\tyctx$ of $\tyctx$ is the sequence of pairs  
	$(\tyctxtwo,\poltwo)$, where $\tyctxtwo$ is a prefix of $\tyctx$, defined as follows:
	\[\begin{array}{lclll}
	\DiPref\tyctx & \defeq & \mset\cdot & \mbox{if }$n=0$
	\\
	\DiPref\tyctx & \defeq & \mset{(\tyctx|_1,\uppt), \ldots,(\tyctx|_n,\uppt^{n-1})} &\mbox{if }$n>0$	
	\end{array}\]
The $i$-th directed prefix (from left to right) $(\tyctxtwo,\poltwo)$ in $\DiPref\tyctx$ induces the type test 
$\state_{\focus}^i \defeq (\tyd, \ruleoc, (\tyctxthreep\ty,\tyctxtwo), \poltwo)$ of $\state$ and focus $\focus \defeq 
(\ruleoc,(\tyctxthreep\tytwo,\tyctxtwo))$, where $\tyctxthree$ is the 
unique type context such that $\tyctx = \tyctxtwop\tyctxthree$.
\end{definition}

According to the idea that type tests correspond to 
the tape tests of the \LIAM,  note that the first element (on the left) of the sequence $\DiPref\tyctx$ has $\uppt$ 
direction, and that the direction alternates along the sequence. This is the analogous to the fact that the tape test 
associated to the first logged position on the tape 
(from left to right) has always direction $\downp$, and passing to the test of the next logged position on the tape switches the direction. 

\begin{example}[Type test]
 Let us now give examples of type tests in the example of \SIAM run that we provided. We 
compute the tape tests of $\uppt\red{13}$. Its type is 
\[\arr{\mset{\arr{\mset{\ctxholep\initty}}\initty}}\initty\]
with respect to the notation of the previous definition, we have $\ty = \initty$ and $\tyctx = 
\arr{\mset{\arr{\mset{\ctxhole}}\initty}}\initty$. The level of $\tyctx$ is $2$. 
Tape tests are associated with the pairs in 
$\DiPref{\arr{\mset{\arr{\mset{\ctxhole}}\initty}}\initty}$, namely 
$\mset{(\arr{\mset{\ctxhole}}\initty, \uppt), (\arr{\mset{\arr{\mset{\ctxhole}}\initty}}\initty,\downpt)}$.
\end{example}

\begin{definition}[State respecting a focus]
	Let $\focus=(\ruleoc, (\ty,\tyctx))$ be a focus. A \SIAM state $\state$ respects $\focus$ if it is an axiom 
	$\tjudg{}{\blue\var}{\ctxholep\ty_\downpt}$ for some variable $\var$ (the typing context of $\state$, which is omitted 
by convention, is $\var:\ty$).
\end{definition}

\begin{definition}[S-Exhaustible states]
	The set $\exstates_S$ of S-exhaustible states is the 
	smallest set such that if $\state\in\exstates_S$, then for each type or 
	judgement test of $\state_\focus$ of focus $\focus$ there exists a run 	
$\run: \state_f \tosiam^*\tomachbttwo\statetwo$ where $\statetwo$ respects $\focus$ and for 
the shortest such run $\statetwo\in\exstates_S$.
\end{definition}

\begin{lemma}[S-exhaustible invariant]
\label{l:S-invariant-siam}
	Let $\tm$ be a closed term, $\tyd\pof\tjudg{\tye}{\tm}{\ty}$ a sequence type derivation for it, and $\run:\ \tjudg{}{\tm}{\ctxholep\ty_\uppt} \tosiam^k
\state$ an initial \SIAM run. Then $\state$ is S-exhaustible. 
\end{lemma}

% !TEX root = ../main.tex
% !TeX spellcheck = en_US
%%%%%%%%%%%%%%%%%%%%%%%%%%%
\begin{proof}
		By
	induction on $k$. For $k=0$ there is nothing to prove because the initial state $\state_0 = \tjudg{}{\tm}{\ctxholep\ty_\uppt}$ has has no judgement nor type tests. Then suppose
	$\run':\state_0\toliam^{k-1}\statetwo$ and that the run continues with $\statetwo\tosiam\state$. By \ih, $\statetwo$ is 
S-exhaustible.

\emph{Terminology}: when a test state satisfies the clause in the definition of S-exhaustible states we say that it is \emph{positive}. 

	 Cases of 
	$\statetwo\tosiam\state$:
	
	\begin{itemize}
		\item Case $\tomachdotone$.
		\[\begin{array}{clc}
		\statetwo=\infer{\tjudg{}{\red{\tm\tmtwo}}{\tyctxp{\initty_{\uppt}}(=\ty)}} 
		{\tjudg{}{\tm}{\arr{\mty}{\ty}} & \mset\vdash} &
		\tomachdotone &
		\infer{\tjudg{}{\tm\tmtwo}{\ty 
			}}{\tjudg{}{\red\tm}{\arr{\mty}{\tyctxp{\initty_{\uppt}}}} & 
			\mset\vdash}=\state
		\end{array}\]
		\begin{itemize}
			\item \emph{Judgement tests.} Note that $\state$ has the same judgement tests of $\statetwo$, which are 
positive by the \ih
			\item \emph{Type tests.} We first consider the type tests of direction $\uppt$. 
			Let us $\red{\state_\focus}$ be one of them. We observe that there 
			is a 
			corresponding type test $\red{\statetwo_\focus}$ of $\statetwo$, that by \ih it is positive, and that 
			$\red{\statetwo_\focus}\tosiam\red{\state_\focus}$. Since the machine is deterministic also $\red{\state_\focus}$ 
is positive. Let us now consider a type test $\blue{\state_\focus}$ of direction $\downpt$. We observe 
			that there is a 
			corresponding type test $\blue{\statetwo_\focus}$ of $\statetwo$, that it is positive by \ih, and that 
			$\blue{\state_\focus}\to\blue{\statetwo_\focus}$. Then $\blue{\state_\focus}$ is positive.
		\end{itemize}
		\item Case $\tomachdottwo$. Identical to the previous one.
		\item Case $\tomachvar$.
		\[
		\begin{array}{clc}
		\statetwo=\infer*{\infer{\tjudg{}{\la\var\ctxp{\var}} 
				{\arr{\mset{\ldots\ty_i\ldots}}\tytwo}}{}}
		{\infer[i]{\tjudg{}{\red\var}{\tyctxp{\initty_{\uppt}}_i(=\ty_i)}}{}}  & 
		\tomachvar
		& \infer*{\infer{\tjudg{}{\blue{\la\var\ctxp{\var}}}
				{\arr{\mset{\ldots\tyctxp{\initty_{\downpt}}_i\ldots}}\tytwo}}{}}
		{\infer[i]{\tjudg{}{\var}{\ty_i}}{}}=\state		
		\end{array}
		\]
		\begin{itemize}
			\item \emph{Judgement tests.} Judgement tests of $\state$ are a subset 
			of judgement tests of $\statetwo$ and thus positive by \ih
			
			\item \emph{Type tests.} Let $n$ be the level of $\tyctx$. Let $\state^j$ be the type test of $\state$ 
associated to the $j$-th triple in $\DiPref{\arr{\mset{\ldots\tyctx\ldots}}\tytwo}$. Three 
cases, depending on the index $j$ of $\state^j$:
\begin{enumerate}
 \item $j = 1$: then $\state^1$ is $\tjudg{}{\red{\la\var\ctxp{\var}} }
			{\tyctxp{\initty}_{i\uppt},\arr{\mset{\ldots\ctxhole\ldots}}\tytwo}$. Note that 
$\state^1\tomachbttwo\,\tjudg{} 
			{\blue\var}{\tyctxp{\initty}_{i\downpt},\ctxhole}$, which has no type tests 
			and has the same judgement tests of $\statetwo$, which by 
			\ih are positive. Hence, $\state^1$ is S-exhaustible.
 
    \item \emph{$j$ is even}: for 
			$\blue\state^j$ (of direction $\downpt$) there is a corresponding type test $\red\statetwo^{j-1}$ of odd index 
of
			$\statetwo$, having 
direction $\uppt$ and such that 
			$\red\statetwo^{j-1}\tomachvar\blue\state^j$.  Thus one can 
			conclude by \ih and determinism of the \SIAM.
    
    \item \emph{$j \neq 1$ is odd}: for 
			$\red\state^j$ (of direction $\uppt$) there is a 
			corresponding type test $\blue\statetwo^{j-1}$ of even index of
			$\statetwo$, having direction $\downpt$ and such that 
			$\red\state^j\tomachbttwo \blue\statetwo^{j-1}$. Thus one can 
			conclude by \ih
\end{enumerate}

		\end{itemize}
		\item Case $\tomachbttwo$.
		\[\begin{array}{clc}
		\statetwo=\infer*{\infer{\tjudg{}{\red{\la\var\ctxp{\var}}}
				{\arr{\mset{\ldots\tyctxp{\initty_{\uppt}}_i\ldots}}\tytwo}}{}}
		{\infer[i]{\tjudg{}{\var}{\ty_i (= \tyctxp\initty_i)}}{}}
		& \tomachbttwo
		& \infer*{\infer{\tjudg{}{\la\var\ctxp{\var}} 
				{\arr{\mset{\ldots\ty_i\ldots}}\tytwo}}{}}
		{\infer[i]{\tjudg{}{\blue\var}{\tyctxp{\initty_{\downpt}}_i}}{}}=\state
		\end{array}\]
		\begin{itemize}
			\item \emph{Judgement tests.} The first type test of $\statetwo$ is 
			$\statetwo^1 \defeq \tjudg{}{\red{\la\var\ctxp{\var}}} 
			{\tyctxp{\initty}_{i\uppt},\arr{\mset{\ldots\ctxhole\ldots}}\tytwo}$. Note that 
$\statetwo^1\tomachbttwo\,\tjudg{} 
			{\blue\var}{\tyctxp{\initty}_{i\downpt},\ctxhole} \eqdef \statethree$ and that $\statethree$ exhausts 
$\statetwo^1$, and it is the first such state. Since $\statetwo^1$ is positive, $\statethree$ is S-exhaustible. Note 
that $\statethree$ has the same judgment tests of $\state$, which are then 
positive.
			
			\item \emph{Type tests.} For each odd type test 
			$\red\state^i$ of $\state$ (whose direction is $\uppt$), the corresponding even 
			type test $\blue\statetwo^{i+1}$ of $\statetwo$ has direction $\downpt$, is positive by \ih, and such that 
			$\red\state^i\tomachvar\blue\statetwo^{i+1}$. Then $\red\state^i$ is positive. For each even type 
			test 
			$\blue\state^i$ of $\state$ (whose direction is $\downpt$), the corresponding odd 
			type test $\red\statetwo^{i+1}$ of $\statetwo$ has direction $\uppt$, is positive by \ih, and such that 
			$\red\statetwo^{i+1}\tomachbttwo\blue\state^i$. Then $\state^i$ is positive by determinism of the 
\SIAM.
		\end{itemize}
		
		%%%%%%%%%%%%%%%%%%%%%%%%%%%%%%%%%
		\item Cases $\tomachdotthree$ and $\tomachdotfour$. They are identical to case 
		$\tomachdotone$.
		
		%%%%%%%%%%%%%%%%%%%%%%%%%%%%%%%%%
		\item Case $\tomacharg$.
		\[\begin{array}{clc}
		\statetwo=\infer{\tjudg{}{\tm\tmtwo}{\ty}} 
		{\tjudg{}{\blue\tm}{\arr{\mset{\ldots 
						\tyctxp{\initty_{\downpt}}_i\ldots}}{\ty}}
			& \tjudgi{}{\tmtwo}{\tytwo_i (=\tyctxp{\initty_{\downpt}}_i)}}
		& \tomacharg &
		\infer{\tjudg{}{\tm\tmtwo}{\ty}} 
		{\tjudg{}{\tm}{\arr{\mset{\ldots 
						\tytwo_i\ldots}}{\ty}}
			& \tjudgi{}{\red\tmtwo}{\tyctxp{\initty_{\uppt}}_i}}=\state
		\end{array}\]
		\begin{itemize}
			\item \emph{Judgement tests.} Judgement tests of $\state$ are those of 
			$\statetwo$, which are positive by \ih, plus 
			$\state^\tmtwo \defeq \tjudg{}{\blue\tmtwo}{\tyctxp{\initty}_{i\downpt},\ctxhole}$. 
			Please note that 
			$\state^\tmtwo\tomachbtone\,\tjudg{}{\red\tm}{\tyctxp{\initty}_{i\uppt}, 
				\arr{\mset{\ldots \ctxhole\ldots}}{\ty}}\eqdef\statetwo^t$. Now,
			$\statetwo^t$ is a type test of $\statetwo$ and by \ih is 
			positive. Then $\state^\tmtwo$ is positive. 
			
			\item \emph{Type tests.} For each odd type test 
			$\red\state^i$ of $\state$ (whose direction is $\uppt$), the corresponding even 
			type test $\blue\statetwo^{i+1}$ of $\statetwo$ has direction $\downpt$, is positive by \ih, and such that 
			$\blue\statetwo^{i+1}\tomacharg\red\state^i$. Then $\state^i$ is positive by determinism of the 
\SIAM. For each even type 
			test 
			$\blue\state^i$ of $\statetwo$ (whose direction is $\downpt$), the corresponding odd 
			type test $\red\statetwo^{i+1}$ of $\statetwo$ has direction $\uppt$, is positive by \ih, and such that 
			$\blue\state^i\tomachbtone\red\statetwo^{i+1}$. Then $\blue\state^i$ is positive.
		\end{itemize}
		
		%%%%%%%%%%%%%%%%%%%%%%%%%%%
		\item Case $\tomachbtone$.
		\[\begin{array}{clc}
		\statetwo=\infer{\tjudg{}{\tm\tmtwo}{\ty}} 
		{\tjudg{}{\tm}{\arr{\mset{\ldots 
						\tytwo_i\ldots}}{\ty}}
			& \tjudgi{}{\blue\tmtwo}{\tyctxp{\initty_{\downpt}}_i(=\tytwo_i)}}
		& \tomachbtone &
		\infer{\tjudg{}{\tm\tmtwo}{\ty}} 
		{\tjudg{}{\red\tm}{\arr{\mset{\ldots 
						\tyctxp{\initty_{\uppt}}_i\ldots}}{\ty}}
			& \tjudgi{}{\tmtwo}{\tytwo_i}}=\state
		\end{array}\]
		\begin{itemize}
			\item \emph{Judgement tests.} All judgement tests of $\state$ are judgement 
			test of $\statetwo$, which are this way positive by \ih
			
			\item \emph{Type tests.} The first type test of $\state$ is 
			$\state^1 \defeq \tjudg{}{\red\tm}{\tyctxp{\initty}_{i\uppt}, 
				\arr{\mset{\ldots \ctxhole\ldots}}{\ty}}$. 
			Please note that $\statetwo^\tmtwo \defeq \tjudg{}{\blue\tmtwo} {\tyctxp{\initty}_{i\downpt},\ctxhole}$ is a 
judgement test of $\statetwo$ such that $\statetwo^\tmtwo \tomachbtone \state^1$. By \ih, 
$\statetwo^\tmtwo$ is positive. By determinism of the \SIAM, $\state^1$ is positive.

			For each odd type test 
			$\red\state^i$ of $\state$ (whose direction is $\uppt$), the corresponding even 
			type test $\blue\statetwo^{i-1}$ of $\statetwo$ has direction $\downpt$, is positive by \ih, and such that 
			$\blue\statetwo^{i-1}\tomachbtone\red\state^i$. Then $\red\state^i$ is positive by determinism 
of the \SIAM. For each even type 
			test 
			$\blue\state^i$ of $\statetwo$ (whose direction is $\downpt$), the corresponding odd 
			type test $\red\statetwo^{i-1}$ of $\statetwo$ has direction $\uppt$, is positive by \ih, and such that 
			$\blue\state^i\tomacharg\red\statetwo^{i-1}$. Then $\blue\state^i$ is positive.
			
		\end{itemize}
	\end{itemize}
\end{proof}

\subsection{Extracting \LIAM States from \SIAM S-Exhaustible States, and the \LIAM/\SIAM Strong Bisimulation} 
From S-exhaustible states one is able to \emph{extract} \LIAM states, as the 
following definition shows. Please note that the definition is well-founded, 
precisely because the objects are S-exhaustible states. Indeed, the induction 
principle used to define S-exhaustability allows recursive definition on 
S-exhaustible states to be well-behaved. 
\begin{definition}[Extraction of logged positions]
	Let $\state$ be an S-exhaustible \SIAM state in a derivation $\tyd$, $\tm$ 
	be the final term in $\tyd$, and $\state_\focus$ be a judgement or type 
	test of $\state$. Since 
		$\state$ is 
		S-exhaustible, there is an exhausting run 
		$\state_\focus\tosiam^+ \statetwo
		\in\exstates_S$. Let $\var$ be the variable of $\statetwo$. Then the logged position extracted from 
$\state_\focus$ is 
		$\elpos{\state_\focus} \defeq 
(\var,\la\var\ctxtwo_n,\elpos{\statetwo^1}\cdot\ldots\cdot\elpos{\statetwo^n})$, where 
		$\ctxtwo_n$ is the context (of level $n$)
		retrieved traversing $\tyd$ from $\statetwo$ to 
		the binder of $\l\var$ of $\var$ in $\tm$ and $\statetwo^i$ is the 
		$i$-th judgement test of $\statetwo$.
\end{definition}

\begin{definition}[Extraction of logs, tapes, and states]
	Let $\state=(\tyd, \ruleoc, (\ty,\tyctx), \pol)$ be an S-exhaustible \SIAM state where $\tm$ is the final term in 
$\tyd$, and $\ruleoc$ is $\tjudg{\tye}{\tmtwo}{\tyctxp\ty}$. The \LIAM state extracted from $\state$  is 
	$\estate{\state} \defeq \nopolstate{\tmtwo}{\ctx_\state}{\etape\state}{\elog{\state}}{\pol}$ 
	where
	\begin{itemize}
		\item \emph{Context}: $\ctx_\state$ is the only term context such that $\tm = \ctx_\state\ctxholep\tmtwo$;
		\item \emph{Log}:
		$\elog{\state}\defeq\lpos_1\cdots\lpos_i\cdots\lpos_n$ where $\lpos_i = \elpos{\state^i_\focus}$ where 
$\state^i_\focus$ is the $i$-th judgement test of $\state$.		
		\item \emph{Tape}: $\etape\state = \etapeauxs{\tyctx,0}$ where $\etapeauxs{\tyctx,i}$ is the auxiliary function 
defined by induction on $\tyctx$ as follows.
		\[\begin{array}{lcl}
		\etapeauxs{\ctxhole,i} &\defeq &\stempty 
		\\
		\etapeauxs{\arr\mty{\tyctx}, i} & \defeq & \resm\cdot \etapeauxs{\tyctx,i}
		\\
		\etapeauxs{\arr{\mset{\myldots\tyctx\myldots}}\tytwo,i}
		& \defeq & \elpos{\state^i_\focus}\cdot\etapeauxs{\tyctx, i+1}
		\end{array}
		\]
		where $\state^i_\focus$ is the $i$-th type test of $\state$.
	\end{itemize}
	We use $\bisimtypes$ for the extraction relation between S-exhaustible \SIAM states and \LIAM states defined as $(\state, \estate{\state}) \in\bisimtypes$.
\end{definition}

First of all, we show that the extracted stated respects the \LIAM invariant about the length of the log.

\begin{lemma}
\label{l:extraction-length}
 Let $\state$ be an S-exhaustible \SIAM state and $\estate{\state} = 
\nopolstate{\tm}{\ctx_\state}{\etape\state}{\elog{\state}}{\pol}$ the \LIAM state extracted from it. Then the level of 
$\ctx_\state$ is exactly the length of $\elog{\state}$, that is, $(\tm, \ctx_\state,\elog{\state})$ is a logged 
position.
\end{lemma}

\begin{proof}
The length of $\elog{\state}$ is the number of judgement tests of $\state$, 
which is the number of 
$\tyapp$ rules traversed descending from the focused judgement $\ruleoc$ of 
$\state$ to the final judgement of $\tyd$. 
The level of $\ctx_\state$ is the number of arguments in which the hole of $\ctx_\state$ is contained, which are 
exactly the number of 
$\tyapp$ rules traversed descending from $\ruleoc$ to the final judgement of 
$\tyd$.
\end{proof}

\begin{proposition}[\SIAM-\LIAM bisimulation]
Let $\tm$ a closed and $\towh$-normalizable term, and $\tyd\pof\tjudg{}{\tm}{\initty}$ a type derivation. Then 
$\bisimtypes$ is a strong bisimulation between S-exhaustible \SIAM states on $\tyd$ and \LIAM states on $\tm$. 
Moreover, if $\state_\tyd \bisimtypes \state_\l$ then $\state_\tyd$ is \SIAM reachable if and only if $\state_\l$ is 
\LIAM reachable.
\end{proposition}

% !TEX root = ../main.tex
% !TeX spellcheck = en_US
%%%%%%%%%%%%%%%%%%%%%%%%%%%
\begin{proof}
  Assuming the bisimulation part of the statement, the moreover part follows from a trivial induction on the length of 
the initial run, since initial state are bisimilar and the bisimulation is exactly the fact that $\bisimtypes$ is 
stable by transitions.

For the bisimulation part, we consider each possible transitions. We focus on the half of the proof showing that \SIAM 
transitions are simulated by the \LIAM, the other half is essentially identical. 
	
	\begin{itemize}
    %%%%%%%%%%%%%%%%%%%%%%%%%%
		\item Case $\tomachdotone$.
		\[\begin{array}{clc}
		\statetwo=\infer{\tjudg{}{\red{\tm\tmtwo}}{\tyctxp{\initty_{\uppt}}(=\ty)}} 
		{\tjudg{}{\tm}{\arr{\mty}{\ty}} & \mset\vdash} &
		\tomachdotone &
		\infer{\tjudg{}{\tm\tmtwo}{\ty 
			}}{\tjudg{}{\red\tm}{\arr{\mty}{\tyctxp{\initty_{\uppt}}}} & 
			\mset\vdash}=\state
		\\[8pt]	
			\bisimtypes&&
		\\[8pt]
		\estate\state=\dstate{ \tm\tmtwo }{ \ctx_\statetwo }{ \etape{\statetwo} }{ \elog{\statetwo} } 
		&\iamdap& 
		\dstate{ \tm }{ \ctxp{\ctxhole\tmthree} }{ \resm\cdot \etape{\statetwo} }{ 
			\elog{\statetwo} } = \state_\l
		\end{array}\]
		
			Note that $\ctx_\state = \ctx_\statetwo\ctxholep{\ctxhole\tmthree}$, $\elog\state = 
\elog{\statetwo}$, and $\etape{\statetwo} = \resm\cdot\etape{\state}$. Then, $\state_\l = \estate{\state}$, that is, 
$\state \bisimtypes \state_\l$.

    %%%%%%%%%%%%%%%%%%%%%%%%%%
    \item Case $\tomachdottwo$. Identical to the previous one.
    
    %%%%%%%%%%%%%%%%%%%%%%%%%%
		\item Case $\tomachvar$.
		\[\begin{array}{clc}
		\statetwo=\infer*{\infer{\tjudg{}{\la\var\ctxtwo_n\ctxholep{\var}} 
				{\arr{\mset{\ldots\ty_i\ldots}}\tytwo}}{}}
		{\infer[i]{\tjudg{}{\red\var}{\tyctxp{\initty_{\uppt}}_i(=\ty_i)}}{}}  & 
		\tomachvar
		& \infer*{\infer{\tjudg{}{\blue{\la\var\ctxtwo_n\ctxholep{\var}}}
				{\arr{\mset{\ldots\tyctxp{\initty_{\downpt}}_i\ldots}}\tytwo}}{}}
		{\infer[i]{\tjudg{}{\var}{\ty_i}}{}}=\state		
		\\[8pt]
		\bisimtypes&&
		\\[8pt]
		\estate\state=\dstate{ \var }{ \underbrace{\ctxp{\l\var.\ctxtwo_n}}_{=\ctx_\statetwo} }{ \etape\statetwo }{ 
			\underbrace{\tlog_n\cdot\tlog}_{=\elog\statetwo} } 
		&\tomachvar &
		\ustate{ \l\var.\ctxtwo_n\ctxholep\var}{ \ctx }{ 
			(\var,\l\var.\ctxtwo_n,\tlog_n)\cdot\etape\statetwo }{ \tlog } = \state_\l		
		\end{array}\]
		First of all, $\ctx_\state$ has shape $\ctxp{\l\var.\ctxtwo_n}$ for some $n$ has the descending path from the 
focused judgement to the final judgement passes through the showed $\tylam$ rule. Then $\ctx_{\statetwo} = \ctx$.

About the log, 
by \reflemma{extraction-length} there is a correspondance between the level of term contexts and the length of the 
extracted log, so that $\elog\state$ has at least length $n$, that is, $\elog\statetwo = \tlog_n\cdot\tlog$, and 
$\elog\state = \tlog$.
		
		About the tape, note that $\etape\state = \elpos{\state^1_\focus}\cons\etapeaux{\tyctx,1}\state$ where 
$\state^1_\focus$ is the first type test of $\state$. To show that $\estate\state = 
(\var,\l\var.\ctxtwo_n,\tlog_n)\cdot\etape\statetwo$ we have to show two things:
\begin{enumerate}
 \item $\elpos{\state^1_\focus} = 
(\var,\l\var.\ctxtwo_n,\tlog_n)$.
		Note that $\state^1_\focus$ is $\tjudg{}{\red{\la\var\ctxp{\var}} }
			{\tyctxp{\initty}_{i\uppt},\arr{\mset{\ldots\ctxhole\ldots}}\tytwo}$. Note that 
$\state^1_\focus\tomachbttwo\,\tjudg{} 
			{\blue\var}{\tyctxp{\initty}_{i\downpt},\ctxhole} = \statethree$, where $\statethree$ focusses on the same 
judgement of $\statetwo$, and that $\statethree$ is the state that S-exhausts $\state^1_\focus$. By definition of 
extraction, $\elpos{\state^1_\focus} = (\var,\l\var.\ctxtwo_n,\tlog_n)$. 

  \item $\etapeaux{\tyctx,1}\state = \etape\statetwo$, that is, $\etapeaux{\tyctx,1}\state = 
\etapeaux{\tyctx,0}\statetwo$. Note that $\etapeaux{\tyctx,1}\state$ and $\etapeaux{\tyctx,0}\statetwo$ may differ only 
in the content of logged positions (obtained by extracting from tape tests), which is the only thing that depends on 
the direction and the state, the rest being uniquely determined by the type context $\tyctx$. Here one has to repeat 
the reasoning done in the $\tomachbttwo$ case of the proof of the S-exhaustible invariant 
(\reflemma{S-invariant-siam}), that shows that the tape test of index $i>1$  for $\state$ and the one of index $i-1$ of 
$\statetwo$ exhaust on the same state, and thus induce the same logged position. Then $\etapeaux{\tyctx,1}\state = 
\etape\statetwo$.
\end{enumerate}
Then $\estate\state = (\var,\l\var.\ctxtwo_n,\tlog_n)\cdot\etape\statetwo$, and so $\state_\l = \estate{\state}$, that 
is, $\state \bisimtypes \state_\l$.

		%%%%%%%%%%%%%%%%%%%%%%%%%%
		\item Case $\tomachbttwo$.
		\[\small\begin{array}{clc}
		\statetwo=\infer*{\infer{\tjudg{}{\red{\la\var\ctxp{\var}}}
				{\arr{\mset{\ldots\tyctxp{\initty_{\uppt}}_i\ldots}}\tytwo}}{}}
		{\infer[i]{\tjudg{}{\var}{\ty_i (= \tyctxp\initty_i)}}{}}
		& \tomachbttwo
		& \infer*{\infer{\tjudg{}{\la\var\ctxp{\var}} 
				{\arr{\mset{\ldots\ty_i\ldots}}\tytwo}}{}}
		{\infer[i]{\tjudg{}{\blue\var}{\tyctxp{\initty_{\downpt}}_i}}{}}=\state
		\\[8pt]
		\bisimtypes&&
		\\[8pt]
		\estate\statetwo=\dstate{ \la\var\ctxtwo_n\ctxholep{\var} }{ \ctx_\statetwo }{
			\underbrace{(\var,\la\var\ctxtwo_n,\tlog_n)\cons\etapeaux{\tyctx,1}\statetwo}_{=\etape\statetwo} }{ \elog\statetwo 
}
		&	\tomachbttwo & 
		\ustate{ \var}{ \ctx_{\statetwo}\ctxholep{\la\var\ctxtwo_n} }{ \etapeaux{\tyctx,1}\statetwo }{	\tlog_n 
\cons\elog\statetwo } = \statetwo_\l
		\end{array}\]
		
		About the tape of $\estate\statetwo$, note that $\etape\statetwo = 
\elpos{\statetwo^1_\focus}\cons\etapeaux{\tyctx,1}\statetwo$ where $\state^1_\focus$ is 
the first type test of $\statetwo$. We have to show that $\state^1_\focus$ exhausts on $\var$, so that 
$\elpos{\state^1_\focus} = (\var,\l\var.\ctxtwo_n,\tlog_n)$ for some $\tlog_n$.
		Note that $\state^1_\focus$ is $\tjudg{}{\red{\la\var\ctxp{\var}} }
			{\tyctxp{\initty}_{i\uppt},\arr{\mset{\ldots\ctxhole\ldots}}\tytwo}$. Note that 
$\state^1_\focus\tomachbttwo\,\tjudg{} 
			{\blue\var}{\tyctxp{\initty}_{i\downpt},\ctxhole} = \statethree$, where $\statethree$ focusses on the same 
judgement of $\state$, and that $\statethree$ is the state that S-exhausts $\state^1_\focus$. By definition of 
extraction, $\elpos{\state^1_\focus} = (\var,\l\var.\ctxtwo_n,\tlog_n)$ where $\tlog_n$ is the extraction of the first 
$n$ judgement tests of $\state$.  Then $\ctx_\state = \ctx_{\statetwo}\ctxholep{\la\var\ctxtwo_n}$ and $\elog\state = 
\tlog_n \cons\elog\statetwo$.

About the tape, for $\state$ we have to prove that $\etapeaux{\tyctx,1}\statetwo = \etape\state = 
\etapeaux{\tyctx,0}\state$. This is done as for $\tomachvar$, mimicking the reasoning in the proof of the S-exhaustible 
invariant (\reflemma{S-invariant-siam}).

Then, $\state_\l = \estate{\state}$, that is, $\state \bisimtypes \state_\l$.

		%%%%%%%%%%%%%%%%%%%%%%%%%%%%%%%%%
		\item Cases $\tomachdotthree$ and $\tomachdotfour$. They are identical to case 
		$\tomachdotone$.
		
		%%%%%%%%%%%%%%%%%%%%%%%%%%%%%%%%%
		\item Case $\tomacharg$.
		\[\begin{array}{clc}
		\statetwo=\infer{\tjudg{}{\tm\tmtwo}{\ty}} 
		{\tjudg{}{\blue\tm}{\arr{\mset{\ldots 
						\tyctxp{\initty_{\downpt}}_i\ldots}}{\ty}}
			& \tjudgi{}{\tmtwo}{\tytwo_i (=\tyctxp{\initty_{\downpt}}_i)}}
		& \tomacharg &
		\infer{\tjudg{}{\tm\tmtwo}{\ty}} 
		{\tjudg{}{\tm}{\arr{\mset{\ldots 
						\tytwo_i\ldots}}{\ty}}
			& \tjudgi{}{\red\tmtwo}{\tyctxp{\initty_{\uppt}}_i}}=\state
		\\[8pt]
		\bisimtypes&&
		\\[8pt]
		\estate\statetwo=\ustate{ \tm }{ \underbrace{\ctxtwop{\ctxhole\tmtwo}}_{=\ctx_{\statetwo}} }{ 
\underbrace{\elpos{\statetwo^1_\focus}\cdot\etapeaux{\tyctx,1}{\statetwo}}_{=\etape\statetwo} }{ \elog\statetwo } 
		& \tomacharg &
		\dstate{ \tmtwo }{ \ctxtwop{\tm\ctxhole} }{ \etapeaux{\tyctx,1}{\statetwo} }{ 
\elpos{\statetwo^1_\focus}\cdot\elog\statetwo } = \statetwo_\l
		\end{array}\]
		where $\statetwo^1_\focus$ is the first type test of $\statetwo$. Obviously, $\ctx_{\state} = \ctxtwop{\tm\ctxhole} 
$. For the log we have to show that $\elog\state$ is equal to $\elpos{\statetwo^1_\focus}\cdot\elog\statetwo $, which 
amounts to show that the first judgement test $\state^1$ of $\state$ exhausts on the same state as the first tape test 
$\statetwo^1_\focus$ of $\statetwo$. This is exactly the reasoning done in the proof of the S-exhaustible invariant. 
Similarly, one obtains that $\etapeaux{\tyctx,1}\statetwo = \etape\state = 
\etapeaux{\tyctx,0}\state$.

    %%%%%%%%%%%%%%%%%%%%%%%%%%%
		\item Case $\tomachbtone$.
		\[\begin{array}{clc}
		\statetwo=\infer{\tjudg{}{\tm\tmtwo}{\ty}} 
		{\tjudg{}{\tm}{\arr{\mset{\ldots 
						\tytwo_i\ldots}}{\ty}}
			& \tjudgi{}{\blue\tmtwo}{\tyctxp{\initty_{\downpt}}_i(=\tytwo_i)}}
		& \tomachbtone &
		\infer{\tjudg{}{\tm\tmtwo}{\ty}} 
		{\tjudg{}{\red\tm}{\arr{\mset{\ldots 
						\tyctxp{\initty_{\uppt}}_i\ldots}}{\ty}}
			& \tjudgi{}{\tmtwo}{\tytwo_i}}=\state
			\\[8pt]
			\bisimtypes&&
			\\[8pt]
		\estate\statetwo=\ustate{ \tmtwo }{ \underbrace{\ctxtwop{\tm\ctxhole}}_{=\ctx_{\statetwo} }}{ \etape\statetwo }{ 
\underbrace{\elpos{\state^1_\focus}\cdot\tlog}_{=\elog\statetwo} }
		& \tomachbtone &
		\dstate{ \tm }{ \ctxtwop{\ctxhole\tmtwo} }{ \elpos{\state^1_\focus}\cdot\etape\statetwo }{ \tlog } 
= \statetwo_\l
		\end{array}\]
  where $\statetwo^1_\focus$ is the first judgement test of $\statetwo$. Obviously, $\ctx_{\state} = 
\ctxtwop{\ctxhole\tmtwo} 
$. For the log, there is nothing to prove. For the tape, we have to show that $\etape\state$ is equal to 
$\elpos{\state^1_\focus}\cdot\etape\statetwo$, which 
amounts to show two things. First, that the first tape test $\state^1$ of $\state$ exhausts on the same state as 
the first judgement test 
$\statetwo^1_\focus$ of $\statetwo$. Second, that $\etapeaux{\tyctx,1}\state = \etape\statetwo = 
\etapeaux{\tyctx,0}\statetwo$. Both points follow exactly the reasoning done in the proof of the S-exhaustible 
invariant.   
  \end{itemize}
\end{proof}

\subsection{The \SIAM is acyclic and thus weights measure its time}
\label{ssect:acyclic-app}
First of all, we prove the abstract lemma that says that every state is reachable in a bi-deterministic transition 
system with only one initial state.

\begin{lemma}
	Let $\tsys{}$ be an acyclic bi-deterministic transition system on a finite set of states $\mathcal{S}$ and with only 
one 	initial state $\state_i$. Then all states in $\mathcal{S}$ are reachable from $\state_i$, and reachable only once.
\end{lemma}
\begin{proof}
	Let us consider a generic state $\state\in \mathcal{S}$ and show that it is reachable from $\state_i$. If 
$\state=\state_i$ we are done. 
	Otherwise, since the system is bi-deterministic we can go deterministically go backwards from $\state$. Since states 
are finite and there are no cycles, then the reduction sequence must end on an initial state, that is, on $\state_i$. 
Thus $\state$ is reachable from $\state_i$. If a state is reachable twice, then clearly there is a cycle, absurd.
\end{proof}

In order to prove that the \SIAM is acyclic, we need to show that if 
$\tm\towh\tmtwo$, then cycles are preserved between the sequence type 
derivation $\tyd$ for $\tm$ and the sequence type derivation $\tydtwo$ for 
$\tmtwo$. One way to show this fact is building a bisimulation between states 
of $\tyd$ and states of $\tydtwo$, since bisimulations preserve 
(non)termination. This idea has been already exploited 
by~\citet{IamPPDPtoAppear}, in order to prove the correctness of the \LIAM.

\paragraph{Weak Head Contexts} First of all, we need the notion of weak head context $\hctx$ defined as:
$$\hctx \defeq \ctxhole \mid \hctx \tm$$
Note that if $\tm \towh \tmtwo$ then $\tm = \hctxp{(\la\var\tmthree)\tmfour}$ and $\tmtwo = 
\hctxp{\tmthree\isub\var\tmfour}$.

\paragraph{Explaining the Bisimulation via a Diagram} Let us give an intuitive explanation of the bisimulation 
$\relf$.	Given two type derivations 
	$\tyd\pof\tjudg{}{\hctxp{(\la\var\tmthree)\tmfour}}{\initty}$ and 
	$\tydtwo\pof\tjudg{}{\hctxp{\tmthree\isub\var\tmfour}}{\initty}$, it is 
	possible 
	to define a relation $\relf$ between states of the former and of the latter 
	as depicted in the figure below. The key points are:
	\begin{enumerate}
	 \item each axiom for $\var$ in $\tyd$ is $\relf$-related with the 
judgement for the argument $w$ that replaces it in $\tydtwo$.
  \item Both the judgement for $r$ and the one for $(\la\var r)w$ are 
  $\relf$-related to $r\isub\var w$.
  \item The judgement for $\la\var r$ is not $\relf$-related to any judgement 
  of $\tydtwo$.

	\end{enumerate}
	
	\includegraphics[scale=0.9]{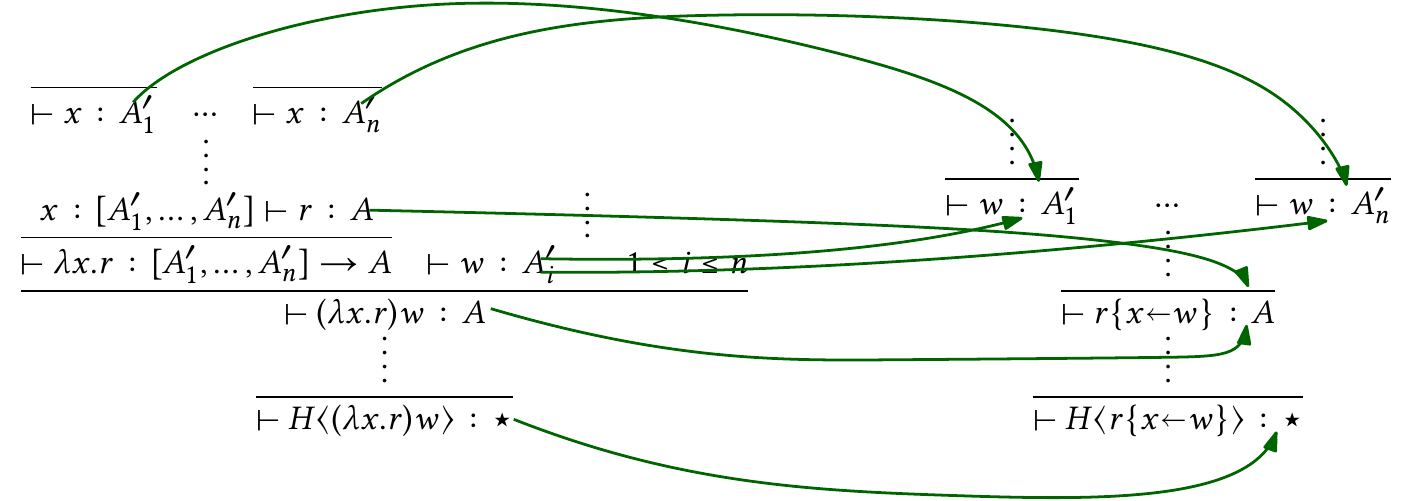}

	\paragraph{Defining $\relf$} In order to define $\relf$ formally, we enrich each type judgment (occurrence)
 $\tjudg{}{\tm}{\tyctxp\initty}$ with a context $\ctx$ such that $\ctxp\tm$ is 
the term in the final judgement of the derivation $\tyd$, obtaining 
$\tjudg{}{(\tm,\ctx)}{\tyctxp\initty}$. 
	
	\begin{definition}[Bisimulation $\relf$]
	The 
definition of $\relf$ for $\tjudg{}{(\tm,\ctx)}{\tyctxp\initty}$ has 4 clauses:

	\begin{itemize}
		\item 
		$\mathsf{rdx}$: the redex is in $\tm$, that is, $\tm = \hctxp{(\la\var\tmtwo)\tmthree}$, and so $\ctx$ is a head 
context $\hctxtwo$:
		$$\tjudg{}{(\hctxp{(\la\var\tmtwo)\tmthree},\hctxtwo)}{\tyctxp\initty} 
		\,\relfrdx
		\tjudg{}{(\hctxp{\tmtwo\isub\var\tmthree},\hctxtwo)}{\tyctxp\initty}$$
		
		\item $\mathsf{body}$: the term $\tm$ is part of the body of the abstraction involved in the redex:
		$$\tjudg{}{(\tm,\hctxp{(\la\var\ctxtwo)\tmtwo})}{\tyctxp\initty} 
		\,\relfbody 
		\tjudg{}{(\tm\isub\var\tmtwo,\hctxp{\ctxtwo\isub\var\tmtwo})}{\tyctxp\initty}$$
		
		\item $\mathsf{arg}$: the term $\tm$ is part of the argument of the redex:
		$$\tjudg{}{(\tm,\hctxp{(\la\var\ctxtwop\var)\ctxthree})}{\tyctxp\initty}
		\,\relfarg 
		\tjudg{}{(\tm,\hctxp{\ctxtwo\isub\var{\ctxthreep\tm}\ctxholep\ctxthree})}{\tyctxp\initty}$$
		
		\item $\mathsf{ext}$: The term $\tm$ is disjoint form the redex, that then takes place only in $\ctx$:
		$$\tjudg{}{(\tm,\hctxtwop{\hctxp{(\la\var\tmthree)\tmtwo}\ctxtwo})}{\tyctxp\initty}
		 \relfext\,
		\tjudg{}{(\tm,\hctxtwop{\hctxp{\tmthree\isub\var\tmtwo}\ctxtwo})}{\tyctxp\initty}$$
	\end{itemize}
\end{definition}
Please note that the only states of $\tyd$ which are not mapped to any state 
of 
$\tydtwo$ are those relative to the judgment 
$\tjudg{}{\la\var\tmthree}{\arr{\mset{\tytwo_1...\tytwo_n}}\ty}$.

\begin{proposition}
	$\relf$ is a loop-preserving bisimulation between \SIAM states.
\end{proposition}

% !TeX spellcheck = en_US
% !TEX root = ../main.tex
%%%%%%%%%%%%%%%%%%%%%%%%%%%%%%%%%%%%%%%%%%%%%%%%%%%%%%%%%%%%%%%%%%%%%
\begin{proof}\footnote{Also this proof requires colors.}
We inspect the 4 cases of the definition of $\relf$.
\begin{itemize}
 \item Rule $\mathsf{rdx}: \tjudg{}{(\hctxp{(\la\var\tm)\tmtwo},\hctxtwo)}{\tyctxp\initty} 
		\,\relfrdx
		\tjudg{}{(\hctxp{\tm\isub\var\tmtwo},\hctxtwo)}{\tyctxp\initty}$. Cases for $\uppt$ (by cases of $\hctx$):
	\begin{itemize}		
		\item $\hctx=\ctxhole$. The diagram is closed by rule $\mathsf{body}$:
		\[\small
		\begin{array}{ccccc}
		\tjudg{}{(\red{(\la\var\tm)\tmtwo},\hctxtwo)}{\tyctxp{\initty}} & \tosiam 
		& 
		\tjudg{}{(\red{\la\var\tm},\hctxtwop{\ctxhole\tmtwo})}{\arr\mty\tyctxp{\initty}}
		& \tosiam & 
		\tjudg{}{(\red{\tm},\hctxtwop{(\la\var\ctxhole)\tmtwo})}{\tyctxp{\initty}}\\[4pt]
		\relfrdx&&&&\relfbody\\[4pt]
		\tjudg{}{(\red{\tm\isub\var\tmtwo},\hctxtwo)}{\tyctxp{\initty}} &&=&& 
\tjudg{}{(\red{\tm\isub\var\tmtwo},\hctxtwo)}{\tyctxp{\initty}}
		\end{array}
		\]	
		
		\item $\hctx=\hctxthree\tmfive$. The diagram is closed by rule $\relfrdx$:
		\[
		\begin{array}{ccc}
		\tjudg{}{(\red{\hctxthreep\tmthree\tmfive},\hctxtwo)}{\tyctxp{\initty}} 
		& \tosiam & 
		\tjudg{}{(\red{\hctxthreep\tmthree},\hctxtwop{\ctxhole\tmfive})} 
		{\arr\mty{\tyctxp{\initty}}}\\[4pt]
		\relfrdx&&\relfrdx\\[4pt]
		\tjudg{}{(\red{\hctxthreep\tmfour\tmfive},\hctxtwo)}{\tyctxp{\initty}} 
		& 
		\tosiam & 
		\tjudg{}{(\red{\hctxthreep\tmfour},\hctxtwop{\ctxhole\tmfive})} 
		{\arr\mty\tyctxp{\initty}}
		\end{array}
		\]

		\end{itemize}
		
		Cases for $\downpt$ (by cases of $\hctxtwo$):
		\begin{itemize}
		\item $\hctxtwo=\ctxhole$. Both machines are stuck.
		\[
		\begin{array}{c}
		\tjudg{}{(\tmthree,\blue\ctxhole)}{\tyctxp{\initty}}\\[4pt]
		\relfrdx\\[4pt]
		\tjudg{}{(\tmfour,\blue\ctxhole)}{\tyctxp{\initty}}
		\end{array}
		\]
		\item $\hctxtwo=\hctxthreep{\ctxhole\tmfive}$. Two subcases depending 
		on the type context.
		If the focus is on the right of the arrow the diagram is closed by rule 
		$\mathsf{rdx}$.
		\[
		\begin{array}{ccc}
		\tjudg{}{(\tmthree,\blue{\hctxthreep{\ctxhole\tmfive}})} 
		{\arr\mty{\tyctxp{\initty}}} & \tosiam & 
		\tjudg{}{(\tmthree\tmfive,\blue{\hctxthree})} {{\tyctxp{\initty}}}
		\\[4pt]
		\relfrdx&&\relfrdx\\[4pt]
		\tjudg{}{(\tmfour,\blue{\hctxthreep{\ctxhole\tmfive}})} 
		{\arr\mty{\tyctxp{\initty}}} & \tosiam &
		\tjudg{}{(\tmfour\tmfive,\blue{\hctxthree})} {{\tyctxp{\initty}}}
		\end{array}
		\]
		If the focus is on the left of the arrow the diagram is closed by rule 
		$\mathsf{ext}$.
		\[
		\begin{array}{ccc}
		\tjudg{}{(\tmthree,\blue{\hctxthreep{\ctxhole\tmfive}})} 
		{\arr{\mset{\ldots\tyctxp{\initty}\ldots}}{\ty}} & \tosiam & 
		\tjudg{}{(\red\tmfive,\hctxthreep{\tmthree\ctxhole})} {{\tyctxp{\initty}}}
		\\[4pt]
		\relfrdx&&\relfext\\[4pt]
		\tjudg{}{(\tmfour,\blue{\hctxthreep{\ctxhole\tmfive}})} 
		{\arr{\mset{\ldots\tyctxp{\initty}\ldots}}{\ty}} & \tosiam & 
		\tjudg{}{(\red\tmfive,\hctxthreep{\tmfour\ctxhole})} {{\tyctxp{\initty}}}
		\end{array}
		\]
	\end{itemize}	
	
	%%%%%%%%%%%%%%%%%%%%%%%%%%%%
	\item Rule $\mathsf{body}$: $\tjudg{}{(\tm,\hctxp{(\la\var\ctxtwo)\tmtwo})}{\tyctxp\initty} 
		\,\relfbody 
		\tjudg{}{(\tm\isub\var\tmtwo,\hctxp{\ctxtwo\isub\var\tmtwo})}{\tyctxp\initty}$. Cases of $\uppt$ (by cases of 
$\tm$):
	\begin{itemize}
		\item $\tm =\tmthree\tmfour$. Trivially 
		closed by rule $\mathsf{body}$.
		\item  $\tm=\la\vartwo\tmthree$. If $\tm:\initty$ both machines are 
		stuck. If $\tm:\arr\mty\ty$, the diagram is trivially closed by rule 
		$\mathsf{body}$.
		\item $\tm=\var$. Diagram closed by rule $\mathsf{arg}$.
		\[
		\small\begin{array}{ccccl}
		\tjudg{}{(\red\var,\hctxp{(\la\var\ctxtwo)\tmtwo}}{\tyctxp{\initty}} 
		&\tosiam& 
		\tjudg{}{(\la\var\ctxtwop\var,\blue{\hctxp{\ctxhole\tmtwo}})} 
		{\arr{\mset{...\tyctxp{\initty}...}}\tytwo}{} 
		& \tosiam &
		\tjudg{}{(\red{\tmtwo},\hctxp{(\la\var\ctxtwop\var)\ctxhole})}
		{\tyctxp{\initty}}\\[4pt]
		\relfbody&&&&\relfarg\\[4pt]
		
\tjudg{}{(\red{\tmtwo},\hctxp{\ctxtwo\isub\var\tmtwo})}{\tyctxp{\initty}}&&=&&\tjudg{}{(\red{\tmtwo},\hctxp{
\ctxtwo\isub\var\tmtwo})}{\tyctxp{\initty}}
		\end{array}
		\]
		\end{itemize}
		
		Cases of $\downpt$ (by cases of $\ctxtwo$):
	\begin{itemize}
		\item $\ctxtwo=\ctxhole$. The diagram is closed by rule $\mathsf{rdx}$
		\[
		\small\begin{array}{ccccl}
		\tjudg{}{(\tm,\blue{\hctxp{(\la\var\ctxhole)\tmtwo})}}{\tyctxp\initty} 
		&\tosiam& 
		\tjudg{}{(\la\var\tm,\blue{\hctxp{\ctxhole\tmtwo}})} 
		{\arr\mty{\tyctxp\initty}}
		& \tosiam &
		\tjudg{}{((\la\var\tm)\tmtwo,\blue\hctx)}
		{\tyctxp{\initty}}\\[4pt]
		\relfbody&&&&\relfrdx\\[4pt]
		
		\tjudg{}{(\tm\isub\var\tmtwo,\blue{\hctx})}{\tyctxp\initty}&&=&& 
		\tjudg{}{(\tm\isub\var\tmtwo,\blue{\hctx})}{\tyctxp\initty}
		\end{array}
		\]
		
		\item $\ctxtwo=\ctxthreep{\la\vartwo\ctxhole}$, 
		$\ctxtwo=\ctxthreep{\ctxhole\tmthree}$ and 
		$\ctxtwo=\ctxthreep{\tmthree\ctxhole}$. The diagram is trivially 
		closed by rule $\mathsf{body}$.
	\end{itemize}
	
	\item Rule $\relfarg$: 
	$\tjudg{}{(\tm,\hctxp{(\la\var\ctxtwop\var)\ctxthree})}{\tyctxp\initty}
	\,\relfarg 
	\tjudg{}{(\tm,\hctxp{\ctxtwo\isub\var{\ctxthreep\tm}\ctxholep\ctxthree})}{\tyctxp\initty}$.
	Cases of $\uppt$ (by cases of $\tm$) are all trivial: they are closed by 
	rule $\relfarg$ itself. The only non trivial case for $\downpt$ (by cases 
	of 
	$\ctxthree$) is when $\ctxthree=\ctxhole$.
	\[
	\small\begin{array}{ccccl}
	\tjudg{}{(\tm,\blue{\hctxp{(\la\var\ctxtwop\var)\ctxhole}})}
	{\tyctxp{\initty}}
	&\tosiam& 
	\tjudg{}{(\red{\la\var\ctxtwop\var},\hctxp{\ctxhole\tm})} 
	{\arr{\mset{...\tyctxp{\initty}...}}\tytwo}{} 
	& \tosiam &
	\tjudg{}{(\var,\blue{\hctxp{(\la\var\ctxtwo)\tm}}}{\tyctxp{\initty}} \\[4pt]
	\relfarg&&&&\relfbody\\[4pt]
	\tjudg{}{(\tm,\blue{\hctxp{
			\ctxtwo\isub\var\tm}})}{\tyctxp{\initty}}&&=&&
	\tjudg{}{(\tm,\blue{\hctxp{\ctxtwo\isub\var\tm})}}{\tyctxp{\initty}}
	\end{array}
	\]
	
	\item Rule $\relfext$: 
	$\tjudg{}{(\tm,\hctxtwop{\hctxp{(\la\var\tmthree)\tmtwo}\ctxtwo})}{\tyctxp\initty}
	\relfext\,
	\tjudg{}{(\tm,\hctxtwop{\hctxp{\tmthree\isub\var\tmtwo}\ctxtwo})}{\tyctxp\initty}$.
	Cases of $\uppt$ (by cases of $\tm$) are all trivial: they are closed by 
	rule $\relfext$ itself. The only non trivial case for $\downpt$ (by cases 
	of 
	$\ctxtwo$) is when $\ctxtwo=\ctxhole$. We put $\tmfive\defeq 
	\hctxp{(\la\var\tmthree)\tmtwo}$ and $\tmfour\defeq 
	\hctxp{\tmthree\isub\var\tmtwo}$.
	\[
	\begin{array}{ccc}
	 \tjudg{}{(\tm,\blue{\hctxtwop{\tmfive\ctxhole}})} {{\tyctxp{\initty}}} & 
	 \tosiam & \tjudg{}{(\red\tmfive,{\hctxtwop{\ctxhole\tm}})} 
	 {\arr{\mset{\ldots\tyctxp{\initty}\ldots}}{\ty}}
	\\[4pt]
	\relfext&&\relfrdx\\[4pt]
	\tjudg{}{(\tm,\blue{\hctxtwop{\tmfour\ctxhole}})} {{\tyctxp{\initty}}} & 
	\tosiam & 
	\tjudg{}{(\red\tmfour,{\hctxtwop{\ctxhole\tm}})} 
	{\arr{\mset{\ldots\tyctxp{\initty}\ldots}}{\ty}}
	\end{array}
	\]
	\end{itemize}
\end{proof}

\begin{corollary}
	If $\tyd\pof\tjudg{}{\hctxp{(\la\var\tmthree)\tmfour}}{\initty}$ contains a 
	cycle, the also 
	$\tydtwo\pof\tjudg{}{\hctxp{\tmthree\isub\var\tmfour}}{\initty}$ contains a 
	cycle.
\end{corollary}
\begin{proof}
	If the run of the \SIAM on $\tyd\pof\tjudg{}{\hctxp{(\la\var\tmthree)\tmfour}}{\initty}$ loops then there exists a 
state $\state_\tyd$ such that a computation 
	starting from $\state_\tyd$ diverges. Every state but 
	$\tjudg{}{(\la\var\tmthree,\hctxp{\ctxhole\tmfour})}{\tyctxp\initty}$, 
	which however is not final, is 
	related by $\relf$ to a state $\state_{\tydtwo}$ of $\tsys\tydtwo$. By 
	preservation of (non)termination (see \cite{IamPPDPtoAppear}), also  
$\state_{\tydtwo}$ diverges. Since 
	$\state_{\tydtwo}$ has a finite number of states, there must be a cycle.
\end{proof}
\begin{corollary}
	For each type derivation $\tyd\pof\tjudg{}{\tm}{\initty}$, $\tsys\tyd$ has 
	no cycles.
\end{corollary}
\begin{proof}
	Since $\tm$ is typable, then it has normal form, call it $\tmtwo$. Clearly 
	the type derivation for $\tmtwo$ has no cycles. By the previous corollary, 
	also $\tyd$ cannot have any of them.
\end{proof}

\begin{theorem}
\label{thm:iamtimetypes-app}
	For every closed term $\tm$, the IAM takes $n$ steps in $\tm$ iff
	$\WeightTimeIAM{\tyd}=n$ for every $\tyd\pof\tjudg{}{\tm}{\initty}$.
\end{theorem}
\begin{proof}
	Every state of $\tsys{\tyd}$ is traversed exactly once, during a 
	computation that starts from the initial state. Thus the length of 
	the 
	computation is the cardinality of the states of $\tsys{\tyd}$. Since a 
	state in a type judgment $\tjudg{\tye}{\tm}{\ty}$ occurring in $\tyd$ is 
	given by an occurrence 
	of $\initty$ in $\ty$, then for every judgment  the 
	number of associated states is $\size{\ty}$. Then, it is immediate to 
	note that the 
	number of states in a type derivation ending in 
	$\tyd\pof\wtjudg{}{n}{\tm}{\initty}$ is exactly $n$.
\end{proof}

% !TeX spellcheck = en_US
% !TEX root = main.tex

\section{Duplication Example}\label{ap:dup}

We provide an example that illustrates how duplication is handled by the 
different abstract machines. We
consider the $\lambda$-term $\tm:=(\la\var\var\var)(\la\vartwo\vartwo)$.
\subsection{The \LIAM}
The first steps of the computation are 
needed to reach the head variable, namely $\var$.
\[{\footnotesize
	\begin{array}{l|c|c|c|c|c}
	&\mathsf{Sub}\mbox{-}\mathsf{term} & \mathsf{Context} & \mathsf{\Log} & 
	\mathsf{Tape} & \mathsf{Dir}
	\\
	\cline{1-6}
	&\ndstatetab{(\la x xx)(\la y y)} {\ctxhole} {\epsilon} {\epsilon} \downp\\
	\iamdap&\ndstatetab{\la x xx} {\ctxhole(\la y y)} {\resm} {\epsilon} 
	\downp\\
	\iamdlamone&\ndstatetab{xx} {(\la x \ctxhole)(\la y y)} {\epsilon} 
	{\epsilon} \downp\\
	\iamdap&\ndstatetab{x} {(\la x \ctxhole x)(\la y y)} {\resm} {\epsilon} 
	\downp
	\end{array}}
\]
Once the head variable $\var$ has been found, the machine switches to upward 
mode $\upp$ in order to find its argument $\la\vartwo\vartwo$.
\[{\footnotesize
	\begin{array}{l|c|c|c|c|c}
	&\mathsf{Sub}\mbox{-}\mathsf{term} & \mathsf{Context} & \mathsf{\Log} & 
	\mathsf{Tape} & \mathsf{Dir}
	\\
	\cline{1-6}
	&\ndstatetab{x} {(\la x \ctxhole x)(\la y y)} {\resm} 
	{\epsilon}{\downp} \\
	\iamdvar&\nustatetab{\la x xx} {\ctxhole(\la y y)} 
	{(x,\la x \ctxhole x,\epsilon)\cdot\resm} {\epsilon}{\upp} \\
	\iamuaplone&\ndstatetab{\la y y} {(\la x xx)\ctxhole} {\resm} 
	{(x,\la x \ctxhole x,\epsilon)}{\downp} \\
	\end{array}}
\]
Intuitively, the first occurrence of $\var$ has been substituted for 
$\la\vartwo\vartwo$, thus forming a new virtual $\beta$-redex 
$(\la\vartwo\vartwo)x$. Indeed, a $\resm$ is on top of the tape, thus allowing 
the \LIAM to inspect $\la\vartwo\vartwo$, reaching its head variable $\vartwo$.
\[{\footnotesize
	\begin{array}{l|c|c|c|c|c}
	&\mathsf{Sub}\mbox{-}\mathsf{term} & \mathsf{Context} & \mathsf{\Log} & 
	\mathsf{Tape} & \mathsf{Dir}
	\\
	\cline{1-6}
	&\ndstatetab{\la y y} {(\la x xx)\ctxhole} {\resm} 
	{(x,\la x \ctxhole x,\epsilon)}\downp \\
	\iamdlamone&\ndstatetab{y} {(\la x xx)(\la y \ctxhole)} {\epsilon} 
	{(x,\la x \ctxhole x,\epsilon)}\downp \\
	\iamdvar&\nustatetab{\la y y} {(\la x xx)\ctxhole} 
	{(y,\la\vartwo\ctxhole,\epsilon)} {(x,\la x \ctxhole x,\epsilon)}\upp \\
	\end{array}}
\]
Once the head variable $y$ has been found, the machine, in upward mode $\upp$, 
starts looking for the argument of $y$ from its binder $\la\vartwo\vartwo$. 
However, $\la\vartwo\vartwo$ was not the left side of an application forming a 
$\beta$-redex. Indeed, it was virtually substituted for the first occurrence of 
$\var$, in the log, thus creating the virtual redex $(\la\vartwo\vartwo)x$. Its 
argument is thus the second occurrence of $\var$. The \LIAM is able to 
retrieve it, walking again the path towards the variable $\la\vartwo\vartwo$ 
has been virtually substituted for, namely the first occurrence of $x$, saved 
in the log. This is what we call backtracking.
\[{\footnotesize
	\begin{array}{l|c|c|c|c|c}
	&\mathsf{Sub}\mbox{-}\mathsf{term} & \mathsf{Context} & \mathsf{\Log} & 
	\mathsf{Tape} & \mathsf{Dir}
	\\
	\cline{1-6}
	&\nustatetab{\la y y} {(\la x xx)\ctxhole} 
	{(y,\la\vartwo\ctxhole,\epsilon)} {(x,\la x \ctxhole x,\epsilon)}\upp \\
	\iamuapr&\ndstatetab{\la x xx} {\ctxhole(\la y y)} 
	{(x,\la x \ctxhole x,\epsilon)\cdot(y,\la\vartwo\ctxhole,\epsilon)} 
	{\epsilon}\downp \\
	\iamdlamtwo&\nustatetab{x} {(\la x \ctxhole x)(\la y y)} 
	{(y,\la\vartwo\ctxhole,\epsilon)} {\epsilon}\upp \\
	\iamuaplone&\ndstatetab{x} {(\la x x\ctxhole)(\la y y)} {\epsilon} 
	{(y,\la\vartwo\ctxhole,\epsilon)}\downp \\
	\end{array}}
\]
Notice that we are able to backtrack because we saved the occurrence of the 
substituted variable in the token, otherwise the machine would not be able to 
know  which occurrence of $x$ is the right one. Of course, when the first 
occurrence of $x$ is reached the $\IAM$, now again in upward mode $\upp$, finds 
immediately its argument, that is the second occurrence of $x$. At this point 
the machine looks for the argument of this last occurrence of $x$, finding, of 
course, again $\la\vartwo\vartwo$.
\[{\footnotesize
	\begin{array}{l|c|c|c|c|c}
	&\mathsf{Sub}\mbox{-}\mathsf{term} & \mathsf{Context} & \mathsf{\Log} & 
	\mathsf{Tape} & \mathsf{Dir}
	\\
	\cline{1-6}
	&\ndstatetab{x} {(\la x x\ctxhole)(\la y y)} {\epsilon} 
	{(y,\la\vartwo\ctxhole,\epsilon)}\downp \\
	\iamdvar&\nustatetab{\la x xx} {\ctxhole(\la y y)} 
	{(x,\la x x\ctxhole,(y,\la\vartwo\ctxhole,\epsilon))} {\epsilon}\upp \\
	\iamuaplone&\ndstatetab{\la y y} {(\la x xx)\ctxhole} {\epsilon} 
	{(x,\la x x\ctxhole,(y,\la\vartwo\ctxhole,\epsilon))}\downp\\
	\end{array}}
\]
The computation then stops, signaling that $\tm$ has weak head normal form. 
Please notice that the position on the log has now a nested structure. Indeed 
it carries information about the virtual substitutions already performed.

\subsection{The \LJAM}
The execution on the \LJAM is very similar to the \LIAM one. Logged positions 
now save the whole term and log.

\[{\footnotesize
	\begin{array}{l|c|c|c|c|c}
		&\mathsf{Sub}\mbox{-}\mathsf{term} & \mathsf{Context} & \mathsf{\Log} & 
		\mathsf{Tape} & \mathsf{Dir}
		\\
		\cline{1-6}
		&\ndstatetab{(\la x xx)(\la y y)} {\ctxhole} {\epsilon} {\epsilon} 
		\downp\\
		\iamdap&\ndstatetab{\la x xx} {\ctxhole(\la y y)} {\resm} {\epsilon} 
		\downp\\
		\iamdlamone&\ndstatetab{xx} {(\la x \ctxhole)(\la y y)} {\epsilon} 
		{\epsilon} \downp\\
		\iamdap&\ndstatetab{x} {(\la x \ctxhole x)(\la y y)} {\resm} {\epsilon} 
		\downp\\
		\iamdvar&\nustatetab{\la x xx} {\ctxhole(\la y y)} 
		{(x,(\la x \ctxhole x)(\la\vartwo\vartwo),\epsilon)\cdot\resm} 
		{\epsilon}{\upp} \\
\end{array}}
\]

The execution continues as in the \LIAM case. We put $\lpos_\var\defeq(x,(\la x 
\ctxhole x)(\la\vartwo\vartwo),\epsilon)$.

\[{\footnotesize
	\begin{array}{l|c|c|c|c|c}
		&\mathsf{Sub}\mbox{-}\mathsf{term} & \mathsf{Context} & \mathsf{\Log} & 
		\mathsf{Tape} & \mathsf{Dir}
		\\
		\cline{1-6}
		&\nustatetab{\la x xx} {\ctxhole(\la y y)} 
		{\lpos_\var\cdot\resm} 
		{\epsilon}{\upp} \\
		\iamuaplone&\ndstatetab{\la y y} {(\la x xx)\ctxhole} {\resm} 
		{\lpos_\var}\downp \\
		\iamdlamone&\ndstatetab{y} {(\la x xx)(\la y \ctxhole)} {\epsilon} 
		{\lpos_\var}\downp \\
		\iamdvar&\nustatetab{\la y y} {(\la x xx)\ctxhole} 
		{(y,(\la\var\var)(\la\vartwo\ctxhole),\lpos_\var)} {\lpos_\var}\upp \\
\end{array}}
\]

In this situation the \LIAM would backtrack. The \LJAM, instead, directly jumps 
to the previously saved logged position.

\[{\footnotesize
	\begin{array}{l|c|c|c|c|c}
		&\mathsf{Sub}\mbox{-}\mathsf{term} & \mathsf{Context} & \mathsf{\Log} & 
		\mathsf{Tape} & \mathsf{Dir}
		\\
		\cline{1-6}
		&\nustatetab{\la y y} {(\la x xx)\ctxhole} 
		{(y,(\la\var\var)(\la\vartwo\ctxhole),\lpos_\var)} {\lpos_\var}\upp \\
		\iamujump&\nustatetab{x} {(\la x \ctxhole x)(\la y y)} 
		{(y,(\la\var\var)(\la\vartwo\ctxhole),\lpos_\var)} {\epsilon}\upp \\
\end{array}}
\]

The computation then resumes as is the \LIAM case. We put $\lpos_\vartwo\defeq 
(y,(\la\var\var)(\la\vartwo\ctxhole),\lpos_\var)$.

\[{\footnotesize
	\begin{array}{l|c|c|c|c|c}
		&\mathsf{Sub}\mbox{-}\mathsf{term} & \mathsf{Context} & \mathsf{\Log} & 
		\mathsf{Tape} & \mathsf{Dir}
		\\
		\cline{1-6}
		&\nustatetab{x} {(\la x \ctxhole x)(\la y y)} 
		{\lpos_\vartwo} {\epsilon}\upp \\
		\iamuaplone&\ndstatetab{x} {(\la x x\ctxhole)(\la y y)} {\epsilon} 
		{\lpos_\vartwo}\downp \\
		\iamdvar&\nustatetab{\la x xx} {\ctxhole(\la y y)} 
		{(x,(\la x x\ctxhole)(\la\vartwo\vartwo),\lpos_\vartwo)} {\epsilon}\upp 
		\\
		\iamuaplone&\ndstatetab{\la y y} {(\la x xx)\ctxhole} {\epsilon} 
		{(x,(\la x x\ctxhole)(\la\vartwo\vartwo),\lpos_\vartwo)}\downp\\
\end{array}}
\]

\subsection{The KAM}
The \KAM, as the other machines, looks for the head variable of the term, inspecting its spine. Every time an application is encountered, the argument is saved in the stack, together with its environment. When abstractions are encountered, an entry in the environment is created, linking the abstracted variable with the closure on the top of the stack.
\[{\footnotesize
	\begin{array}{l|c|c|c|c}
		&\mathsf{Sub}\mbox{-}\mathsf{term} & \mathsf{Context} & \mathsf{Env.} & 
		\mathsf{Stack}
		\\
		\cline{1-5}
		&\dstatetab{(\la x xx)(\la y y)} {\ctxhole} {\epsilon} {\epsilon} \\
		\kamdap&\dstatetab{\la x xx} {\ctxhole(\la y y)} 
		{(\la\vartwo\vartwo,(\la\var\var\var)\ctxhole,\epsilon)} {\epsilon} \\
		\kamdlam&\dstatetab{xx} {(\la x \ctxhole)(\la y y)} {\epsilon} 
		{\esub\var{(\la\vartwo\vartwo,(\la\var\var\var)\ctxhole,\epsilon)}=:\env}\\
		\kamdap&\dstatetab{x} {(\la x \ctxhole x)(\la y y)} {(\var,(\la x x 
		\ctxhole)(\la y y),\env)} {\env} 
		\\
\end{array}}
\]

This way, when a variable is encountered, the \KAM can hop directly to the sub-term it would be substituted for, just by inspecting the environment. Moreover, the right environment for the sub-term can be restored from its closure.

\[{\footnotesize
	\begin{array}{l|c|c|c|c}
		&\mathsf{Sub}\mbox{-}\mathsf{term} & \mathsf{Context} & \mathsf{Env.} & 
		\mathsf{Stack}
		\\
		\cline{1-5}
		&\dstatetab{x} {(\la x \ctxhole x)(\la y y)} {(\var,(\la x x 
			\ctxhole)(\la y y),\env)} {\env}\\
		\iamdvar&\dstatetab{\la y y} {(\la x xx)\ctxhole} {(\var,(\la x x 
			\ctxhole)(\la y y),\env)} 
		{\stempty}\\
		\kamdlam&\dstatetab{y} {(\la x xx)(\la y \ctxhole)} {\epsilon} 
		{\esub\vartwo{(\var,(\la x x \ctxhole)(\la y y),\env)}}\\
		\iamdvar&\dstatetab{\var}{(\la x x \ctxhole)(\la y y)}{\stempty}{\env}\\
		\iamdvar&\dstatetab{\la y 
		y}{(\la\var\var\var)\ctxhole}{\stempty}{\stempty}
\end{array}}
\]

\subsection{The \LPAM}
The \LPAM starts the computation looking for the head variable, namely $\var$, traversing the spine of$\tm$. When it has been found, the machine, now in $\upp$ mode, turns to query the argument, saving the variable position on the tape.
\[{\footnotesize
	\begin{array}{l|c|c|c|c|c|c}
		&\mathsf{Sub}\mbox{-}\mathsf{term} & \mathsf{Context} & \mathsf{Hist.} 
		& \mathsf{Index} & \mathsf{Tape} & \mathsf{Dir}
		\\
		\cline{1-7}
		&\pamstatedtab{(\la x xx)(\la y y)} {\ctxhole} {\epsilon} {0} 
		{\epsilon} &\downp\\
		\iamdap&\pamstatedtab{\la x xx} {\ctxhole(\la y y)} {\epsilon} {0} 
		{\resm} &\downp\\
		\iamdlamone&\pamstatedtab{xx} {(\la x \ctxhole)(\la y y)} {\epsilon} {0}
		{\epsilon} &\downp\\
		\iamdap&\pamstatedtab{x} {(\la x \ctxhole x)(\la y y)} 
		{\epsilon} {0} {\resm} &\downp\\
		\iamdvar&\pamstateutab{\la x xx} {\ctxhole(\la y y)} {\epsilon} {0} 
		{(x,(\la x \ctxhole x)(\la\vartwo\vartwo))\cdot\resm}&{\upp} \\
\end{array}}
\]

When the argument, namely $\la\vartwo\vartwo$ has been found, the variable gets saved in the history permanently, together with the index needed to reconstruct the chain of virtual substitutions.

\[{\footnotesize
	\begin{array}{l|c|c|c|c|c|c}
		&\mathsf{Sub}\mbox{-}\mathsf{term} & \mathsf{Context} & \mathsf{Hist.} 
		& \mathsf{Index} & \mathsf{Tape} & \mathsf{Dir}
		\\
		\cline{1-7}
		&\pamstateutab{\la x xx} {\ctxhole(\la y y)} 
		{\epsilon}{0}{\pos_\var\cdot\resm} &{\upp} \\
		\iamuaplone&\pamstatedtab{\la y y} {(\la x xx)\ctxhole} 
		{(\pos_\var,0)}{1}{\resm} &\downp \\
		\iamdlamone&\pamstatedtab{y} {(\la x xx)(\la y \ctxhole)}  
		{(\pos_\var,0)}{1}{\epsilon}&\downp \\
		\iamdvar&\pamstateutab{\la y y} {(\la x xx)\ctxhole} 
		 {(\pos_\var,0)}{1}{(y,(\la\var\var)(\la\vartwo\ctxhole))}&\upp 
		 \\
\end{array}}
\]

Jumps are handled by retrieving the position in the history at the current index, that is then decreased by one in order to coherently update the information associated to substitutions.

\[{\footnotesize
	\begin{array}{l|c|c|c|c|c|c}
		&\mathsf{Sub}\mbox{-}\mathsf{term} & \mathsf{Context} & \mathsf{Hist.} 
		& \mathsf{Index} & \mathsf{Tape} & \mathsf{Dir}
		\\
		\cline{1-7}
		&\pamstateutab{\la y y} {(\la x xx)\ctxhole} 
		 {(\pos_\var,0)}{1}{(y,(\la\var\var)(\la\vartwo\ctxhole))}&\upp \\
		\iamujump&\pamstateutab{x} {(\la x \ctxhole x)(\la y y)} 
		 {(\pos_\var,0)}{0}{(y,(\la\var\var)(\la\vartwo\ctxhole))}&\upp \\
\end{array}}
\]

In transition $\iamdvar$, if the level of the binding context is greater than zero, one has to follow the chain of pointers in order to set the correct information about substitutions. 

\[{\footnotesize
	\begin{array}{l|c|c|c|c|c|c}
		&\mathsf{Sub}\mbox{-}\mathsf{term} & \mathsf{Context} & \mathsf{Hist.} 
		& \mathsf{Index} & \mathsf{Tape} & \mathsf{Dir}
		\\
		\cline{1-7}
		&\pamstateutab{x} {(\la x \ctxhole x)(\la y y)} 
		 {(\pos_\var,0)}{0}{\pos_\vartwo}&\upp \\
		\iamuaplone&\pamstatedtab{x} {(\la x x\ctxhole)(\la y y)}  
		{(\pos_\vartwo,0)\cons(\pos_\var,0)}{2}{\epsilon}&\downp \\
		\iamdvar&\pamstateutab{\la x xx} {\ctxhole(\la y y)} 
		 {(\pos_\vartwo,0)\cons(\pos_\var,0)}{0}{(x,(\la x 
		 x\ctxhole)(\la\vartwo\vartwo))}&\upp 
		\\
		\iamuaplone&\pamstatedtab{\la y y} {(\la x xx)\ctxhole}  
		{(\pos_{\var_2},0)\cons(\pos_\vartwo,0)\cons(\pos_\var,0)}{3} 
		{\epsilon}&\downp\\
\end{array}}
\]

In general, the way in which the \LPAM works is not very intuitive. The simplest way to understand it, is to think about the correspondence with the \LJAM. Given a log $\tlog_n\defeq(\var,\ctx,\tlogtwo) \cons\lpos_{2}\cdots\lpos_n$, a history $\history$ and index $i$, if $i=|\history|$, then the \LPAM is considering the log $\tlog_n$. Otherwise diminishing $i$ by one and considering the position in the history corresponding to the index $i$ amounts to consider the log $\tlogtwo$. Finally, following the chain of pointers $k$ times, \ie considering as index $\phi_\history^k(i)$ amounts to consider the log $\lpos_{k+1}\cdots\lpos_{n}$.

\subsection{The \SIAM}
As we have already proved, the \SIAM is strongly bisimilar to the \LIAM. One can indeed observe that the sequence of states reached by the \SIAM is the same sequence of the \LIAM. The \SIAM does not need additional data structures, since every sub-term has been already duplicated in advanced as many times as needed. The argument $\la\vartwo\vartwo$, for example, is typed twice, since it is substituted for \emph{two} occurrences of $\var$ during the evaluation of $\tm$. 

\[
\infer{\tjudg{}{(\la\var\var\var)(\la\vartwo\vartwo)}{\initty_{\uppt{\red{1}}}}}{
	\infer{\tjudg{}{\la\var\var}{\arr{\mset{\initty_{\downpt{\blue 
	{12}}},\arr{\mset{\initty_{\uppt{\red
	 {9}}}}}{\initty_{\downpt{\blue
	 {5}}}}}}{\initty_{\uppt{\red
	 {2}}}}}}{
		\infer{\tjudg{\var:\mset{\initty,\arr{\mset{\initty}}{\initty}}}{\var\var}{\initty_{\uppt{\red
		 {3}}}}}{
			\infer{\tjudg{\var:\mset{\arr{\mset{\initty}}{\initty}}}{\var}{\arr{\mset{\initty_{\downpt{\blue
			 {10}}}}}{\initty_{\uppt{\red
			 {4}}}}}}{}
			 & \infer{\tjudg{\var:\mset\initty}{\var}{\initty_{\uppt{\red 
			 {11}}}}}{}}} &
		\infer{\tjudg{}{\la\vartwo\vartwo}{\arr{\mset{\initty_{\downpt{\blue 
		{8}}}}}{\initty_{\uppt{\red
		 {6}}}}}}{
			\infer{\tjudg{\vartwo:\initty}{\vartwo}{\initty_{\uppt{\red 
			{7}}}}}{}} &
		\infer{\tjudg{}{\la\vartwo\vartwo}{\initty_{\uppt{\red {13}}}}}{}}
\]

%\section{Appendix}

\end{document}